\documentclass[final]{siamltex}

\usepackage[]{amsmath,amssymb,epsfig}

\usepackage{subfigure}
\usepackage{amsmath}
\usepackage{amssymb}
\usepackage{graphicx}
\usepackage{comment}
\usepackage{array}
\usepackage{algorithm}
\usepackage{algorithmic}
\usepackage{url}
\usepackage{color}

\newcommand{\mb}[1]{\mbox{\boldmath$#1$}}
\newcounter{ale}

\newcommand{\abc}{\item[\alph{ale})]\stepcounter{ale}}

\newenvironment{liste}{\begin{itemize}}{\end{itemize}}
\newcommand{\aliste}{\begin{liste} \setcounter{ale}{1}}
\newcommand{\zliste}{\end{liste}}

\newenvironment{abcliste}{\aliste}{\zliste}

\begin{document}

\title{Eigenvector Synchronization, Graph Rigidity and the Molecule Problem}
\author{
Mihai Cucuringu%
\thanks{PACM, Princeton University, Fine Hall, Washington Road, Princeton NJ 08544-1000 USA, email: mcucurin@math.princeton.edu}
\and
Amit~Singer%
\thanks{Department of Mathematics and PACM, Princeton University, Fine Hall, Washington Road, Princeton NJ 08544-1000 USA, email: amits@math.princeton.edu}
\and
David Cowburn%
\thanks{Department of Biochemistry, Albert Einstein College of Medicine of Yeshiva University, 1300 Morris Park Ave, Bronx, NY 10461, USA, email: david.cowburn@einstein.yu.edu
}
}

\date{\today}
\maketitle

\begin{abstract}
The graph realization problem has received a great deal of attention in recent years, due to its importance in applications such as wireless sensor networks and structural biology. In this paper, we extend on previous work and propose the 3D-ASAP algorithm, for the graph realization problem in $\mathbb{R}^3$, given a sparse and noisy set of distance measurements. 3D-ASAP is a divide and conquer, non-incremental and non-iterative algorithm,  which integrates local distance information into a global structure determination. Our approach starts with identifying, for every node, a subgraph of its 1-hop neighborhood graph, which can be accurately embedded in its own coordinate system. In the noise-free case, the computed coordinates of the sensors in each patch must agree with their global positioning up to some unknown rigid motion, that is, up to translation, rotation and possibly reflection. In other words, to every patch there corresponds an element of the Euclidean group Euc(3) of rigid transformations in $\mathbb{R}^3$, and the goal is to estimate the group elements that will properly align all the patches in a globally consistent way. Furthermore, 3D-ASAP successfully incorporates information  specific to the molecule problem in structural biology, in particular information on known  substructures and their orientation. In addition, we also propose 3D-SP-ASAP, a faster version of 3D-ASAP, which uses a spectral partitioning algorithm as a preprocessing step for dividing the initial graph into smaller subgraphs. Our extensive numerical simulations show that 3D-ASAP and 3D-SP-ASAP are very robust to high levels of noise in the measured distances and to sparse connectivity in the measurement graph, and compare favorably to similar state-of-the art localization algorithms.
\end{abstract}

\begin{keywords}
Graph realization, distance geometry, eigenvectors, synchronization, spectral graph theory, rigidity theory, SDP, the molecule problem, divide-and-conquer.
\end{keywords}


\begin{AMS} 15A18, 49M27, 90C06, 90C20, 90C22, 92-08, 92E10
\end{AMS}

\section{Introduction}
\label{intro}


In the \textit{graph realization problem}, one is given a graph $G=(V,E)$ consisting of a set of $|V|=n$ nodes and $|E|=m$ edges, together with a non-negative distance measurement $d_{ij}$ associated with each edge, and is asked to compute a realization of $G$ in the Euclidean space $\mathbb{R}^d$ for a given dimension $d$. In other words, for any pair of adjacent nodes $i$ and $j$, $(i,j) \in E$, the distance $d_{ij} = d_{ji}$ is available, and the goal is to find a $d$-dimensional embedding $p_1, p_2, \ldots, p_n \in \mathbb{R}^d$ such that $\|p_i-p_j \| =d_{ij}, \text{ for all } (i,j) \in E$.

Due to its practical significance, the graph realization problem has attracted a lot of attention in recent years, across many communities. The problem and its variants come up naturally in a variety of settings such as wireless sensor networks \cite{biswas_stress_sdp,overview}, structural biology \cite{molecule_problem}, dimensionality reduction, Euclidean ball packing and multidimensional scaling (MDS) \cite{cox}. In such real world applications, the given distances $d_{ij}$ between adjacent nodes are not accurate, $d_{ij} = \|p_i - p_j \| + \varepsilon_{ij}$ where $\varepsilon_{ij}$ represents the added noise, and the goal is to find an embedding that realizes all known distances $d_{ij}$ as best as possible.

When all $n(n-1)/2$ pairwise distances are known, a $d$-dimensional embedding of the complete graph can be computed using classical MDS. However, when many of the distance constraints are missing, the problem becomes significantly more challenging because the rank-$d$ constraint on the solution is not convex. Applying a rigid transformation (composition of rotation, translation and possibly reflection) to a graph realization results in another graph realization, because rigid transformations preserve distances. Whenever an embedding exists, it is unique (up to rigid transformations) only if there are enough distance constraints, in which case the graph is said to be globally rigid 
(see, e.g., \cite{conditions_UGR}). The graph realization problem is known to be difficult; Saxe has shown it is strongly NP-complete in one dimension, and strongly NP-hard for higher dimensions \cite{saxe,yemini}. Despite its difficulty, the graph realization problem has received a great deal of attention in the networking and distributed computation communities, and numerous heuristic algorithms exist that approximate its solution. In the context of sensor networks \cite{Ji2004,Yang2004,Yang2006,Yang2007}, there are many algorithms that solve the graph realization problem, and they include methods such as global optimization \cite{smacof}, semidefinite programming (SDP) \cite{BiswasYe,biswas_stress_sdp,biswas,So07,SY05,univRig} and local to global approaches \cite{moore,ShangRuml,PATCHWORK,LRE,ARAP}.

A popular model for the graph realization problem is that of a geometric graph model, where the distance between two nodes is available if and only if they are within sensing radius $\rho$ of each other, i.e., $(i,j)\in E \iff d_{ij} \leq \rho$. The graph realization problem is NP-hard also under the geometric graph model \cite{Yang2006}. Figure \ref{fig:intro_Example} shows an example of a measurement graph for a data set of $n=500$ nodes, with sensing radius $\rho=0.092$ and average degree $deg=18$, i.e. each node knows, on average, the distance to its $18$ closest neighbors.
\begin{figure}[h]
\begin{center}
\vspace{0mm}
\includegraphics[width=0.49\columnwidth]{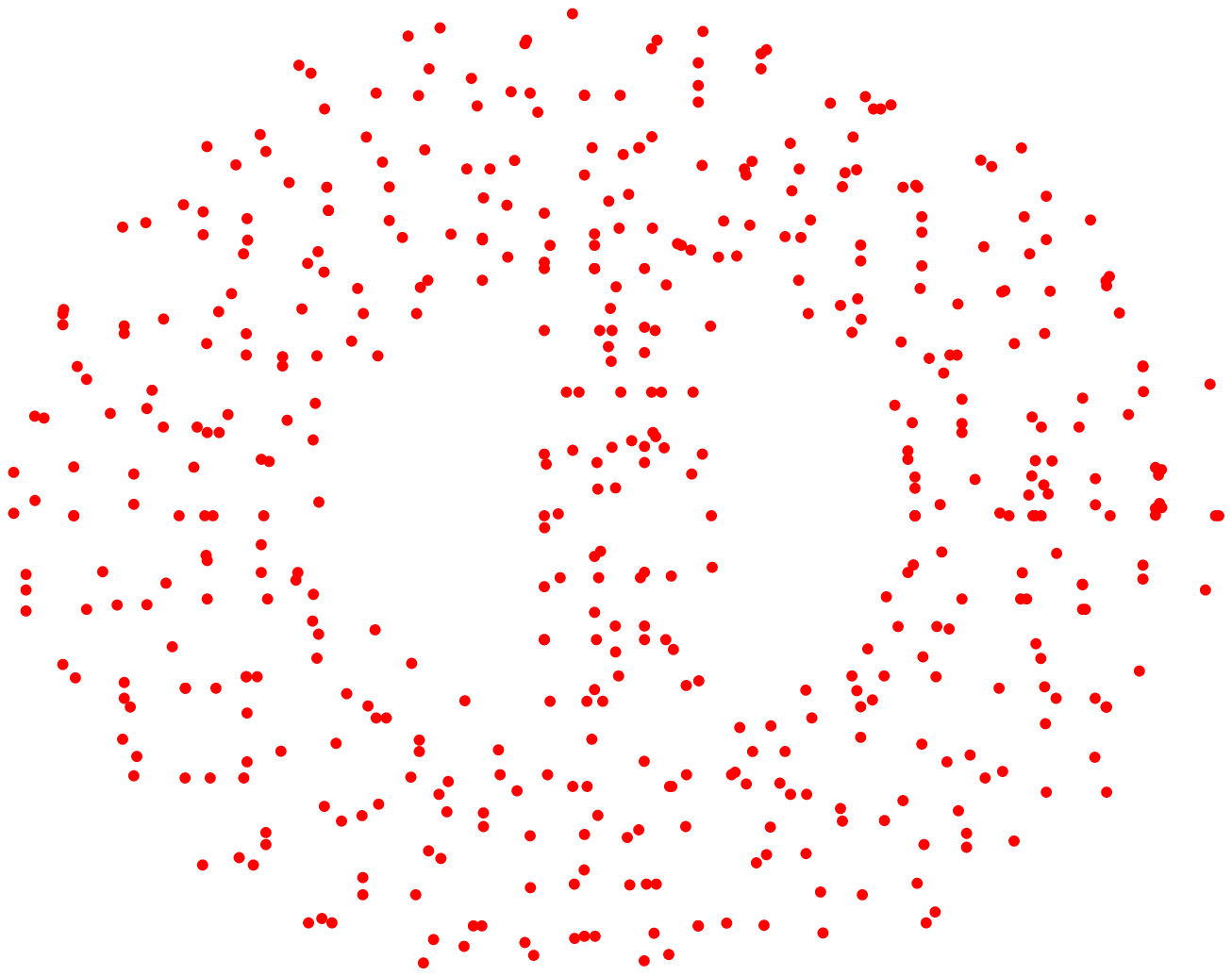}
\includegraphics[width=0.49\columnwidth]{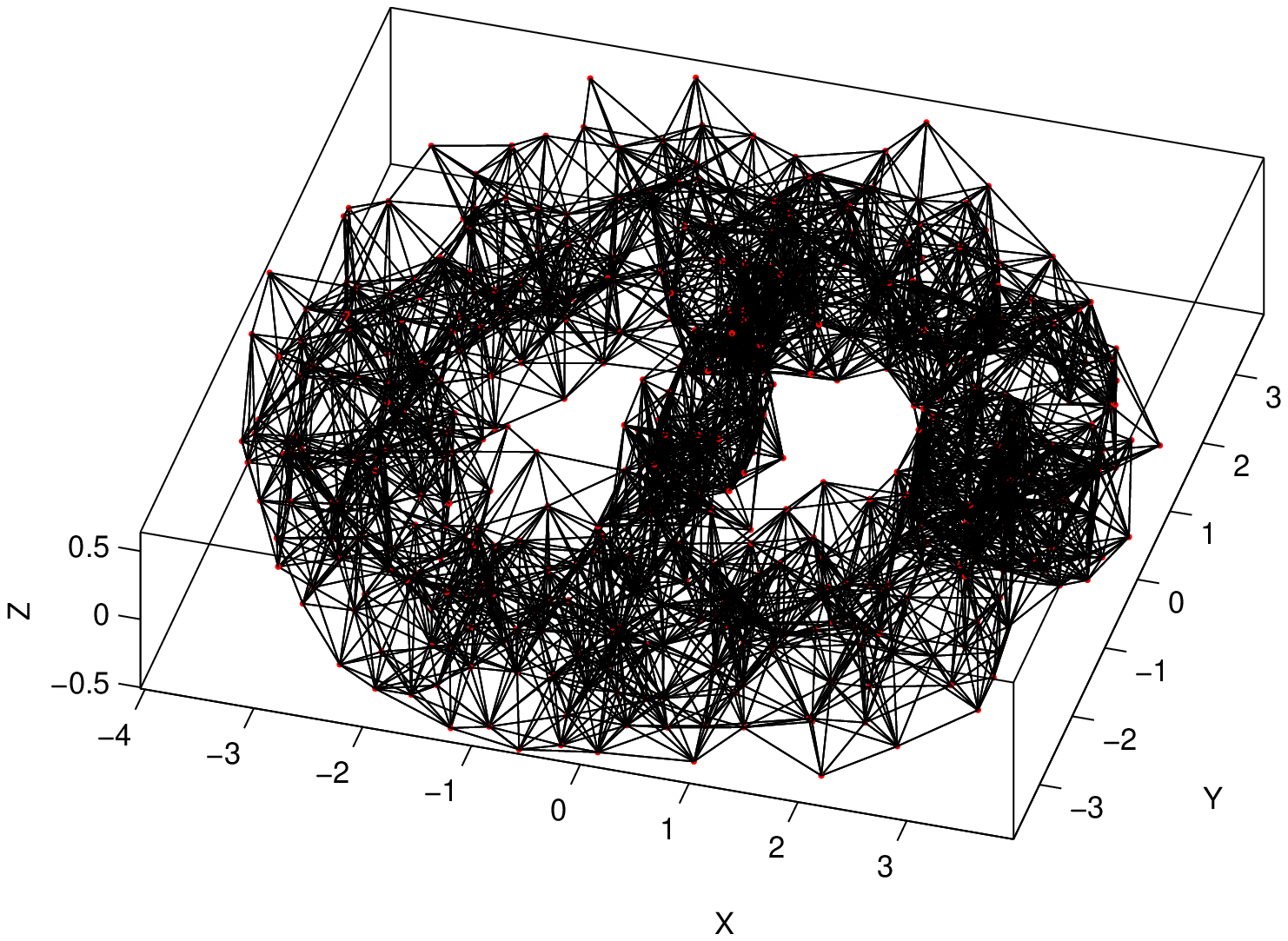}
\vspace{-6mm}
\end{center}
\caption{2D view of the BRIDGE-DONUT data set, a cloud of $n=500$ points in $\mathbb{R}^3$ in the shape of a donut (left), and the associated measurement geometric graph of radius $\rho = 0.92 $ and average degree $deg=18$ (right).}
\label{fig:intro_Example}
\end{figure}

The graph realization problem in $\mathbb{R}^3$ is of particular importance because it arises naturally in the application of nuclear magnetic resonance (NMR) to structural biology. NMR spectroscopy is a well established modality for atomic structure determination, especially for relatively small proteins (i.e., with atomic mass less than 40kDa) \cite{DC1}, and contributes to progress in structural genomics \cite{DC2}. General properties of proteins such as bond lengths and angles can be translated into accurate distance constraints. In addition, peaks in the NOESY experiments are used to infer spatial proximity information between pairs of nearby hydrogen atoms, typically in the range of 1.8 to 6 angstroms. The intensity of such NOESY peaks is approximately proportional to the distance to the minus $6^{th}$ power, and it is thus used to infer distance information between pairs of hydrogen atoms (NOEs). Unfortunately, NOEs provide only a rough estimate of the true distance, and hence the need for robust algorithms that are able to provide accurate reconstructions even at high levels of noise in the NOE data. In addition, the experimental data often contains potential constraints that are ambiguous, because of signal overlap resulting in incomplete assignment \cite{aria}.
The structure calculation based on the entire set of distance constraints, both accurate measurements and NOE measurements, can be thought of as instance of the graph realization problem in $\mathbb{R}^3$  with noisy data.

In this paper, we focus on the molecular distance geometry problem (to which we will refer from now on as the molecule problem) in $\mathbb{R}^3$, although the approach is applicable to other dimensions as well. In \cite{ASAP}  we introduced 2D-ASAP, an eigenvector based synchronization algorithm that solves the sensor network localization problem in $\mathbb{R}^2$. We summarize below the approach used in 2D-ASAP, and further explain the differences and improvements that its generalization to three-dimensions brings. Figure \ref{fig:pipeline} shows a schematic overview of our algorithm, which we call 3D-As-Synchronized-As-Possible (3D-ASAP).

The 2D-ASAP algorithm proposed in \cite{ASAP} belongs to the group of algorithms that integrate local distance information into a global structure determination. For every sensor, we first identify  globally rigid subgraphs of its $1$-hop neighborhood that we call patches. Each patch is then separately localized in a coordinate system of its own using either the stress minimization approach of \cite{gotsman}, or by using semidefinite programming (SDP). In the noise-free case, the computed coordinates of the sensors in each patch must agree with their global positioning up to some unknown rigid motion, that is, up to translation, rotation and possibly reflection. To every patch there corresponds an element of the Euclidean group Euc(2) of rigid transformations in the plane, and the goal is to estimate the group elements that will properly align all the patches in a globally consistent way. By finding the optimal alignment of all pairs of patches whose intersection is large enough, we obtain measurements for the ratios of the unknown group elements. Finding group elements from noisy measurements of their  ratios is also known as the synchronization problem \cite{timeSync2,timeSync1}. For example, the synchronization of clocks in a distributed network from noisy measurements of their time offsets is a particular example of synchronization over $\mathbb{R}$. \cite{sync} introduced an eigenvector method for solving the synchronization problem over the group SO(2) of planar rotations. This algorithm serves as the basic building block for our 2D-ASAP and 3D-ASAP algorithms. Namely, we reduce the graph realization problem to three consecutive synchronization problems that overall solve the synchronization problem over Euc(2). Intuitively, we use the eigenvector method for the compact part of the group (reflections and rotations), and use the least-squares method for the non-compact part (translations).
In the first step, we solve a synchronization problem over $\mathbb{Z}_2$ for the possible reflections of the patches using the eigenvector method. In the second step, we solve a synchronization problem over SO(2) for the rotations,  also using the eigenvector method. Finally, in the third step, we solve a synchronization problem over $\mathbb{R}^2$ for the translations by solving an overdetermined linear system of equations using the method of least squares. This solution yields the estimated coordinates of all the sensors up to a global rigid transformation. Note that the groups $\mathbb{Z}_2$ and SO(2) are compact, allowing us to use the eigenvector method, while the group $\mathbb{R}^2$ is non-compact and requires a different synchronization method such as least squares.

In the present paper, we extend on the approach used in 2D-ASAP to accommodate for the additional challenges posed by rigidity theory in $\mathbb{R}^3$ and other constraints that are specific to the molecule problem. In addition, we also increase the robustness to noise and speed of the algorithm. The following paragraphs are a brief summary of the improvements 3D-ASAP brings, in the order in which they appear in the algorithm.

First, we address the issue of using a divide and conquer approach from the perspective of three dimensional global rigidity, i.e., of decomposing the initial measurement graph into many small overlapping patches that can be uniquely localized. Sufficient and necessary conditions for two dimensional combinatorial global rigidity have been established only recently, and motivated our approach for building patches in 2D-ASAP \cite{JacksonJordan,conditions_UGR}. Due to the recent coning result in rigidity theory \cite{coning}, it is also possible to extract globally rigid patches in dimension three. However, such globally rigid patches cannot always be localized accurately by SDP algorithms, even in the case of noiseless data. To that end, we rely on the notion of \textit{unique localizability} \cite{SY05} to localize noiseless graphs, and introduce the notion of a weakly uniquely localizable (WUL) graph, in the case of noisy data.

Second, we use a median-based denoising algorithm in the preprocessing step, that overall produces more accurate patch localizations. Our approach is based on the observation that a given edge may belong to several different patches, the localization of each of which may result in a different estimation for the distance. The median of these different estimators from the different patches is a more accurate estimator of the underlying distance.

Third, we incorporated in 3D-ASAP the possibility to integrate prior available information. As it is often the case in real applications (such as NMR), one has readily available structural information on various parts of the network that we are trying to localize. For example, in the NMR application, there are often subsets of atoms (referred to as ``molecular fragments", by analogy with the fragment molecular orbital approach, e.g., \cite{fedorov}) whose relative coordinates are known a priori, and thus it is desirable to be able to incorporate such information in the reconstruction process. Of course, one may always input into the problem all pairwise distances within the molecular fragments.
However, this requires increased computational efforts while still not taking full advantage of the available information, i.e., the orientation of the molecular fragment. Nodes that are aware of their location are often referred to as anchors, and anchor-based algorithms make use of their existence when computing the coordinates of the remaining sensors. Since in some applications the presence of anchors is not a realistic assumption, it is important to have efficient and noise-robust anchor-free algorithms, that can also incorporate the location of anchors if provided. However, note that having molecular fragments is not the same as having anchors; given a set of (possibly overlapping) molecular fragments, no two of which can be joined in a globally rigid body, only one molecular fragment can be treated as anchor information (the nodes of that molecular fragment will be the anchors), as we do not know a priori  how the individual molecular fragments relate to each other in the same global coordinate system.

Fourth, we allow for the possibility of combining the first two steps (computing the reflections and rotations) into one single step, thus doing synchronization over the group of orthogonal transformations O(3) $=\mathbb{Z}_2 \times$  SO(3) rather than individually over $\mathbb{Z}_2$ followed by SO(3). However, depending on the problem being considered and the type of available information, one may choose not to combine the above two steps. For example, when molecular fragments are  present, we first do synchronization over $\mathbb{Z}_2$ with anchors, as detailed in Section \ref{sync_anchors}, followed by synchronization over SO(3).

Fifth, we incorporate another median-based heuristic in the final step, where we compute the translations of each patch by solving, using least squares, three overdetermined linear systems, one for each of the $x,y$ and $z$ axis. For a given axis, the displacement between a pair of nodes appears in multiple patches, each resulting in a different estimation of the displacement along that axis. The median of all these different estimators from different patches provides a more accurate estimator for the displacement. In addition, after the least squares step, we introduce a simple heuristic that corrects the scaling of the noisy distance measurements. Due to the geometric graph model assumption and the uniform noise model, the distance measurements taken as input by 3D-ASAP are significantly scaled down, and the least squares step further shrinks the distances between nodes in the initial reconstruction.

Finally, we introduce 3D-SP-ASAP, a variant of 3D-ASAP which uses a spectral partitioning algorithm in the pre-processing step of building the patches. This approach is somewhat similar to the recently proposed DISCO algorithm of \cite{DISCO}. The philosophy behind DISCO is to recursively divide large problems into two smaller problems, thus building a binary tree of subproblems, which can ultimately be solved by the traditional SDP-based localization methods. 3D-ASAP has the disadvantage of generating a number of smaller subproblems (patches) that is linear in the size of the network, and localizing all resulting patches is a computationally expensive task, which is exactly the issue addressed by 3D-SP-ASAP.

\begin{figure}[ht]
\begin{center}
\vspace{0mm}
\includegraphics[width=1.0\columnwidth]{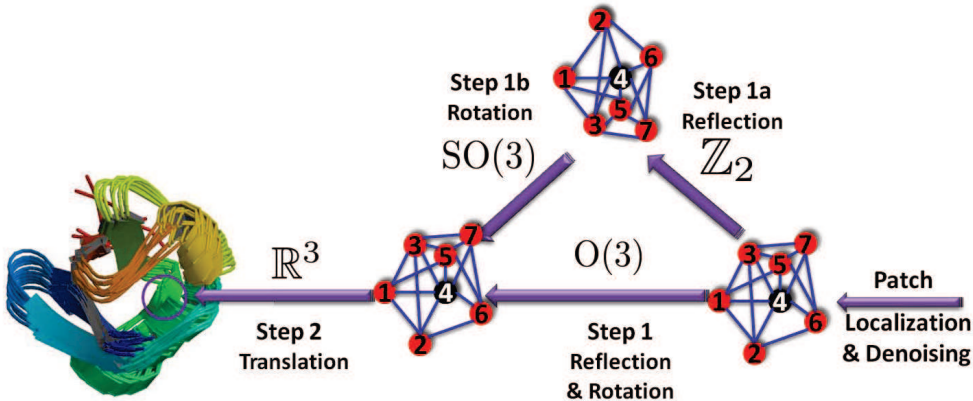}
\end{center}
\vspace{-2mm}
\caption{The 3D-ASAP recovery process for a patch in the 1d3z molecule graph. The localization of the rightmost subgraph in its own local frame is obtained from the available accurate bond lengths and noisy NOE measurements by using one of the SDP algorithms. To every patch, like the one shown here, there corresponds an element of Euc(3) that we try to estimate. Using the pairwise alignments, in Step 1 we estimate both the reflection and rotation matrix from an eigenvector synchronization computation over O(3), while in Step 2 we find the estimated coordinates by solving an overdetermined system of linear equations. If there is available information on the reflection or rotations of some patches, one may choose to further divide Step 1 into two consecutive steps. Step 1a is synchronization over $\mathbb{Z}_2$, while Step 1b is synchronization over SO(3), in which the missing reflections and rotations are estimated.}
\label{fig:pipeline}
\end{figure}

From a computational point of view, all steps of the algorithm can be implemented in a distributed fashion and scale linearly in the size of the network, except for the eigenvector computation, which is nearly-linear\footnote{Every iteration of the power method or the Lanczos algorithm that are used to compute the top eigenvectors is linear in the number of edges of the graph, but the number of iterations is greater than $O(1)$ as it depends on the spectral gap.}.
We show the results of numerous numerical experiments that demonstrate the robustness of our algorithm to noise and various topologies of the measurement graph.

This paper is organized as follows:
Section \ref{relwork} is a brief survey of related approaches for solving the graph realization problem in $\mathbb{R}^3$.
Section \ref{asap} gives an overview of the 3D-ASAP algorithm we propose.
Section \ref{WULS} introduces the notion of weakly uniquely localizable graphs used for breaking up the initial large network into patches, and explains the process of embedding and aligning the patches.
Section \ref{sec:sp-asap} proposes a variant of the 3D-ASAP algorithm by using a spectral clustering algorithm as a preprocessing step in breaking up the measurement graph into patches.
In Section \ref{MDA} we introduce a novel median-based denoising technique that improves the localization of individual patches, as well as a heuristic that corrects the shrinkage of the distance measurements.
Section \ref{sync_anchors} gives an analysis of different approaches to the synchronization problem over $\mathbb{Z}_2$ with anchor information, which is useful for incorporating molecular fragment information when estimating the reflections of the remaining patches.
In Section \ref{numexp}, we detail the results of numerical simulations in which we tested the performance of our algorithms in comparison to existing state-of-the-art algorithms.
Finally, Section \ref{conclusion} is a summary and a discussion of possible extensions of the algorithm and its  usefulness in other applications.

\section{Related work} \label{relwork}
Due to the importance of the graph realization problem, many heuristic strategies and numerical algorithms have been proposed in the last decade.
A popular approach to solving the graph realization problem is based on SDP, and has attracted considerable attention in recent years \cite{BiswasYe,sdpAngle,biswas_stress_sdp,biswas,univRig}. We defer to Section \ref{sec:sub_emb} a description of existing SDP relaxations of the graph realization problem.  Such SDP techniques are usually able to localize accurately small to medium sized problems (up to a couple thousands atoms).
However, many protein molecules have more than 10,000 atoms and the SDP approach by itself is no longer computationally feasible due to its increased running time.
In addition, the performance of the SDP methods is significantly affected by the number and location of the anchors, and the amount of noise in the data. To overcome the computational challenges posed by the limitations of the SDP solvers, several divide and conquer approaches have been  proposed recently for the graph realization problem. One of the earlier methods appears in \cite{dist_SDP}, and more recent methods include the Distributed Anchor Free Graph Localization (DAFGL) algorithm of  \cite{DAFGL}, and the DISCO algorithm (DIStributed COnformation) of \cite{DISCO}.

One of the critical assumptions required by the distributed SDP algorithm in \cite{dist_SDP} is that there exist anchor nodes distributed uniformly throughout the physical space. The algorithm relies on the anchor nodes to divide the sensors into clusters, and solves each cluster separately via an SDP relaxation. Combining smaller subproblems together can be a challenging task, however this is not an issue if there exist anchors within each smaller subproblem (as it happens in the sensor network localization problem) because the solution to the clusters induces a global positioning; in other words the alignment process is trivially solved by the existence of anchors within the smaller subproblems. Unfortunately, for the molecule problem, anchor information is scarce, almost inexistent, hence it becomes crucial to develop algorithms that are amenable to a distributed implementation (to allow for solving large scale problems) despite there being no anchor information available. The DAFGL algorithm of \cite{DAFGL} attempted to overcome this difficulty and was successfully applied to molecular conformations, where anchors are inexistent. However, the performance of DAFGL was significantly affected by the sparsity of the measurement graph, and the algorithm could tolerate only up to 5\% multiplicative noise in the distance measurements.

The recent DISCO algorithm of \cite{DISCO} addressed some of the shortcomings of DAFGL, and
used a similar divide-and-conquer approach to successfully reconstruct conformations of very large molecules. At each step, DISCO checks whether the current subproblem is small enough to be solved by itself, and if so, solves it via SDP and further improves the reconstruction by gradient descent. Otherwise, the current subproblem (subgraph) is further divided into two subgraphs, each of which is then solved recursively. To combine two subgraphs into one larger subgraph, DISCO aligns the two overlapping smaller subgraphs, and refines the coordinates by applying gradient descent.
In general, a divide-and-conquer algorithm  consists of two ingredients: dividing a bigger problem into smaller subproblems, and combining the solutions of the smaller subproblems into a solution for a larger subproblem. With respect to the former aspect, DISCO minimizes the number of edges between the two subgroups (since such edges are not taken into account when localizing the two smaller subgroups), while maximizing the number of edges within subgroups, since denser graphs are easier to localize both in terms of speed and robustness to noise. As for the latter aspect, DISCO divides a group of atoms in such a way that the two resulting subgroups have many overlapping atoms. Whenever the common subgroup of atoms is accurately localized, the two subgroups can be further joined together in a robust manner. DISCO employs several heuristics that determine when the overlapping atoms are accurately localized, and whether there are atoms that cannot be localized in a given instance (they do not attach to a given subgraph in a globally rigid way). Furthermore, in terms of robustness to noise, DISCO compared favorably to the above-mentioned divide-and-conquer algorithms.

Finally, another graph realization algorithm amenable to large scale problems is maximum variance unfolding (MVU), a non-linear dimensionality reduction technique proposed by \cite{fastMVU}. MVU produces a low-dimensional representation of the data by maximizing the variance of its embedding while preserving the original local distance constraints. MVU builds on the SDP approach and addresses the issue of the possibly high dimensional solution to the SDP problem. While rank constraints are non convex and cannot be directly imposed, it has been observed that low dimensional solutions emerge naturally when maximizing the variance of the embedding (also known as the maximum trace heuristic). Their main observation is that the coordinate vectors of the sensors are often well approximated by just the first few (e.g., 10) low-oscillatory eigenvectors of the graph Laplacian. This observation allows to replace the original and possibly large scale SDP by a much smaller SDP which leads to a significant reduction in running time.

While there exist many other localization algorithms, we provide here two other such references. One of the more recent iterative algorithms that was observed to perform well in practice compared to other traditional optimization methods is a variant of the gradient descent approach called the stress majorization algorithm, also known as SMACOF \cite{smacof}, originally introduced by \cite{smacof0}. The main drawback of this approach is that its objective function (commonly referred to as the stress) is not convex and the search for the global minimum is prone to getting stuck at local minima, which often makes the initial guess for gradient descent based algorithms important for obtaining satisfactory results. DILAND, recently introduced in \cite{diland}, is a distributed algorithm for localization with noisy distance measurements. Under appropriate conditions on the connectivity and triangulation of the network, DILAND was shown to converge almost surely to the true solution.

\section{The 3D-ASAP Algorithm}
\label{asap}

3D-ASAP is a divide and conquer algorithm that breaks up the large graph into many
smaller overlapping subgraphs, that we call patches, and ``stitches" them together consistently in a global coordinate system with the purpose of localizing the entire measurement graph. Unlike previous graph localization algorithms, we build patches that are globally rigid\footnote{There are several different notions of rigidity that appear in the literature, such as local and global rigidity, and the more recent notions of universal rigidity and unique localizability \cite{univRig,SY05}.} (or weakly uniquely localizable, that we define later),
in order to avoid foldovers in the final solution\footnote{We remark that in the geometric graph model, the non-edges also provide distance information since $(i,j)\notin E$ implies $d_{ij} > \rho$. This information sometimes allows to uniquely localize networks that are not globally rigid to begin with. However, we do not use this information in the standard formulation of our algorithm, but this could be further incorporated to enhance the reconstruction of very sparse networks.}.
We also assume that the given measurement graph is globally rigid to begin with; otherwise the algorithm  will discard the parts of the graph that do not attach globally rigid to the rest of the graph.

We build the patches in the following way. For every node $i$ we denote by $V(i) = \{ j : (i,j) \in E \} \cup \{i\}$ the set of its neighbors together with the node itself, and by $G(i) = (V(i), E(i))$ its subgraph of 1-hop neighbors.
If $G(i)$ is globally rigid, which can be checked efficiently using the randomized algorithm of \cite{GortlerGR}, then we embed it in $\mathbb{R}^3$.
Using SDP for globally rigid (sub)graphs can still produce incorrect localizations, even for noiseless  data. In order to ensure that SDP would give the correct localization, a stronger notion of rigidity is needed, that of unique localizability \cite{SY05}. However, in practical applications the distance measurements are noisy, so we introduce the notion of weakly localizable subgraphs, and use it to build patches that can be accurately localized.
The exact way we break up the 1-hop neighborhood subgraphs into smaller globally rigid or weakly uniquely localizable subgraphs is detailed in Section \ref{sec:sub_wuls}. In Section \ref{sec:sp-asap} we describe an alternative method for decomposing the measurement graph into patches, using a spectral partitioning algorithm. We denote by $N$ the number of patches obtained in the above decomposition of the measurement graph, and note that it may be different from $n$, the number of nodes in $G$, since the neighborhood graph of a node may contribute several patches or none. Also, note that the embedding of every patch in $\mathbb{R}^3$ is given in its own local frame. To compute such an embedding, we use the following SDP-based algorithms: FULL-SDP for noiseless data \cite{BiswasYe}, and SNL-SDP  for noisy data \cite{snlsdp}. Once each patch is embedded in its own coordinate system, one must find the reflections, rotations and translations that will stitch all patches together in a consistent manner, a process to which we refer as \emph{synchronization}.

We denote the resulting patches by $P_1,P_2,\ldots,P_N$. To every patch $P_i$ there corresponds an element $e_i \in $ Euc(3), where Euc(3) is the Euclidean group of rigid motions in $\mathbb{R}^3$. The rigid motion $e_i$ moves patch $P_i$ to its correct position with respect to the global coordinate system. Our goal is to estimate the rigid motions $e_1,\ldots,e_N$ (up to a global rigid motion) that will properly align all the patches in a globally consistent way. To achieve this goal, we first estimate the alignment between any pair of patches $P_i$ and $P_j$ that have enough nodes in common, a procedure we detail later in Section \ref{sec:sub_align}. The alignment of patches $P_i$ and $P_j$ provides a (perhaps noisy) measurement for the ratio $e_i e_j^{-1}$ in Euc(3).  We solve the resulting synchronization problem in a globally consistent manner, such that information from local alignments propagates to pairs of non-overlapping patches. This is done by replacing the synchronization problem over Euc(3) by two different consecutive synchronization problems.

In the first synchronization problem, we simultaneously find the reflections and rotations of all the patches using the eigenvector synchronization algorithm over the group O(3) of orthogonal matrices. When prior information on the reflections of some patches is available, one may choose to replace this first step by two consecutive synchronization problems, i.e., first estimate the missing rotations by doing synchronization over $\mathbb{Z}_2$ with molecular fragment information, as described in Section \ref{sync_anchors}, followed by another synchronization problem over SO(3) to estimate the rotations of all patches. Once both reflections and rotations are estimated, we estimate the translations by solving an overdetermined linear system. Taken as a whole, the algorithm integrates all the available local information into a global coordinate system over several steps by using the eigenvector synchronization algorithm and least squares over the isometries of the Euclidean space. The main advantage of the eigenvector method is that it can recover the reflections and rotations even if many of the pairwise alignments are incorrect. The algorithm is summarized in Table \ref{overview}.

\begin{table}[h]
\begin{center}
\footnotesize
\begin{tabular}{| p{0.15\columnwidth} | p{0.85\columnwidth} |}
\hline
INPUT & $G=(V,E), \; |V|=n, \; |E|=m, \; d_{ij}$ for $(i,j) \in E$ \\
\hline
\hline
Pre-processing Step &
1. Break the measurement graph $G$ into $N$ globally rigid or weakly uniquely localizable patches $P_1,\ldots,P_N$.\\
Patch Localization &
2. Embed each patch $P_i$ separately using either FULL-SDP (for noiseless data), or SNL-SDP (for noisy data), or cMDS (for complete patches).
\\
\hline
Step 1       & 1. Align all pairs of patches $(P_i, P_j)$ that have enough nodes in common. \\
             & 2. Estimate their relative rotation and possibly reflection  $h_{ij} \in O(3)$.\\
             & 3. Build a sparse $3N\times 3N$ symmetric matrix $H=(h_{ij})$ as defined in (\ref{H_ij_def}).\\
Estimating  Reflections
             & 4. Define $\mathcal{H} = D^{-1} H$, where $D$ is a diagonal matrix with \newline
              $D_{3i-2,3i-2}=D_{3i-1,3i-1}=D_{3i,3i} = deg(i)$, for $i = 1,\ldots,N$. \\
and Rotations  & 5. Compute the top 3 eigenvectors $v_i^\mathcal{H}$ of $\mathcal{H}$
                    satisfying $ \mathcal{H} v_i^\mathcal{H}= \lambda_i^\mathcal{H} v_i^\mathcal{H}, i=1,2,3$. \\
             & 6. Estimate the global reflection and rotation of patch $P_i$ by the orthogonal matrix $\hat{h}_i $ that is closest to $\widetilde{h}_i $ in Frobenius norm, where $\widetilde{h}_i $ is the submatrix corresponding to the i$^\text{th}$ patch in the $3N \times 3$ matrix formed by the top three eigenvectors $[v_1^\mathcal{H} v_2^\mathcal{H} v_3^\mathcal{H}]$. \\
             & 7. Update the current embedding of patch $P_i$ by applying the orthogonal transformation $\hat{h}_i $ obtained above (rotation and possibly reflection)\\
\hline
Step 2 &
1. Build the $m \times n$ overdetermined system of linear equations given in (\ref{lsstep3_3d}), after applying the median-based denoising heuristic.\\
Estimating &
2. Include the anchors information (if available) into the linear system.\\
Translations &
3. Compute the least squares solution for the $x$-axis, $y$-axis and $z$-axis coordinates.\\
\hline
\hline
OUTPUT & Estimated coordinates $\hat{p}_1,\ldots,\hat{p}_n$ \\
\hline
\end{tabular}
\end{center}
\caption{Overview of the 3D-ASAP Algorithm}
\label{overview}
\end{table}

\subsection{Step 1: Synchronization over O(3) to estimate reflections and rotations}

As mentioned earlier, for every patch $P_i$ that was already embedded in its local frame, we need to estimate whether or not it needs to be reflected  with respect to the global coordinate system, and what is the rotation that aligns it in the same coordinate system. In 2D-ASAP, we first estimated the reflections, and based on that, we further estimated the rotations. However, it makes sense to ask whether one can combine the two steps, and perhaps further increase the robustness to noise of the algorithm. By doing this, information contained in the pairwise rotation matrices helps in better estimating the reflections, and vice-versa, information on the pairwise reflection between patches helps in improving the estimated rotations. Combining these two steps also reduces the computational effort by half, since we need to run the eigenvector synchronization algorithm only once.

We denote the orthogonal transformation of patch $P_i$ by $h_i \in$ O(3), which is defined up to a global orthogonal rotation and reflection. The alignment of every pair of patches $P_i$ and $P_j$ whose intersection is sufficiently large, provides a measurement $h_{ij}$ (a $3\times3$ orthogonal matrix) for the ratio $h_i h_j^{-1}$. However, some ratio measurements can be corrupted because of errors in the embedding of the patches due to noise in the measured distances. We denote by $G^P = (V^P, E^P)$ the patch graph whose vertices $V^P$ are the patches $P_1,\ldots,P_N$, and two patches $P_i$ and $P_j$ are adjacent, $(P_i, P_j) \in E^P$, iff they have enough vertices in common to be aligned such that the ratio $h_i h_j^{-1}$ can be estimated. We let $A^P$ denote the adjacency matrix of the patch graph, i.e., $A^P_{ij}=1$ if  $(P_i, P_j) \in E^P$, and $A^P_{ij}=0$ otherwise. Obviously, two patches that are far apart and have no common nodes cannot be aligned, and there must be enough\footnote{E.g., four common vertices, although the precise  definition of ``enough" will be discussed later.} overlapping nodes to make the alignment possible. Figures \ref{fig:patch_info_1} and \ref{fig:patch_info_2} show a typical example of the sizes of the patches we consider, as well as their intersection sizes.

The first step of 3D-ASAP is to estimate the appropriate rotation and reflection of each patch. To that end, we use the eigenvector synchronization method as it was shown to perform well even in the presence of a large number of errors. The eigenvector method starts off by building the following $3N \times 3N$ sparse symmetric matrix $H=(h_{ij})$, where
$h_{ij}$ is the a $3\times3$ orthogonal matrix that aligns patches $P_i$ and $P_j$
\begin{equation}
  H_{ij} = \left\{
     \begin{array}{rl}
  h_{ij} & \;\; (i,j) \in E^P  \text{ ($P_i$ and $P_j$ have enough points in common)} \\
  O_{3 \times 3} & \;\; (i,j) \notin E^P \text{ ($P_i$ and $P_j$ cannot be aligned)}
     \end{array}
   \right.
\label{H_ij_def}
\end{equation}
We explain in more detail in Section \ref{sec:sub_align} the procedure by which we align pairs of patches, if such an alignment is at all possible.

Prior to computing the top eigenvectors of the matrix $H$, as introduced originally in \cite{sync}, we choose to use the following normalization (similar to 2D-ASAP in \cite{ASAP}). Let $D$ be a $3N \times 3N$ diagonal matrix\footnote{The diagonal matrix $D$ should not be confused with the partial distance matrix.}, whose entries are given by $D_{3i-2,3i-2}=D_{3i-1,3i-1}=D_{3i,3i} = deg(i)$, for $i = 1,\ldots,N$. We define the matrix
\begin{equation}
\mathcal{H} = D^{-1} H,
\end{equation}
which is similar to the symmetric matrix $D^{-1/2} H D^{-1/2}$ through
        $$\mathcal{H} = D^{-1/2} (D^{-1/2}  H D^{-1/2}) D^{1/2}. $$
Therefore, $\mathcal{H}$ has $3N$ real eigenvalues $\lambda_1^\mathcal{H} \geq \lambda_2^\mathcal{H} \geq \lambda_3^\mathcal{H}  \geq \lambda_4^\mathcal{H}  \geq \cdots \geq \lambda_{3N}^\mathcal{H} $ with corresponding $3N$ orthogonal eigenvectors $v_1^\mathcal{H} ,\ldots,v_{3N}^\mathcal{H} $, satisfying $\mathcal{H} v_i^\mathcal{H}  = \lambda_i^H v_i^\mathcal{H} $.
As shown in the next paragraphs, in the noise free case, $\lambda_1^\mathcal{H}=\lambda_2^\mathcal{H}=\lambda_3^\mathcal{H}$, and furthermore, if the patch graph is connected, then $\lambda_3^\mathcal{H} > \lambda_4^\mathcal{H}$. We define the estimated orthogonal transformations $\hat{h}_1,\ldots,\hat{h}_N \in $ O(3) using the top three eigenvectors $v_1^{\mathcal{H}}, v_2^{\mathcal{H}}, v_3^{\mathcal{H}}$, following the approach used in \cite{Singer_Shkolnisky}.

Let us now show that, in the noise free case, the top three eigenvectors of $\mathcal{H}$ perfectly recover the unknown group elements. We denote by $h_{i}$ the $3 \times 3$ matrix corresponding to the $i^{th}$ submatrix in the $3 \times N$  matrix $[v_1^{\mathcal{H}}, v_2^{\mathcal{H}}, v_3^{\mathcal{H}}]$. In the noise free case, $h_{i}$ is an orthogonal matrix and represents the solution which aligns patch $P_i$ in the global coordinate system, up to a global orthogonal transformation. To see this, we first let $h$ denote the $3N \times 3$ matrix formed by concatenating the true orthogonal transformation matrices $h_1,\ldots,h_N$. Note that when the patch graph $G^P$ is complete, $H$ is a rank 3 matrix since $H = h h^{T}$, and its top three eigenvectors are given by the columns of $h$
\begin{equation}
H h = h h^T h = h N I_3 = N h.
\end{equation}
In the general case when $G^P$ is a sparse connected graph, note that
\begin{equation}
  H h = D h, \mbox{ hence } D^{-1}H h = \mathcal{H} h = h,
\label{eigHh}
\end{equation}
and thus the three columns of $h$ are each eigenvectors of matrix $\mathcal{H}$, associated to the same eigenvalue $\lambda = 1$ of multiplicity 3. It remains to show this is the largest eigenvalue of $\mathcal{H}$ . We recall that the adjacency matrix of $G^P$ is $A^P$, and denote by $\mathcal{A}^P$ the $3N \times 3N$ matrix built by replacing each entry of value 1 in $A^P$ by the identity matrix $I_3$, i.e., $\mathcal{A}^P = A^P \otimes I_3$ where $\otimes$ denotes the tensor product of two matrices. As a consequence, the eigenvalues of $\mathcal{A}^P$ are just the direct products of the eigenvalues of $I_3$ and $A^P $, and the corresponding eigenvectors of $\mathcal{A}^P$ are the tensor products of the eigenvectors of $I$ and $A^P$. Furthermore, if we let $\Delta$ denote the $N \times N$ diagonal matrix with $\Delta_{ii} = deg(i),$ for $i=1,\ldots,N$, it holds true that
\begin{equation}
  D^{-1} \mathcal{A}^P = (\Delta^{-1} A^P) \otimes I_3,
\end{equation}
and thus the eigenvalues of $D^{-1} \mathcal{A}^P$ are the same as the eigenvalues of $\Delta^{-1} A^P$, each with multiplicity $3$. In addition, if $\Upsilon$ denotes the $3N \times 3N$  matrix with diagonal blocks $h_i$, $i=1,\ldots,N$, then the normalized alignment matrix $\mathcal{H}$ can be written as
\begin{equation}
    \mathcal{H} = \Upsilon D^{-1} \mathcal{A}^P \Upsilon^{-1},
\end{equation}
and thus $\mathcal{H}$ and $D^{-1} \mathcal{A}^P$ have the same eigenvalues, which are also the eigenvalues of $\Delta^{-1} A^P$, each with multiplicity 3. Whenever it is understood from the context, we will omit from now on the remark about the multiplicity 3.
Since the normalized discrete graph Laplacian $\mathcal{L}$ is defined as
\begin{equation}
\mathcal{L} = I - \Delta^{-1} A^P,
\label{laplacian}
\end{equation}
it follows that in the noise-free case, the eigenvalues of $I - \mathcal{H}$  are the same as the eigenvalues of $\mathcal{L}$. These eigenvalues are all non-negative, since $\mathcal{L}$ is similar to the positive semidefinite matrix $I- \Delta^{-1/2} A^P  \Delta^{-1/2}$, whose non-negativity follows from the identity $$x^T (I-\Delta^{-1/2} A^P \Delta^{-1/2}) x = \sum_{(i,j)\in E^P} \left(\frac{x_i}{\sqrt{deg(i)}} - \frac{x_j}{\sqrt{deg(j)}} \right)^2 \geq 0.$$ In other words,
\begin{equation}
1-\lambda_{3i-2}^{\mathcal{H}} = 1-\lambda_{3i-1}^{\mathcal{H}} = 1-\lambda_{3i}^{\mathcal{H}} = \lambda_i^{\mathcal{L}} \geq 0, \quad i=1,2,\ldots,N,
\end{equation}
where the eigenvalues of $\mathcal{L}$ are ordered in increasing order, i.e., $\lambda_1^{\mathcal{L}} \leq \lambda_2^{\mathcal{L}} \leq \cdots \leq \lambda_N^{\mathcal{L}}$.
If the patch graph $G^P$ is connected, then the eigenvalue $\lambda_1^{\mathcal{L}}=0$ is simple (thus $\lambda_2^{\mathcal{L}} > \lambda_1^{\mathcal{L}}  $) and its corresponding eigenvector $v_1^{\mathcal{L}}$ is the all-ones vector $\mb{1} = (1,1,\ldots,1)^T$.
Therefore, the largest eigenvalue of $\mathcal{H}$ equals 1 and has multiplicity 3, i.e., $\lambda_1^\mathcal{H}=\lambda_2^\mathcal{H}=\lambda_3^\mathcal{H}=1$, and $ \lambda_4^\mathcal{H} > \lambda_3^\mathcal{H}$. This concludes our proof that,  in the noise free case, the top three eigenvectors of $\mathcal{H} $ perfectly recover the true solution $h_1,\ldots,h_N \in$ O(3), up to a global orthogonal transformation.

However, when distance measurements are noisy and the pairwise alignments between patches are inaccurate, an estimated transformation $\widetilde{h}_i$ may not coincide with $h_i$, and in fact may not even be an orthogonal transformation. For that reason, we estimate $h_{i}$ by the closest orthogonal matrix to $\widetilde{h}_{i}$ in the Frobenius matrix norm\footnote{We remind the reader that the Frobenius norm of an $m \times n$ matrix $A$ can be defined in several ways
$ \|A\|_F^2 = \sum_{i=1}^{m} \sum_{j=1}^{n} |a_{ij}|^2 = Tr (A^T A) = \sum_{i=1}^{ \min(n,m)} \sigma_i^2 $, where $\sigma_i$ are the singular values of $A$.}
\begin{equation}
  \label{est-h}
  \hat{h}_i  =  \underset{ X \in  O(3)} { \operatorname{argmin}} \|\widetilde{h}_{i}  - X\|_F
\end{equation}
We do this by using the well known procedure (e.g.,\cite{arun}), $ \hat{h}_i =U_i V_i^T$, where $ \widetilde{h}_i =U_i \Sigma_i V_i^T$ is the singular value decomposition of $\widetilde{h}_i$, see also \cite{Fan1955} and \cite{Keller1975}. Note that the estimation of the orthogonal transformations of the patches are up to a global orthogonal transformation (i.e., a global rotation and reflection with respect to the original embedding). Also, the only difference between this step and the angular synchronization algorithm in \cite{sync} is the normalization of the matrix prior to the computation of the top eigenvector. The usefulness of this normalization was first demonstrated in 2D-ASAP, in the synchronization process over $\mathbb{Z}_2$ and SO(2).

\begin{figure}[h]
\centering
\subfigure[$\eta=0\%, \; \tau = 0\%$,  and $MSE = 6e-4 $ ]{
\includegraphics[width=0.31\columnwidth]{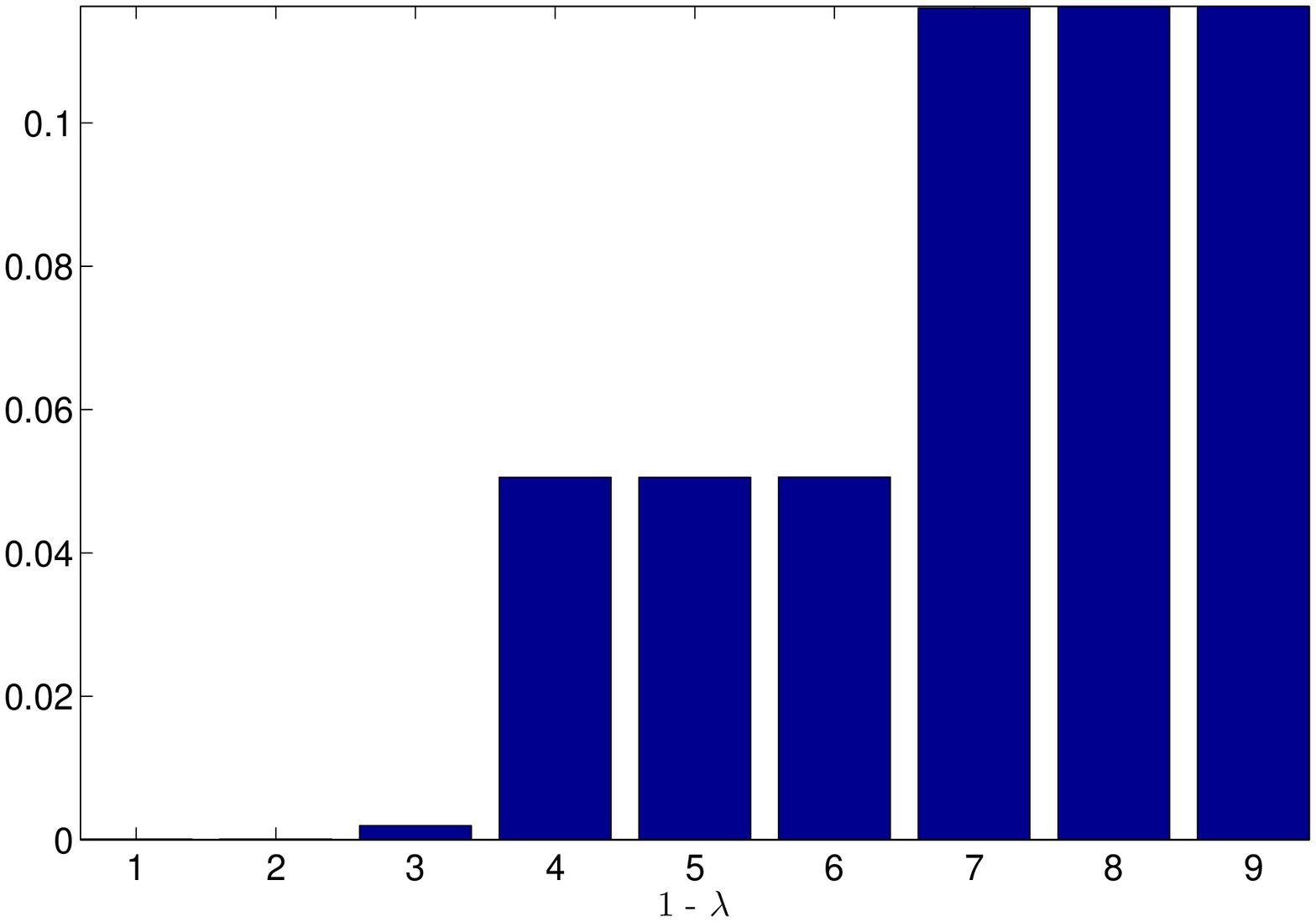}
}
\subfigure[$\eta=20\%, \; \tau = 0\%$, and $MSE = 0.05$]{
\includegraphics[width=0.31\columnwidth]{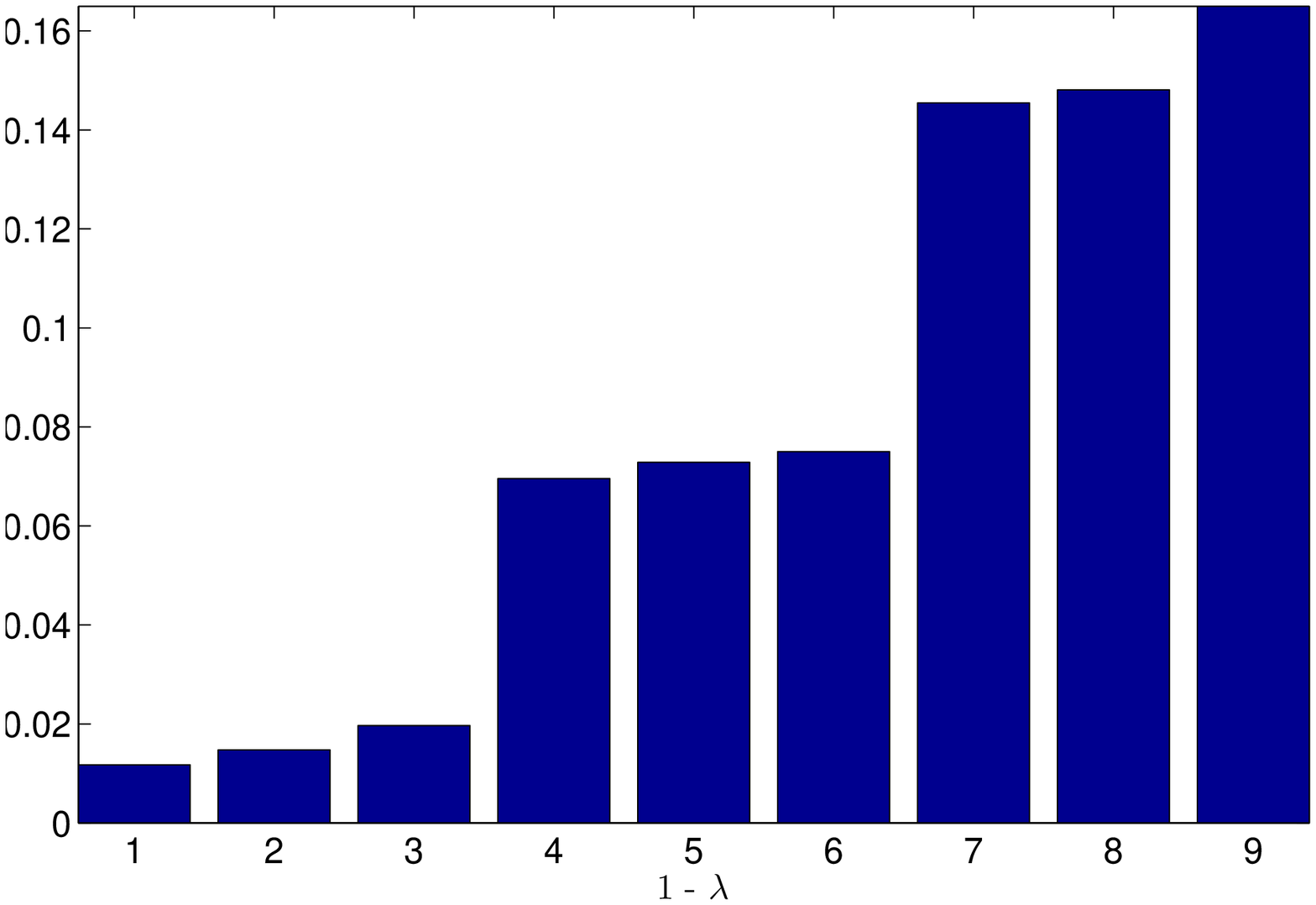}
}
\subfigure[$\eta=40\%, \; \tau = 4\%$, and $MSE = 0.36$]{
\includegraphics[width=0.31\columnwidth]{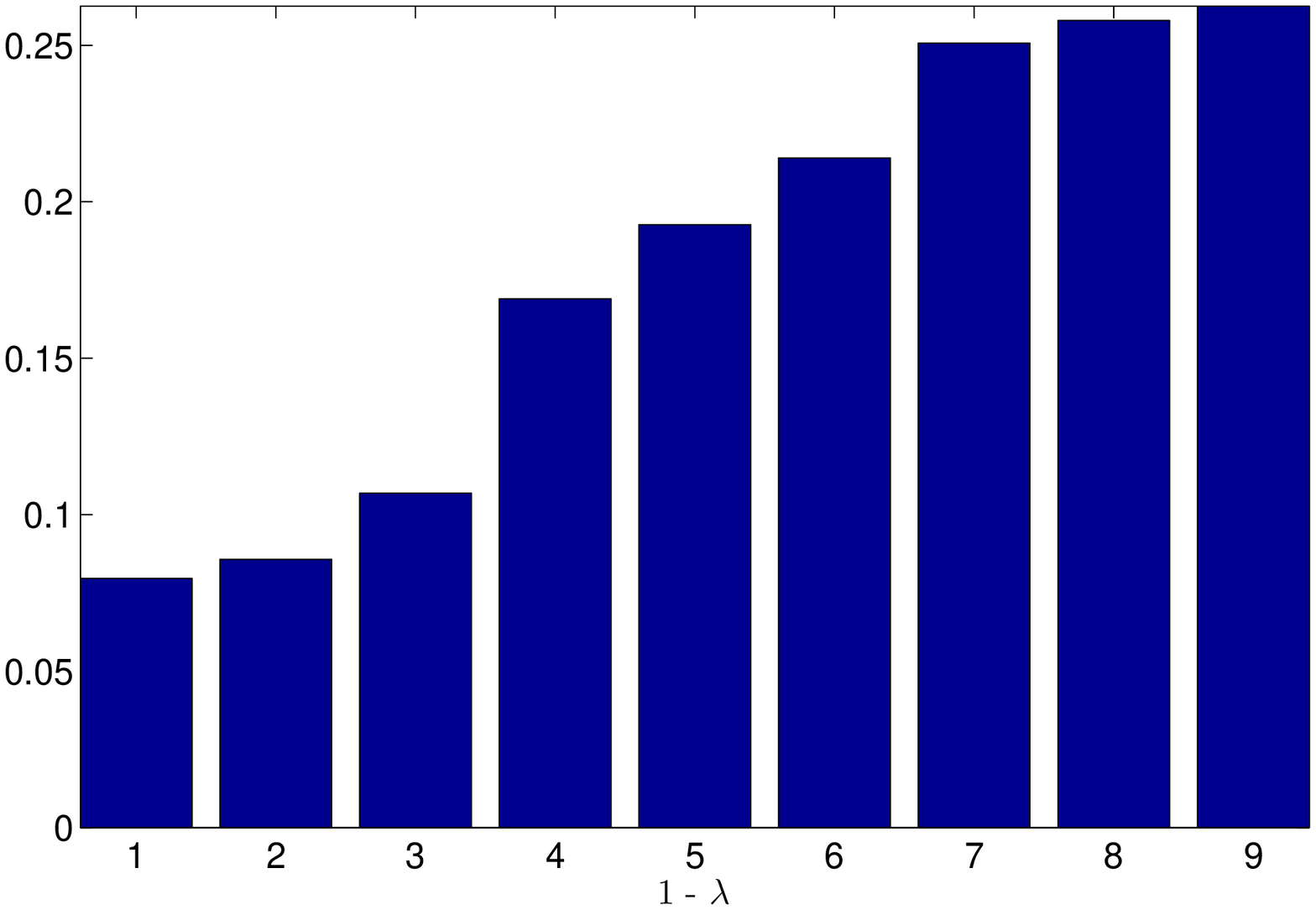}
}
\caption{Bar-plot of the top $9$ eigenvalues of $\mathcal{H}$ for the UNITCUBE and various noise levels $\eta$. The resulting error rate $\tau$ is the percentage of patches whose reflection was incorrectly estimated. To ease the visualization of the eigenvalues of $\mathcal{H}$, we choose to plot $1-\lambda^\mathcal{H}$ because the top eigenvalues of $\mathcal{H}$ tend to pile up near $1$, so it is difficult to differentiate between them by looking at the bar plot of $\lambda^\mathcal{H}$.
}
\label{fig:Spectra_Steps_12}
\end{figure}


We use the mean squared error (MSE) to measure the accuracy of this step of the algorithm in estimating the orthogonal transformations. To that end, we look for an optimal orthogonal transformation $\hat{O} \in $ O(3) that minimizes the sum of squared distances between the estimated orthogonal transformations and the true ones:
\begin{equation}
 \hat{O} = \underset{O \in O(3)}{\operatorname{argmin}}  \sum_{i=1}^{N} \|h_i - O \hat{h}_{i} \|_{F}^{2}
\end{equation}
In other words, $\hat{O}$ is the optimal solution to the registration problem between two sets of orthogonal transformations in the least squares sense. Following the analysis of
\cite{Singer_Shkolnisky}, we make use of properties of the trace such as $Tr(AB)=Tr(BA)$, $Tr(A)=Tr(A^T)$ and notice that

\begin{eqnarray}
 \sum_{i=1}^{N} \|h_i - O \hat{h}_{i} \|_{F}^{2}  &=& \sum_{i=1}^{N} Tr \left[ \left( h_i - O \hat{h}_{i} \right) \left( h_i - O \hat{h}_{i} \right)^T \right]     \nonumber \\
							&=&  \sum_{i=1}^{N} Tr \left[ 2 I - 2 O \hat{h}_i h_i^T \right] =   6 N - 2 Tr \left[ O  \sum_{i=1}^{N}  \hat{h}_i h_i^T \right]
\label{MSE_deriv}
\end{eqnarray}
If we let $Q$ denote the $3 \times 3$ matrix
\begin{equation}
  Q = \frac{1}{N} \sum_{i=1}^{N} \hat{h}_i h_i^T
\end{equation}
it follows from (\ref{MSE_deriv}) that the MSE is given by minimizing
\begin{equation}
   \frac{1}{N}  \sum_{i=1}^{N} \|h_i - O \hat{h}_{i} \|_{F}^{2} = 6 - 2 Tr(O Q).
\end{equation}
In \cite{arun} it is proven that $Tr(O Q) \leq Tr (V U^T Q)$, for all $O \in O(3)$, where $ Q = U \Sigma V^T$ is the singular
value decomposition of $Q$. Therefore, the MSE is minimized by the orthogonal matrix $ \hat{O} = VU^T$ and is given by
\begin{equation}
   \frac{1}{N}  \sum_{i=1}^{N} \|h_i - \hat{O} \hat{h}_{i} \|_{F}^{2} = 6 - 2 Tr( V U^T  U \Sigma V^T ) =
6 - 2 \sum_{k=1}^{3} \sigma_k
\end{equation}
where $\sigma_1, \sigma_2, \sigma_3$ are the singular values of $Q$. Therefore, whenever $Q$ is an orthogonal matrix for which $\sigma_1= \sigma_2= \sigma_3=1$, the MSE vanishes. Indeed, the numerical experiments in Table \ref{tab:Stats_12_UNIT} confirm that  for noiseless data, the MSE is very close to zero. To illustrate the success of the eigenvector method in estimating the reflections, we also compute $\tau$, the percentage of patches whose reflection was incorrectly estimated. Finally, the last two columns in Table \ref{tab:Stats_12_UNIT} show the recovery errors when, instead of doing synchronization over O(3), we first synchronize over $\mathbb{Z}_2$ followed by SO(3). 

\begin{table}[h]
\begin{minipage}[b]{0.99\linewidth}
\centering
\begin{tabular}{|c||c|c||c|c|}
\hline
&  \multicolumn{2}{c||}{O(3)} &  \multicolumn{2}{c|}{$\mathbb{Z}_2$ and SO(3)}  \\
\cline{2-5}
   $\eta$  & $\tau$ & MSE & $\tau$ & MSE  \\
\hline\hline
0\%  &   0\%  &  6e-4  &  0\% & 7e-4 \\
10\% &   0\%  &  0.01 &  0\% & 0.01 \\
20\% &   0\%  &   0.05 & 0\% & 0.05\\
30\% &   5.8\% & 0.35  & 5.3\% & 0.32 \\
40\% &   4\%   &   0.36 & 5\% & 0.40 \\
50\% &   7.4\%  & 0.65  & 9\% & 0.68 \\
\hline
\end{tabular}
\caption{The errors in estimating the reflections and rotations when aligning the $N=200$ patches resulting from for the UNITCUBE graph on $n=212$ vertices, at various levels of noise. We used $\tau$ to denote the percentage of patches whose reflection was incorrectly estimated.}
\label{tab:Stats_12_UNIT}
\end{minipage}
\end{table}


\subsection{Step 2: Synchronization over $\mathbb{R}^3$ to estimate translations}
The final step of the 3D-ASAP algorithm is computing the global translations of all patches and recovering the true coordinates. For each patch $P_k$, we denote by $G_k = (V_k,E_k)$\footnote{Not to be confused with $G(i)=(V(i),E(i))$ defined in the beginning of this section.}
the graph associated to patch $P_k$, where $V_k$ is the set of nodes in $P_k$, and $E_k$ is the set of edges induced by $V_k$ in the measurement graph $G=(V,E)$.
\begin{figure}[h]
\begin{center}
 \includegraphics[width=0.35\columnwidth, keepaspectratio = true]{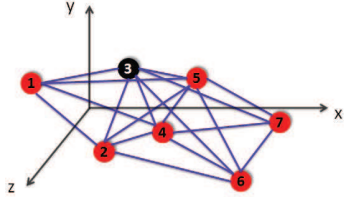}
\end{center}
\caption{ An embedding of a patch $P_k$ in its local coordinate system (frame) after it was appropriately reflected and rotated. In the noise-free case, the coordinates $p_i^{(k)} = (x_i^{(k)},y_i^{(k)},z_i^{(k)})^T$ agree with the global positioning $p_i = (x_i, y_i, z_i)^T$ up to some translation $t^{(k)}$ (unique to all $i$ in $V_k$).}
\label{fig:patch_xy_frame}
\end{figure}
We denote by $p^{(k)}_i = ( x_i^{(k)} , y_i^{(k)}, z_i^{(k)} )^T$ the known local frame coordinates of node $i \in V_k$ in the embedding of patch $P_k$ (see Figure \ref{fig:patch_xy_frame}).

At this stage of the algorithm, each patch $P_k$ has been properly reflected and rotated so that the local frame coordinates are consistent with the global coordinates, up to a translation $t^{(k)} \in \mathbb{R}^3$. In the noise-free case we should therefore have
\begin{equation}
\label{lsqr}
p_i = p_i^{(k)} + t^{(k)}, \quad i\in V_k, \quad k=1,\ldots,N.
\end{equation}
We can estimate the global coordinates $p_1,\ldots,p_n$ as the least squares solution to the overdetermined system of linear equations (\ref{lsqr}), while ignoring the by-product translations $t^{(1)},\ldots,t^{(N)}$. In practice, we write a linear system for the displacement vectors $p_i-p_j$ for which the translations have been eliminated. Indeed, from (\ref{lsqr}) it follows that each edge $(i,j) \in E_k$ contributes a linear equation of the form \footnote{In fact, we can write such equations for every $i,j\in V_k$ but choose to do so only for edges of the original measurement graph.}
\begin{equation}
 p_i - p_j = p^{(k)}_i - p^{(k)}_j, \quad  (i,j) \in E_k, \quad k=1,\ldots,N.
\label{st3const}
\end{equation}
In terms of the $x$, $y$ and $z$ global coordinates of nodes $i$ and $j$, (\ref{st3const}) is equivalent to
\begin{eqnarray}
	x_i - x_j &=& x^{(k)}_i - x^{(k)}_j,\quad  (i,j) \in E_k,\quad k=1,\ldots,N, \label{x_eq}
\end{eqnarray}
and similarly for the $y$ and $z$ equations. We solve these three linear systems separately, and recover the coordinates $x_1,\ldots,x_n$, $y_1,\ldots,y_n$, and $z_1,\ldots,z_n$.  Let $T$ be the least squares matrix associated with the overdetermined linear system in (\ref{x_eq}), $x$ be the $n\times 1$ vector representing the $x$-coordinates of all nodes, and $b^x$ be the vector with entries given by the right-hand side of (\ref{x_eq}). Using this notation, the system of equations given by (\ref{x_eq}) can be written as
\begin{equation}
T x = b^{x},
\end{equation}
and similarly for the $y$ and $z$ coordinates. Note that the matrix $T$ is sparse with only two non-zero entries per row and that the all-ones vector $\mb{1} = (1,1,\ldots,1)^T$ is in the null space of $T$, i.e., $T\mb{1} = 0$, so we can find the coordinates only up to a global translation.

To avoid building a very large least squares matrix, we combine the information provided by the same edges across different patches in only one equation, as opposed to having one equation per patch. In 2D-ASAP \cite{ASAP}, this was achieved by adding up all equations of the form (\ref{x_eq}) corresponding to the same edge $(i,j)$ from different patches, into a single equation, i.e.,
\begin{equation}
\sum_{k \in \{1,\ldots,N\} \; s.t. (i,j) \in E_k} x_i-x_j = \sum_{k \in \{1,\ldots,N\} \; s.t. (i,j) \in E_k}
x_i^{(k)} - x_j^{(k)}, \quad  (i,j) \in E,
\label{lsstep3}
\end{equation}
and similarly for the $y$ and $z$-coordinates. For very noisy distance measurements, the displacements $x_i^{(k)} - x_j^{(k)}$ will also be corrupted by noise and the motivation for (\ref{lsstep3}) was that adding up such noisy values will average out the noise. However, as the noise level increases, some of the embedded patches will be highly inaccurate and will thus generate outliers in the list of displacements above. To make this step of the algorithm more robust to outliers, instead of averaging over all displacements, we select the median value of the displacements and use it to build the least squares matrix
\begin{equation}
x_i-x_j = \underset{k \in \{1,\ldots,N\} \; s.t. (i,j) \in E_k}{\operatorname{median}}
                     \{  x_i^{(k)} - x_j^{(k)} \} , \quad  (i,j) \in E,
\label{lsstep3_3d}
\end{equation}
We denote the resulting $m\times n$ matrix by $\tilde{T}$, and its $m\times 1$ right-hand-side vector by $\tilde{b}^x$. Note that $\tilde{T}$ has only two nonzero entries per row \footnote{Note that some edges in $E$ may not be contained in any patch $P_k$, in which case the corresponding row in $\tilde{T}$ has only zero entries.}, $\tilde{T}_{e,i}=1, \tilde{T}_{e,j}=-1$, where $e$ is the row index corresponding to the edge $(i,j)$. The least squares solution
 $\hat{p} = \hat{p}_1,\ldots,\hat{p}_n$ to
\begin{equation}
\tilde{T} x = \tilde{b}^x,\; \mbox{ }\;\tilde{T} y = \tilde{b}^y, \; \mbox{and }\;\tilde{T} z = \tilde{b}^z,
\label{lsstep3final}
\end{equation}
is our estimate for the coordinates $p = p_1,\ldots,p_n$, up to a global rigid transformation.

Whenever the ground truth solution $p$ is available, we can compare our estimate $\hat{p}$ with $p$. To that end, we remove the global reflection, rotation and translation from $\hat{p}$, by computing the best procrustes alignment between $p$ and $\hat{p}$, i.e. $ \tilde{p} = O \hat{p} + t$, where $O$ is an orthogonal rotation and reflection matrix, and $t$ a translation vector, such that we minimize the distance between the original configuration $p$ and $\tilde{p}$, as measured by the least squares criterion $ \sum_{i=1}^{n} \| p_i - \tilde{p}_i \| ^2$. Figure \ref{fig:hist_coord_error} shows
the histogram of errors in the coordinates, where the error associated with node $i$ is given by $ \parallel p_i - \hat{p_i} \parallel$.

\begin{figure}[h]
\begin{center}
\subfigure[$\eta = 0\%, ERR_c  = 2e-3$]{\includegraphics[width=0.32\columnwidth]{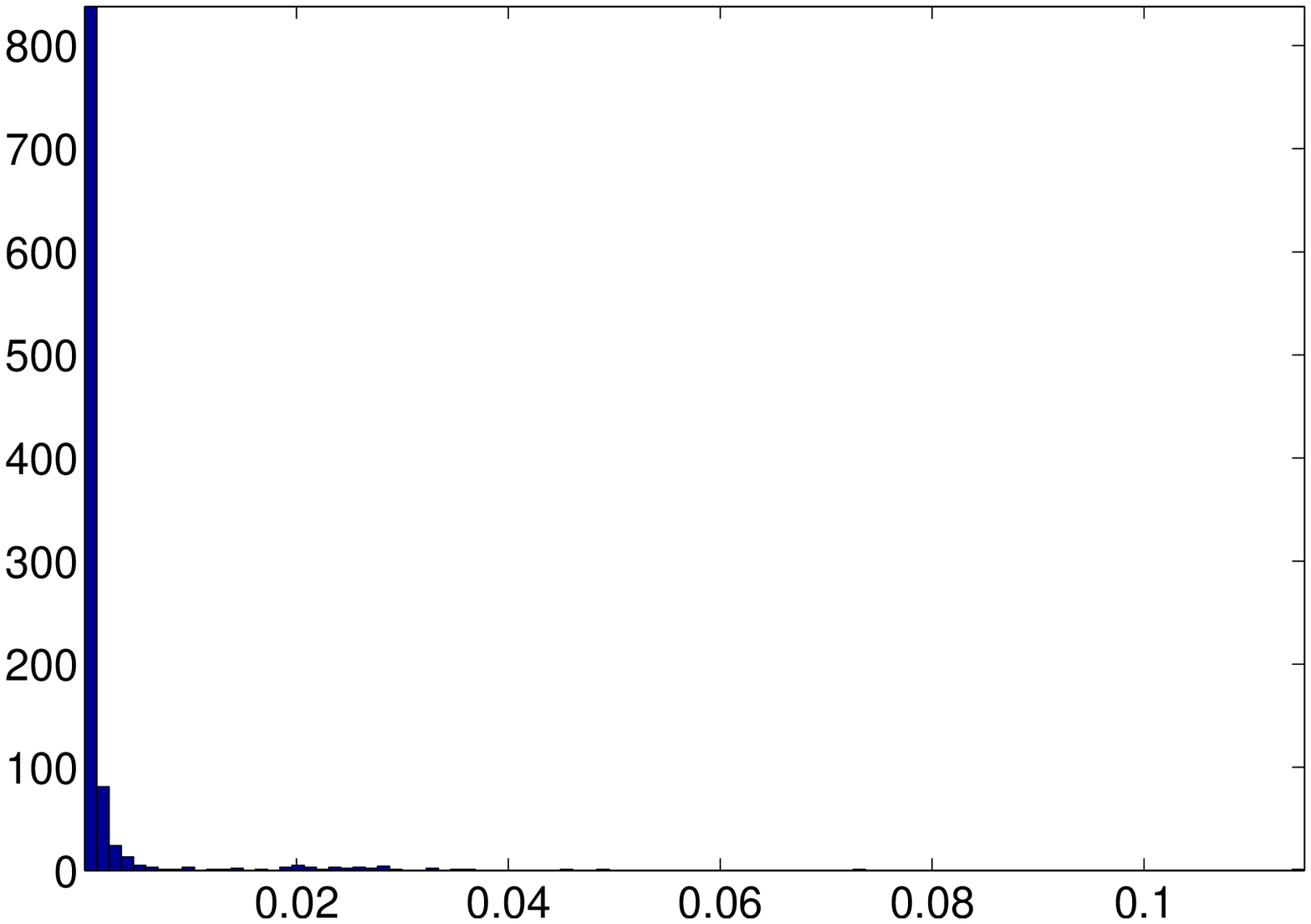}}
\subfigure[$\eta = 30\%, ERR_c =0.57$  ]{\includegraphics[width=0.32\columnwidth]{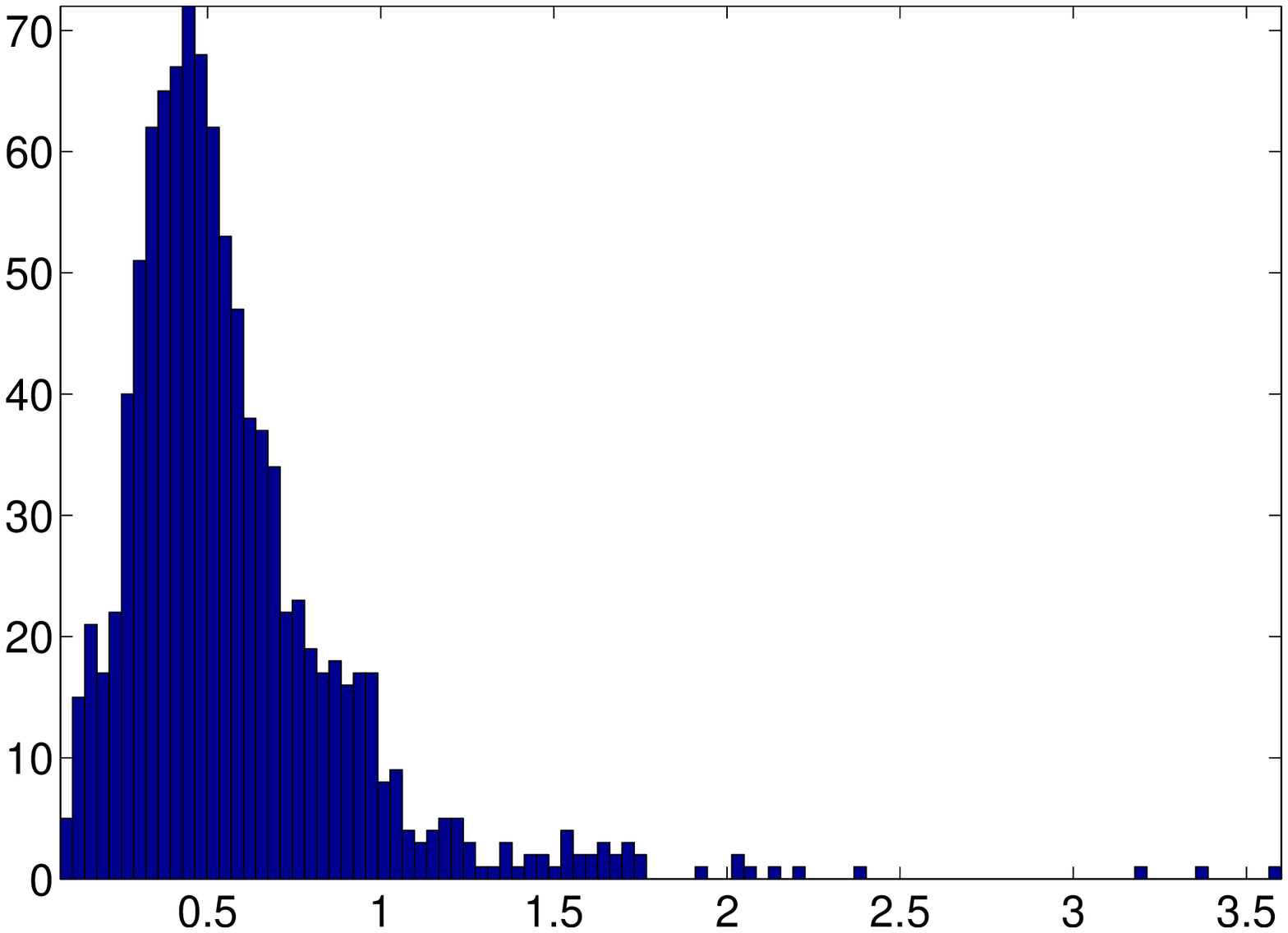}}
\subfigure[$\eta = 50\%, ERR_c =1.23$]{\includegraphics[width=0.32\columnwidth]{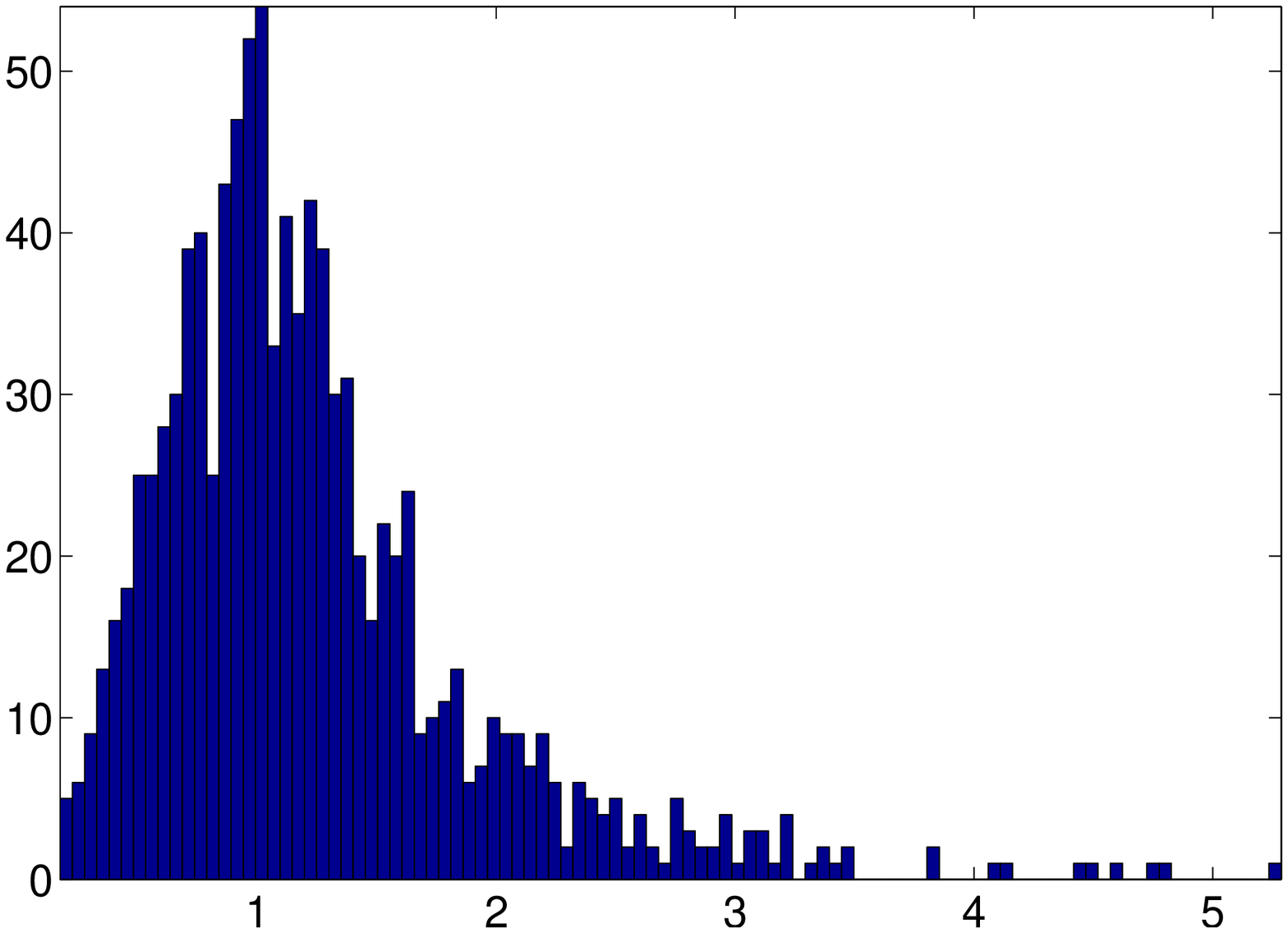}}
\end{center}
\caption{Histograms of coordinate errors $\|p_i - \hat{p}_i\|$ for all atoms in the 1d3z molecule,
for different levels of noise. In all figures, the x-axis is measured in Angstroms. Note the change of scale for Figure (a), and the fact that the largest error showed there is 0.12. We used $ERR_c$ to denote the average coordinate error of all atoms.}
\label{fig:hist_coord_error}
\end{figure}

We remark that in 3D-ASAP anchor information can be incorporated similarly to the 2D-ASAP algorithm \cite{ASAP}; however we do not elaborate on this here since there are no anchors in the molecule problem.

\section{Extracting, embedding, and aligning patches}
\label{WULS}

This section describes how to break up the measurement graph into patches, and how to embed and pairwise align the resulting patches. We start in Section \ref{sec:sub_kstar} with a description of how to extract globally rigid subgraphs from 1-hop neighborhood graphs, and discuss the advantages that  $k$-star graphs bring. However, in practice, we do not follow either approach  to extract patches, since SDP-based methods may still produce inaccurate localizations of globally rigid patches even in the case of noiseless data. Therefore, in Section \ref{sec:sub_wuls} we recall a recent result of \cite{SY05} on \textit{uniquely d-localizable} graphs, which can be accurately localized by SDP-based methods. We thus lay the ground for the notion of \textit{weakly uniquely localizable} (WUL) graphs, which we introduce with the purpose of being able to localize the resulting patches even when the distance measurements are noisy. Section \ref{sec:sub_pseudo} discusses the issue of finding ``pseudo-anchor" nodes, which are needed when extracting WUL subgraphs. In Section \ref{sec:sub_emb} we discuss several SDP-relaxations to the graph localization problem, which we use to embed the WUL patches. In Section \ref{sec:sub_infomol}  we remark on several additional constraints specific to the molecule problem, which are currently not incorporated in 3D-ASAP. Finally, Section \ref{sec:sub_align} explains the procedure for aligning a pair of overlapping patches.

\subsection{Extracting globally rigid subgraphs}
\label{sec:sub_kstar}

Although there is no combinatorial characterization of global rigidity in $\mathbb{R}^3$ (but only necessary conditions \cite{conditions_UGR}), one may exploit the fact that a 1-hop neighborhood subgraph is always a  ``star" graph, i.e., a graph with at least one central node connected to everybody else.
The process of coning a graph $G$ adds a new vertex $v$, and adds edges from $v$ to all original vertices in $G$, creating the cone graph $G * v$. A recent result of \cite{coning} states that a graph is generically globally rigid in $\mathbb{R}^{d-1}$ iff the cone graph is generically globally rigid in $\mathbb{R}^{d}$. This result allows us to reduce the notion of global rigidity in $\mathbb{R}^3$ to global rigidity in $\mathbb{R}^2$, after removing the center node. We recall that sufficient and necessary conditions for two dimensional combinatorial global rigidity have been recently established \cite{JacksonJordan,conditions_UGR}, i.e., 3-connectivity and redundant rigidity  (meaning the graph remains locally rigid after the removal of any single edge). Therefore, breaking a graph into maximally globally rigid components amounts to finding the maximally 3-connected components, and furthermore, extracting the maximally redundantly rigid components. The first $O(n + m)$ algorithm for finding the 3-connected components of a general graph was given by Hopcroft and Tarjan \cite{tarjan}.
A combinatorial characterization of redundant rigidity in dimension two was given in \cite{conditions_UGR}, together with an $O(n^2)$ algorithm.

The rest of this section presents an alternative method for extracting globally rigid subgraphs, while avoiding the notion of redundant rigidity. Motivated by the fact that 1-hop neighborhood graphs may have multiple centers (especially in random geometric graphs), we introduce the following family of graphs. A \textbf{$k$-star graph} is a graph which contains at least $k$ vertices (centers) that are connected to all remaining nodes.  For example, for each node $i$, the local graph $G(i)$ composed of the central node $i$ and all its neighbors takes the form of a 1-star graph (which we simply refer to as a star graph). Note that in our definition, unlike perhaps more conventional definitions  of star graphs, we allow edges between non-central nodes to exist.
 \begin{proposition}
    A $(k-1)$-star graph is generically globally rigid in  $\mathbb{R}^k$ iff it is $(k+1)$-vertex-connected.
\label{ggrprop}
\end{proposition}
\begin{proof}
We prove the statement in the proposition by induction on $k$. The case $k=2$ was previously shown in \cite{ASAP}. Assuming the statement holds true for $k-2$, we show it remains true for $k-1$ as well.

Let $H$ be a $(k+1)$-vertex-connected $(k-1)-$star graph, $\mathcal{S} =\{v_1,\ldots,v_{k-1} \}$ be its center nodes, and
$H^{*}$ the graph obtained by removing one of its center nodes.
Since $H$ is $(k+1)$-vertex-connected then $H^{*}$ must be $k$-vertex-connected, since otherwise, if $u$ is a
cut-vertex in $H^{*}$, then $\{ v_1,u\}$ is a vertex-cut of size 2 in $H$, a contradiction.
By induction, $H^{*}$ is generically globally rigid in  $\mathbb{R}^{k-1}$, since it is a $(k-2)$-star graph that is $k$-vertex-connected.
Using the coning theorem, the generic globally rigidity of $H^{*}$ in $\mathbb{R}^{k-1}$ implies that $H$ is generically globally rigid in $\mathbb{R}^k$.

The converse is a well known statement in rigidity theory \cite{conditions_UGR}. We first say that a framework in $\mathbb{R}^k$ allows a reflection if a separating set of vertices lies in $(k-1)$-dimensional subspace. However, realization for which more than $k$ vertices lie in a $(k-1)$-dimensional subspace are not generic, and therefore, for generic frameworks, reflections occur when there is a subset of $k$ or fewer vertices whose removal disconnects the graph, i.e., $G$ is not $(k+1)$-vertex-connected. In other words, if $H$ is generically globally rigid in  $\mathbb{R}^k$, then it must be $(k+1)$-vertex-connected.
\end{proof}

Using the above result for $k=3$ implies that if the 1-hop neighborhood $G(i)$ of node $i$ has another center node $j\neq i$, then one can break $G(i)$ into maximally globally rigid subgraphs (in $\mathbb{R}^{3} $) by simply finding the 4-connected components of $G(i)$. Since $G(i)$ has two center nodes $i$ and $j$, the problem amounts to finding the 2-connected components of the remaining graph.

This observation suggests another possible approach to solving the network localization problem.  Instead of breaking the initial graph $G$ into patches by first using 1-hop neighborhoods of each node, one can start by considering the 1-hop neighborhood $G(i,j)$ of a pair of nodes $(i,j) \in E(G)$, where vertices of graph $G(i,j)$ are the intersection of the 1-hop neighbors of nodes $i$ and $j$. This approach assures that $G(i,j)$ is a $2$-star graph, and by the above result, can be easily broken into maximally globally rigid subgraphs, without involving the notion of redundant rigidity.

Note that, in practice, we do not follow either of the two approaches introduced in this section, since SDP-based methods may compute inaccurate localizations of globally rigid graphs, even in the case of noiseless data. Instead, we use the notion of weakly uniquely localizable graphs, which we introduce in the next section.

\subsection{Extracting  Weakly Uniquely Localizable (WUL) subgraphs}
\label{sec:sub_wuls}

We first recall some of the notation introduced earlier, that is needed throughout this section and Section \ref{sync_anchors} on synchronization over $\mathbb{Z}_2$ with anchors. We consider a sensor network in $ \mathbb{R}^3$ with $k$ anchors denoted by $\mathcal{A}$, and $n$ sensors denoted by $\mathcal{S}$. An anchor is a node whose location $a_i \in \mathbb{R}^3$ is readily available, $i=1,\ldots,k$, and a sensor is a node whose location $x_j$ is to be determined, $j=1,\ldots,n$. Note that for aesthetic reasons and consistency with the notation used in \cite{SY05}, we use $x$ in this section to denote the true  coordinates, as opposed to $p$ used throughout the rest of the paper\footnote{Not to be confused with the $x$-axis projections of the distances in the least squares step.}.
We denote by $d_{ij}$ the Euclidean distance between a pair of nodes, $(i,j) \in \mathcal{A} \cup \mathcal{S}$. In most applications, not all pairwise distance measurements are available, therefore we denote by $E(\mathcal{S}, \mathcal{S})$ and $E(\mathcal{S},\mathcal{A})$ the set of edges denoting the measured sensor-sensor and sensor-anchor distances. We represent the available distance measurements in an undirected graph $G=(V,E)$ with vertex set $V = \mathcal{A} \cup \mathcal{S}$ of size $|V|=n+k$, and edge set of size $|E|=m$. An edge of the graph corresponds to a distance constraint, that is $(i,j)\in E$ iff the distance between nodes $i$ and $j$ is available and equals $d_{ij} = d_{ji}$, where $i,j \in \mathcal{A} \cup \mathcal{S}$. We denote the partial distance measurements matrix by $D = \{ d_{ij} : (i,j) \in E(\mathcal{S}, \mathcal{S}) \cup E(\mathcal{S},\mathcal{A}) \}$. A solution $x$ together with the anchor set $a$ comprise a \emph{localization} or \emph{realization} $p=(x,a)$ of $G$. A {\em framework} in $\mathbb{R}^d$ is the ensemble $(G,p)$, i.e., the graph $G$ together with the  realization $p$ which assigns a point $p_i$ in $\mathbb{R}^d$ to each vertex $i$ of the graph.


Given a partial set of noiseless distances and the anchor set $a$, the graph realization problem can be formulated as the following system
 \begin{eqnarray}
\| x_i - x_j \|_2^2 & = & d_{ij}^2 \;\;\;\; \text{ for } (i,j) \in E(\mathcal{S}, \mathcal{S}) \nonumber \\
\| a_i - x_j \|_2^2 & = & d_{ij}^2 \;\;\;\; \text{ for } (i,j)  \in E(\mathcal{S}, \mathcal{A}) \nonumber \\
x_i  & \in & \mathbb{R}^d          \;\;\;\; \text{ for } i=1,\ldots,n
\label{NLP1}
\end{eqnarray}
Unless the above system has enough constraints (i.e., the graph $G$ has sufficiently many edges), then $G$ is not globally rigid and there could be multiple solutions. However, if the graph $G$ is known to be (generically) globally rigid in $\mathbb{R}^3$, and there are at least four anchors (i.e., $k\geq 4)$, and $G$ admits a generic realization\footnote{A realization is {\em generic} if the coordinates do not satisfy any non-zero polynomial equation with integer coefficients.}, then (\ref{NLP1}) admits a unique solution. Due to recent results on the characterization of generic global rigidity, there now exists a randomized efficient algorithm that verifies if a given graph is generically globally rigid in $\mathbb{R}^d$ \cite{GortlerGR}. However, this efficient algorithm does not translate into an efficient method for actually computing a realization of $G$. Knowing that a graph is generically globally rigid in $\mathbb{R}^d$ still leaves the graph realization problem intractable, as shown in \cite{Yang2004}. Motivated by this gap between deciding if a graph is generically globally rigid and computing its realization (if it exists), So and Ye introduced the following notion of unique $d$-localizability \cite{SY05}. An instance $(G,p)$ of the graph localization problem is said to be \emph{uniquely d-localizable} if
\begin{enumerate}
  \item the system (\ref{NLP1}) has a unique solution $\tilde{x} = (\tilde{x}_1;\ldots;\tilde{x}_n) \in \mathbb{R}^{nd}$, and
  \item  for any $l > d$,  $\tilde{x} = ((\tilde{x}_1;  \mb{0}),\ldots,(\tilde{x}_n; \mb{0})) \in \mathbb{R}^{nl}$ is the unique solution to the following system:
 \begin{eqnarray}
\| x_i - x_j \|_2^2 & = & d_{ij}^2 \;\;\;\; \text{ for } (i,j) \in E(\mathcal{S}, \mathcal{S}) \nonumber \\
\| (a_i;\mb{0}) - x_j \|_2^2 & = & d_{ij}^2 \;\;\;\; \text{ for } (i,j)  \in E(\mathcal{S}, \mathcal{A})  \\
x_i  & \in & \mathbb{R}^l          \;\;\;\; \text{ for } i=1,\ldots,n \nonumber
\label{NLP_unq}
\end{eqnarray}
\end{enumerate}
where $(v; \mb{0}) $ denotes the concatenation of a vector $v$ of size $d$ with the all-zeros vector $\mb{0}$ of size $l-d$.
The second condition states that  the problem cannot have a non-trivial localization in some higher dimensional space $\mathbb{R}^{l}$ (i.e, a localization different from the one obtained by setting $x_j = ( \tilde{x}_j; \mb{0}) $ for $j=1,\ldots,n$), where anchor points are trivially augmented to $(a_i;\mb{0})$, for $i=1,\ldots,k$.  A crucial observation should now be made: unlike global rigidity, which is a generic property of the graph $G$, the notion of \emph{unique localizability} depends not only on the underlying graph $G$ but also on the particular realization $p$, i.e., it depends on the framework $(G,p)$.

We now introduce the notion of a weakly uniquely localizable graph, essential for the preprocessing step of the 3D-ASAP algorithm, where we break the original graph into overlapping patches. A graph is \emph{weakly uniquely d-localizable} if there exists at least one realization $p \in \mathbb{R}^{(n+k)d}$ (we call this a certificate realization) such that the framework $(G,p)$ is uniquely localizable. Note that if a framework $(G,p)$ is uniquely localizable, then $G$ is a weakly uniquely localizable graph; however the reverse is not necessarily true since unique localizability is not a generic property.

The advantage of working with uniquely localizable graphs becomes clear in light of the following result by \cite{SY05}, which states that the problem of deciding whether a given graph realization problem is uniquely localizable, as well as the problem of determining the node positions of such a uniquely localizable instance, can be solved efficiently by considering the following SDP
\begin{align}
	& {\text{maximize}} & & 0 \nonumber \\
	& \text{subject to}	& &  (\mb{0}; e_i - e_j)(\mb{0}; e_i - e_j)^T  \cdot Z  =  d_{ij}^2, \;\;\;\; \text{ for } (i,j) \in E(\mathcal{S}, \mathcal{S}) \nonumber \\
	&  & &  (a_i; -e_j) (a_i; -e_j)^T  \cdot Z  =  d_{ij}^2, \;\;\;\; \text{ for } (i,j)  \in E(\mathcal{S}, \mathcal{A}) \nonumber  \\
	& &  &  Z_{1:d,1:d}   =   I_d  \nonumber \\
	& &  &  Z  \in  \mathcal{K}^{n+d}
\label{NLP_SDP}
\end{align}
where $e_i$ denotes the all-zeros vector with a $1$ in the $i^{th}$ entry, and $\mathcal{K}^{n+d} = \{ Z_{(n+d) \times (n+d)} | Z= \left[ \begin{array}{cc}
                                                I_d & X  \\
                                                X^{T} & Y \\
\end{array} \right]  \succeq 0 \}$, where $Z \succeq 0 $ means that $Z$ is a positive semidefinite matrix. The SDP method relaxes the constraint $Y=XX^T$ to $Y \succeq  XX^T$, i.e., $Y- XX^T \succeq 0$, which is equivalent to the last condition in (\ref{NLP_SDP}). The following predictor for uniquely localizable graphs introduced in \cite{SY05}, established for the first time that the graph realization problem is uniquely localizable if and only if the relaxation solution computed by an interior point algorithm (which generates feasible solutions of max-rank) has rank $d$ and $Y= XX^{T}$.
\begin{theorem}[\protect{\cite[Theorem 2]{SY05}}]
\label{th_SY05}
Suppose G is a connected graph. Then the following statements are equivalent:
\begin{abcliste}
\abc Problem (\ref{NLP1}) is uniquely localizable
\abc  The max-rank solution matrix $Z$ of (\ref{NLP_SDP}) has rank $d$
\abc  The solution matrix $Z$ represented by (b) satisfies $Y= XX^{T}$.
\end{abcliste}
\label{SO_YE_THM_UL}
\end{theorem}

Algorithm \ref{algo} summarizes our approach for extracting a WUL subgraph of a given graph, motivated by the results of Theorem \ref{SO_YE_THM_UL} and the intuition and numerical experiments described in the next paragraph. Note,  however, that the statements in Theorem \ref{SO_YE_THM_UL} hold true as long as the graph $G$ has at least four anchor nodes. While this may seem a very restrictive condition (since in many real life applications anchor information is rarely available) there is an easy way to get around this, provided the graph contains a clique (complete subgraph) of size at least $4$. As discussed in Section \ref{sec:sub_pseudo}, a patch of size at least $10$ contains such a clique with very high probability. Once such a clique has been found, one may use cMDS to embed it and use the coordinates as anchors. We call such nodes \textit{pseudo-anchors}.

\begin{algorithm}[h]
\caption{Finding a weakly uniquely localizable (WUL) subgraph of a graph with four anchors or pseudo-anchors
}
\begin{algorithmic}[1]
\REQUIRE Simple graph $G=(V,E)$ with $n$ sensors, $k$ anchors and $m$ edges corresponding to known pairwise Euclidean distances, and $\epsilon$ a small positive constant (e.g. $10^{-4}$)
\STATE Randomize a realization $p_1,\ldots,p_n$ in $\mathbb{R}^3$
\STATE If $k<4$, find a complete subgraph of $G$ on $4$ vertices (i.e., $K_4$) and compute an embedding of it (using classical MDS). Denote the set of pseudo-anchors by $\mathcal{A}$
\STATE Solve the SDP relaxation problem formulated in (\ref{NLP_SDP}) using the anchor set $\mathcal{A}$ 
\STATE Denote by the vector $w$ the diagonal elements of the matrix $Y-XX^{T}$.
\STATE Find the subset of nodes $V_0 \in V \backslash \mathcal{A}$ such that $w_i  < \epsilon$
\STATE Denote $G_0 = (V_0 , E_0)$ the weakly uniquely localizable subgraph of $G$.
\end{algorithmic}
\label{algo}
\end{algorithm}

Two known necessary conditions for global rigidity in $\mathbb{R}^3$ are $4$-connectivity and redundant rigidity (meaning the graph remains locally rigid after the removal of any single edge)
\cite{conditions_UGR,Connelly}. One approach to breaking up a patch graph into globally rigid subgraphs, used by the ABBIE algorithm of \cite{molecule_problem}, is to recurse on each 4-connected component of the graph, and then on each redundantly rigid subcomponent; but even then we are still not sure that the resulting subgraphs are globally rigid. Our approach is to extract a WUL subgraph from the 4-connected components of each patch. It may not be clear to the reader at this point what is the motivation for using weak unique localizability since it is not a generic property, and hence attempting to extract a WUL subgraph of a 4-connected graph seems meaningless since we do not know a priori  what is the true realization of such a graph. However, we have observed in our numerical simulations that this approach significantly improves the accuracy of the embeddings. An intuitive motivation for this approach is the following. If the randomized realization in Algorithm \ref{algo} (or what remains of it after removing some of the nodes) is ``faithful", meaning close enough to the true realization, then the WUL subgraph is perhaps generically uniquely localizable, and hence its localization using the SDP in (\ref{NLP_SDP}) under the original distance constraints can be computed accurately, as predicted by Theorem \ref{SO_YE_THM_UL}. We also consider a slight variation of Algorithm \ref{algo}, where we replace step 3 with the SDP relaxation introduced in the FULL-SDP algorithm of \cite{BiswasYe}. We refer to this different approach as Algorithm 2.
Note that we also consider Algorithm 2 in our simulations only for computational reasons, since the running time of the FULL-SDP algorithm is significantly smaller compared to our \texttt{CVX}-based SDP implementation \cite{CVX,gb08} of problem (\ref{NLP_SDP}).

Our intuition about the usefulness of the WUL subgraphs is supported by several numerical simulations.
Figure \ref{fig:wuls_unit} and Table \ref{tab:ANE_UNIT_1_hop} show the reconstruction errors of the patches (in terms of ANE, an error measure introduced in Section \ref{numexp}) in the following scenarios. In the first scenario, we directly embed each 4-connected component, without any prior preprocessing. In the second, respectively third, scenario  we first extract a WUL subgraph from each 4-connected component using Algorithm \ref{algo}, respectively Algorithm 2, and then embed the resulting subgraphs. Note that the subgraph embeddings are computed using FULL-SDP, respectively SNL-SDP, for noiseless, respectively noisy data.
Figure \ref{fig:wuls_unit} contains numerical results from the UNITCUBE graph with noiseless data, in the three scenarios presented above. As expected, the FULL-SDP embedding in scenario 1 gives the highest reconstruction error, at least one order of magnitude larger when compared to Algorithms 1 and 2. Surprisingly, Algorithm 2 produced more accurate reconstructions than Algorithm 1, despite its lower running time. These numerical computations suggest\footnote{Personal communication by Yinyu Ye.} that Theorem \ref{th_SY05} remains true when the formulation in problem (\ref{NLP_SDP}) is replaced by the one considered in the FULL-SDP algorithm \cite{BiswasYe}.

\begin{figure}[h]
\begin{center}
\subfigure[Scenario 1: $\overline{ANE}=8.4e-4$]{\includegraphics[width=0.32\columnwidth]{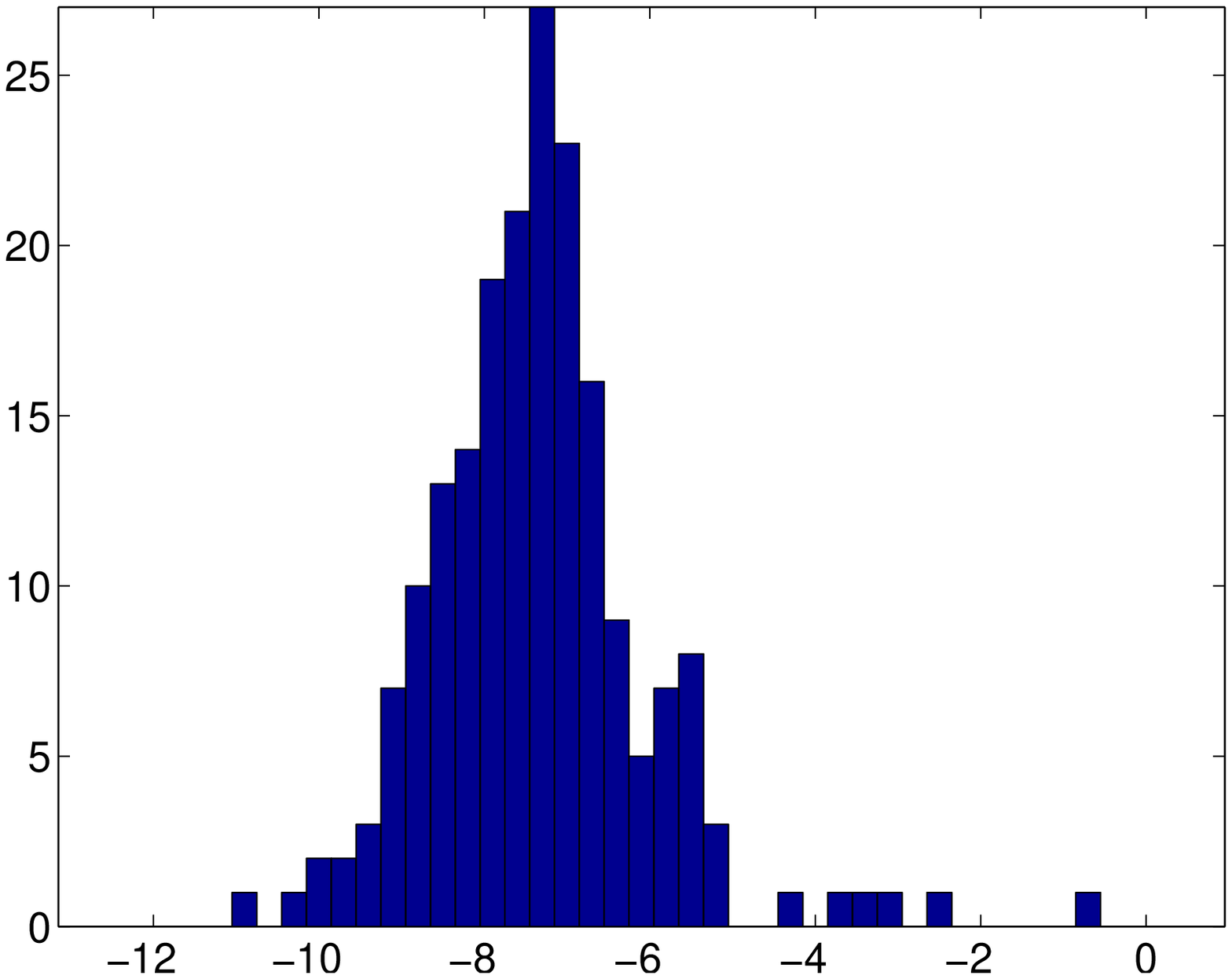}}
\subfigure[Scenario 2: $\overline{ANE}=2.3e-5$]{\includegraphics[width=0.32\columnwidth]{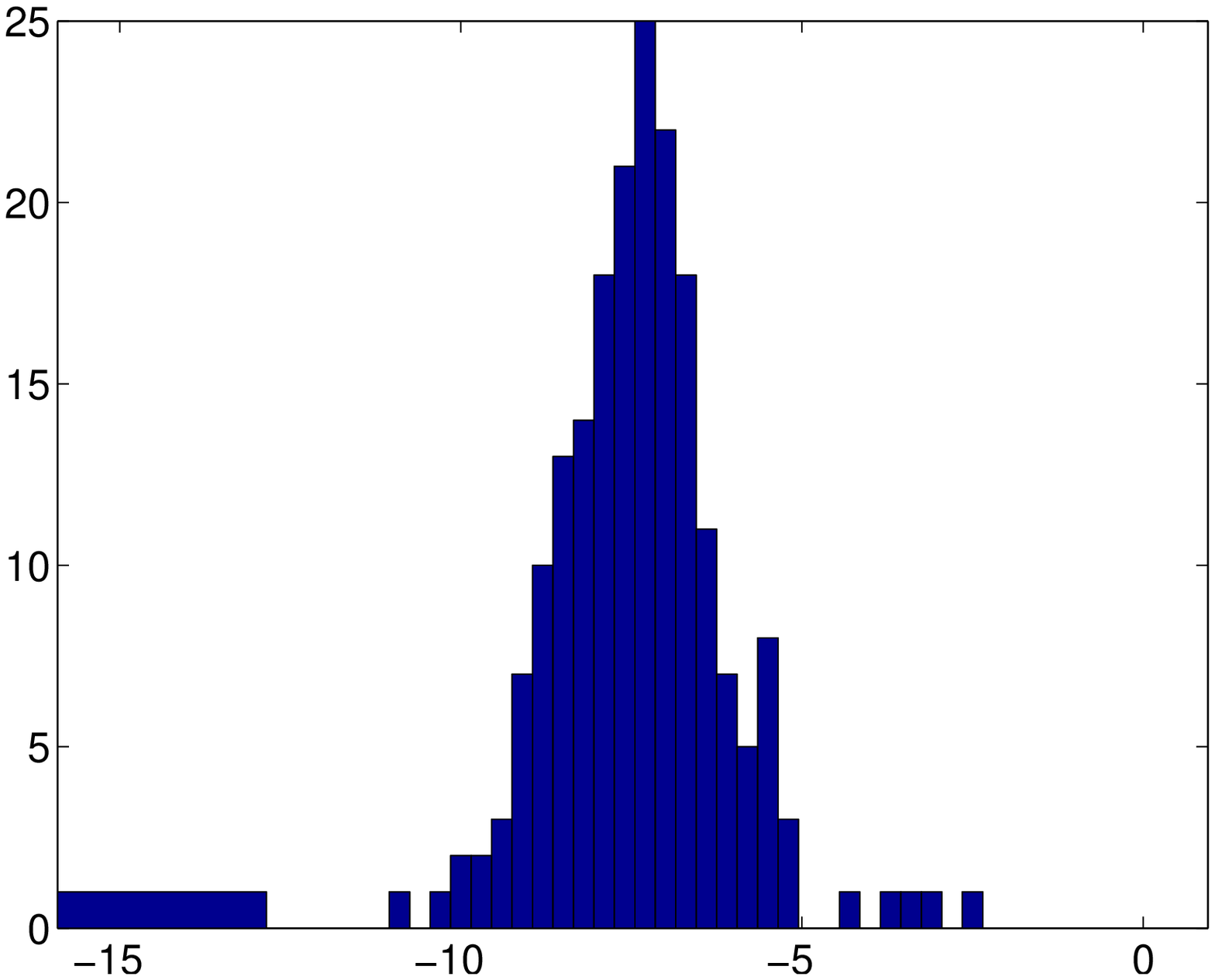}}
\subfigure[Scenario 3: $\overline{ANE}=7.2e-6$]{\includegraphics[width=0.32\columnwidth]{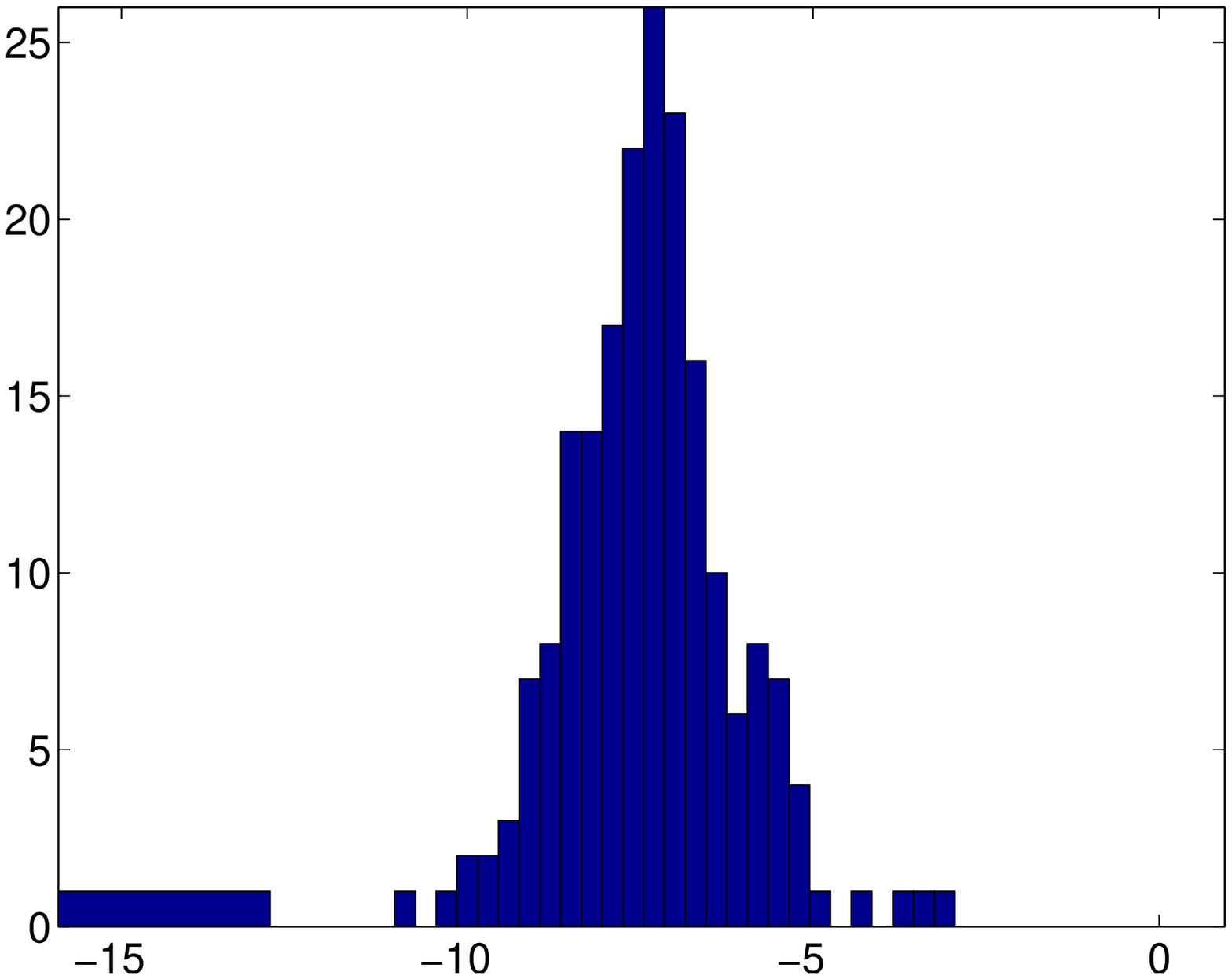}}
\end{center}
\caption{Histogram of reconstruction errors (measured in $ANE$) for the noiseless UNITCUBE graph with $n=212$ vertices, sensing radius $\rho = 0.3$ and average degree $deg = 17$. $\overline{ANE}$ denotes the average errors over all $N=197$ patches. Note that the x-axis shows the ANE in logarithmic scale. Scenario 1: directly embedding the 4-connected components. Scenario 2: embedding the WUL subgraphs extracted using Algorithm 1. Scenario 3: embedding the WUL subgraphs extracted using Algorithm 2. Note that for the subgraph embeddings we use FULL-SDP.}
\label{fig:wuls_unit}
\end{figure}

The results detailed in Figure \ref{fig:wuls_unit}, while showing improvements of the second and third scenarios over the first one, may not entirely convince the reader of the usefulness of our proposed randomized algorithm, since in the first scenario a direct embedding of the patches using FULL-SDP already gives a very good reconstruction, i.e. $8.4e-4$ on average. We regard $4$-connectivity a significant constraint that very likely renders a random geometric star graph to become globally rigid, thus diminishing the marginal improvements of the WUL extraction algorithm. To that end, we run experiments similar to those reported in Figure \ref{fig:wuls_unit}, but this time on the 1-hop neighborhood of each node in the UNITCUBE graph, without further extracting the 4-connected components. In addition, we sparsify the graph by reducing the sensing radius from $\rho=0.3$ to $\rho=0.26$. Table \ref{tab:ANE_UNIT_1_hop} shows the reconstruction errors, at various levels of noise. Note that in the noise free case, scenarios 2 and 3 yield results which are an order of magnitude better than that of scenario 1, which returns a rather poor average ANE of $5.3e-02$. However, for the noisy case, these marginal improvements are considerably smaller.


\begin{table}[h]
\begin{minipage}[b]{0.90\linewidth}
\centering
\begin{tabular}{|c||c|c|c|c|}
\hline
  $\eta$ & Scenario 1 & Scenario 2 &  Scenario 3  \\
\hline\hline
    0\% &  5.3e-02 &  4.9e-03 &  1.3e-03  \\
   10\% &  8.8e-02 &  5.2e-02 &  5.3e-02  \\
   20\% & 1.5e-01  & 1.1e-01  & 1.1e-01   \\
   30\% & 2.3e-01&   2.0e-01&  2.0e-01  \\
\hline
\end{tabular}
\caption{Average reconstruction errors (measured in $ANE$) for the UNITCUBE graph with $n=212$ vertices, sensing radius $\rho = 0.26$ and average degree $deg = 12$. Note that we only consider patches of size greater than or equal to 7, and there are 192 such patches. Scenario 1: directly embedding the 4-connected components. Scenario 2: embedding the WUL subgraphs extracted using Algorithm 1. Scenario 3: embedding the WUL subgraphs extracted using Algorithm 2. Note that for the subgraph embeddings we use FULL-SDP for noiseless data, and SNL-SDP for noisy data.  \;\;\;\;\;\;\;\;\; \;\;\;\;\;\;\;\;\;\;\;\;\;\;\;\;\;  \;\;\;\;\;\;\;\;
}
\label{tab:ANE_UNIT_1_hop}
\end{minipage}
\end{table}

Table \ref{tab:compA_lg_12} shows the total number of nodes removed from the patches by Algorithms 1 and 2, the number of 1-hop neighborhoods which are readily WUL, and the running times. Indeed, for the sparser UNITCUBE graph with $\rho=0.26$, the number of patches which are already WUL is almost half, compared to the case of the denser graph with $\rho=0.30$.

\begin{table}[h]
\begin{minipage}[b]{0.90\linewidth}
\centering
\begin{tabular}{|l|c|c||c|c|}
\hline
  &  \multicolumn{2}{c||}{ $\rho=0.30, N=197$} & \multicolumn{2}{c|}{$\rho=0.26, N=192$ }  \\
			       & Algorithm 1 & Algorithm 2  & Algorithm 1 & Algorithm 2  \\
\hline
Total nr of nodes removed			&  31        & 26     		&  258 &  285\\
Nr of WUL patches  	& 188 	& 191	&   104 &  101 \\
Running time (sec)	& 887 	& 48	&  632 & 26 \\
\hline
\end{tabular}
\caption{ Comparison of the two algorithms for extracting WUL subgraphs, for the UNICUBE graphs with sensing radius $ \rho=0.30$ and  $ \rho=0.26$, and noise level $\eta=0 \%$. The WUL patches are those patches for which the subgraph extraction algorithms did not remove any nodes.
}
\label{tab:compA_lg_12}
\end{minipage}
\end{table}

Finally, we remark on one of the consequences of our approach for breaking up the measurement graph. It is possible for a node not to be contained in any of the patches, even if it attaches in a globally rigid way to the rest of the measurement graph. An easy example is a star graph with four neighbors, no two of which are connected by an edge, as illustrated by the graph in Figure \ref{fig:ex_NSEW}. However, we expect such pathological examples to be very unlikely in the case of random geometric graphs.
\begin{figure}[h]
\begin{center}
\vspace{-4mm}
\includegraphics[width=0.48\columnwidth]{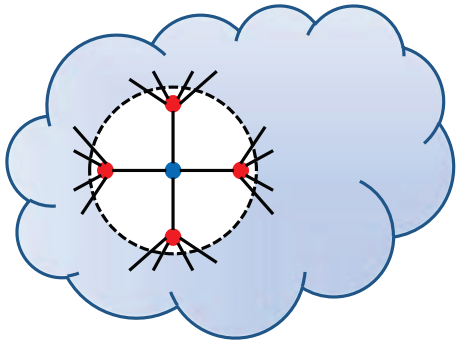}
\end{center}
\vspace{-2mm}
\caption{An example of a graph with a node that attaches globally rigidly to the rest of the graph, but is not contained in any patch, and thus it will be discarded by 3D-ASAP.}
\label{fig:ex_NSEW}
\vspace{-4mm}
\end{figure}

\subsection{Finding pseudo-anchors}
\label{sec:sub_pseudo}

To satisfy the conditions of Theorem \ref{SO_YE_THM_UL}, at least $d+1$ anchors are necessary for embedding a patch, hence for the molecule problem we need $k \geq 4$ such anchors in each patch. Since anchors are not usually available, one may ask whether it is still possible to find such a set of nodes that can be treated as anchors. If one were able to locate a clique of size at least $d+1$ inside a patch graph, then using cMDS it is possible to obtain accurate coordinates for the $d+1$ nodes and treat them as anchors
Whenever this is possible, we call such a set of nodes \textit{pseudo-anchors}. Intuitively, the geometric graph assumption  should lead one into thinking that if the patch graph is dense enough, it is very likely to find a complete subgraph on $d+1$ nodes. While a probabilistic analysis of random geometric graphs with forbidden  $K_{d+1}$ subgraphs is beyond of scope of this paper, we provide an intuitive connection with the problem of packing spheres inside a larger sphere, as well as numerical simulations that support the idea that a patch of size at least $\approx 10$ contains four such pseudo-anchors with some high probability.

To find pseudo-anchors for a given patch graph $G_i$, one needs to locate a complete subgraph (clique) containing at least $d+1$ vertices. Since any patch $G_i$ contains a center node that is connected to every other node in the patch, it suffices to find a clique of size at least $3$  in the 1-hop neighborhood of the center node, i.e., to find a triangle in  $G_i \backslash {i}$. Of course, if a graph is very dense (i.e., has high average degree) then it will be forced to contain such a triangle. To this end, we remind one of the first results in extremal graph theory (Mantel 1907), which states that any given graph on $s$ vertices and more than $ \frac{1}{4}s^2$ edges contains a triangle, the bipartite graph with $V_1 = V_2 = \frac{s}{2}$ being the unique extremal graph without a triangle and containing $ \frac{1}{4}s^2$ edges. However, this quadratic bound which holds for general graphs is very unsatisfactory for the case of random geometric graphs.

Recall that we are using the geometric graph model, where two vertices are adjacent if and only if they are less than distance $\rho$ apart. At a local level, one can think of the geometric graph model as placing an imaginary ball of radius $\rho$ centered at node $i$, and connecting $i$ to all nodes within this ball; and also connecting two neighbors $j,k$ of $i$ if and only if $j$ and $k$ are less than $\rho$ units apart. Ignoring the center node $i$, the question to ask becomes how many nodes can one fit into a ball of radius $\rho$ such that there exist at least $d$ nodes whose pairwise distances are all less than $\rho$. In other words, given a geometric graph $H$ inscribed in a sphere of radius $\rho$, what is the smallest number of nodes of $H$ that forces the existence of a $K_d$.

The astute reader might immediately be lead into thinking that the problem above can be formulated as a sphere packing problem. Denote by $x_1, x_2, \ldots x_m$ the set of $m$ nodes (ignoring the center node) contained in a sphere of radius $\rho$. We would like to know what is the smallest $m$ such that at least $d=3$ nodes are pairwise adjacent, i.e. their pairwise distances are all less than $\rho$.

To any node $x_i$ associate a smaller sphere $S_i$ of radius $\frac{\rho}{2}$. Two nodes $x_i,x_j$ are adjacent, meaning less than distance $\rho$ apart, if and only if their corresponding spheres $S_i$ and $S_j$  overlap. This line of thought leads one into thinking how many non-overlapping small spheres can one pack into a larger sphere. One detail not to be overlooked is that the radius of the larger sphere should be $\frac{3}{2} \rho$, and not $\rho$, since a node $x_i$ at distance $\rho$ from the center of the sphere has its corresponding sphere $S_i$ contained in a sphere of radius $\frac{3}{2} \rho$. We have thus reduced the problem of asking what is the minimum size of a patch that would guarantee the existence of four anchors, to the problem of determining the smallest number of spheres of radius $\frac{1}{2}\rho$ that can be ``packed" in a sphere of radius $\frac{3}{2} \rho$ such that at least three of the smaller spheres pairwise overlap. Rescaling the radii such that $\frac{3}{2} \rho =1$ (hence $\frac{1}{2} \rho =\frac{1}{3}$), we ask the equivalent problem: \textit{How many spheres of radius $\frac{1}{3}$ can be packed inside a sphere of radius $1$, such that at least three spheres pairwise overlap.}

A related and slightly simpler problem is that of finding the densest packing on $m$ equal spheres of radius $r$  in a sphere of radius 1, such that no two of the small spheres overlap. This problem has been recently considered in more depth, and analytical solutions have been obtained for several values of $m$. If $r=\frac{1}{3}$ (as in our case) then the answer is  $m=13$  and this constitutes a lower bound for our problem.

However, the arrangements of spheres that prevent the existence of three pairwise overlapping spheres are far from random, and motivated us to running the following experiment. For a given $m$, we generate $m$ random spheres of radius $\frac{1}{3}$ inside the unit sphere, and count the number of times when at least three spheres pairwise overlap. We ran this experiment $15,000$ times for different values of $m=5,6,7$, respectively $8$, and obtained the following success rates $69\%, 87\%, 96\%$, respectively $99\% $, i.e., the percentage of times when the random realizations of spheres of radius $\frac{1}{3}$ inside a unit sphere produced three pairwise overlapping spheres. The simulation results show that about $9$ spheres would guarantee, with very high probability, the existence of three pairwise overlapping spheres. In other words, for a patch of size $10$ including the center node, there exist with high probability at least $4$ nodes that are pairwise adjacent, i.e., the four pseudo-anchors we are looking for.

\subsection{Embedding patches}
\label{sec:sub_emb}
After extracting patches, i.e., WUL subgraphs of the 1-hop neighborhoods, it still remains to localize each patch in its own frame.
Under the assumptions of the geometric graph model, it is likely that 1-hop neighbors of the central node will also be interconnected, rendering a relatively high density of edges for the patches. Indeed, as indicated by Figure \ref{fig:patch_info_1} (right panel), most patches have at least half of the edges present.
For noiseless distances, we embed the patches using the FULL-SDP algorithm \cite{BiswasYe}, while for noisy distances we use the SNL-SDP algorithm of \cite{snlsdp}. To improve the overall localization result, the SDP solution is used as a starting point for a gradient-descent method.


\begin{figure}[h]
\begin{center}
\includegraphics[width=0.49\columnwidth]{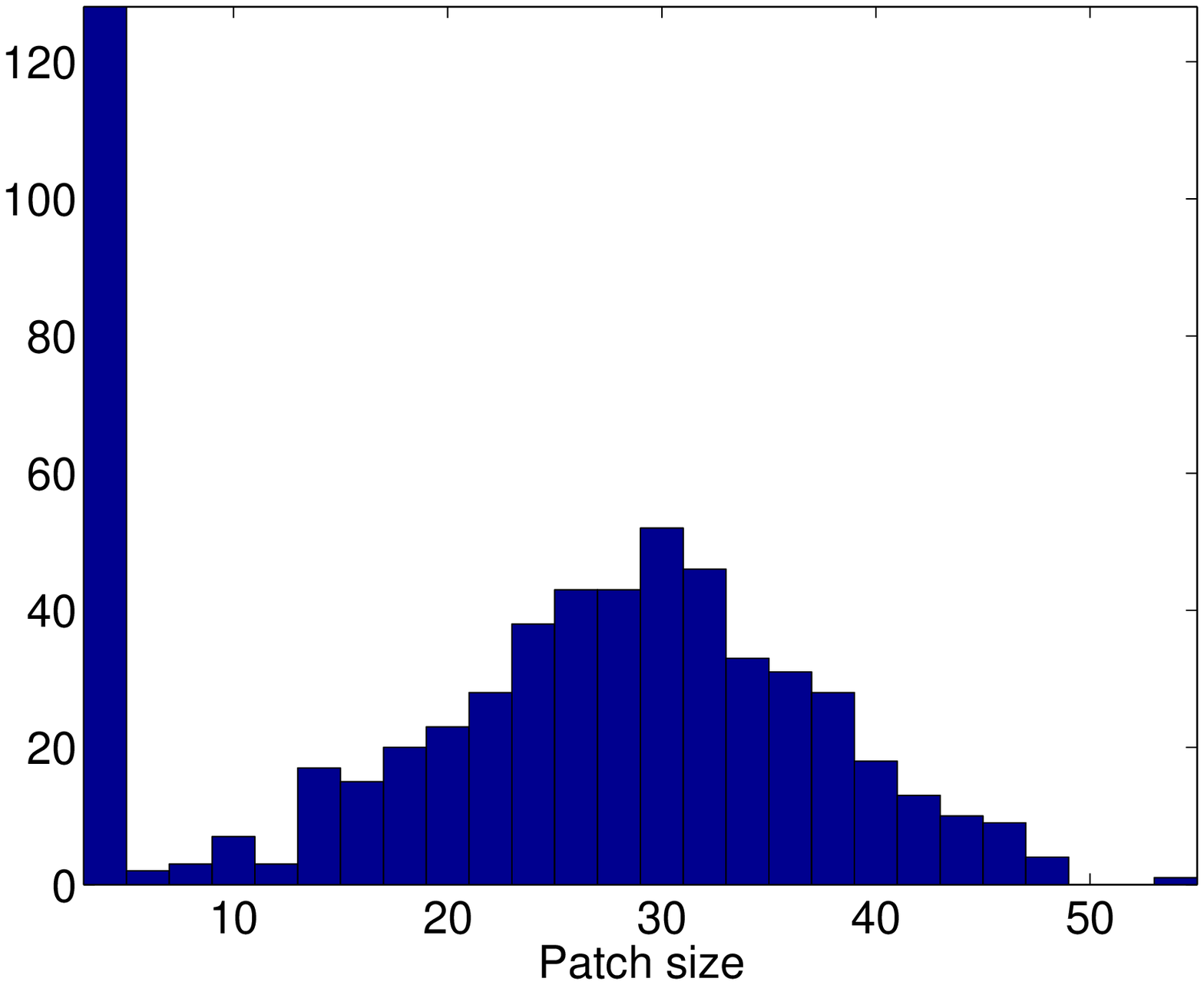}
\includegraphics[width=0.49\columnwidth]{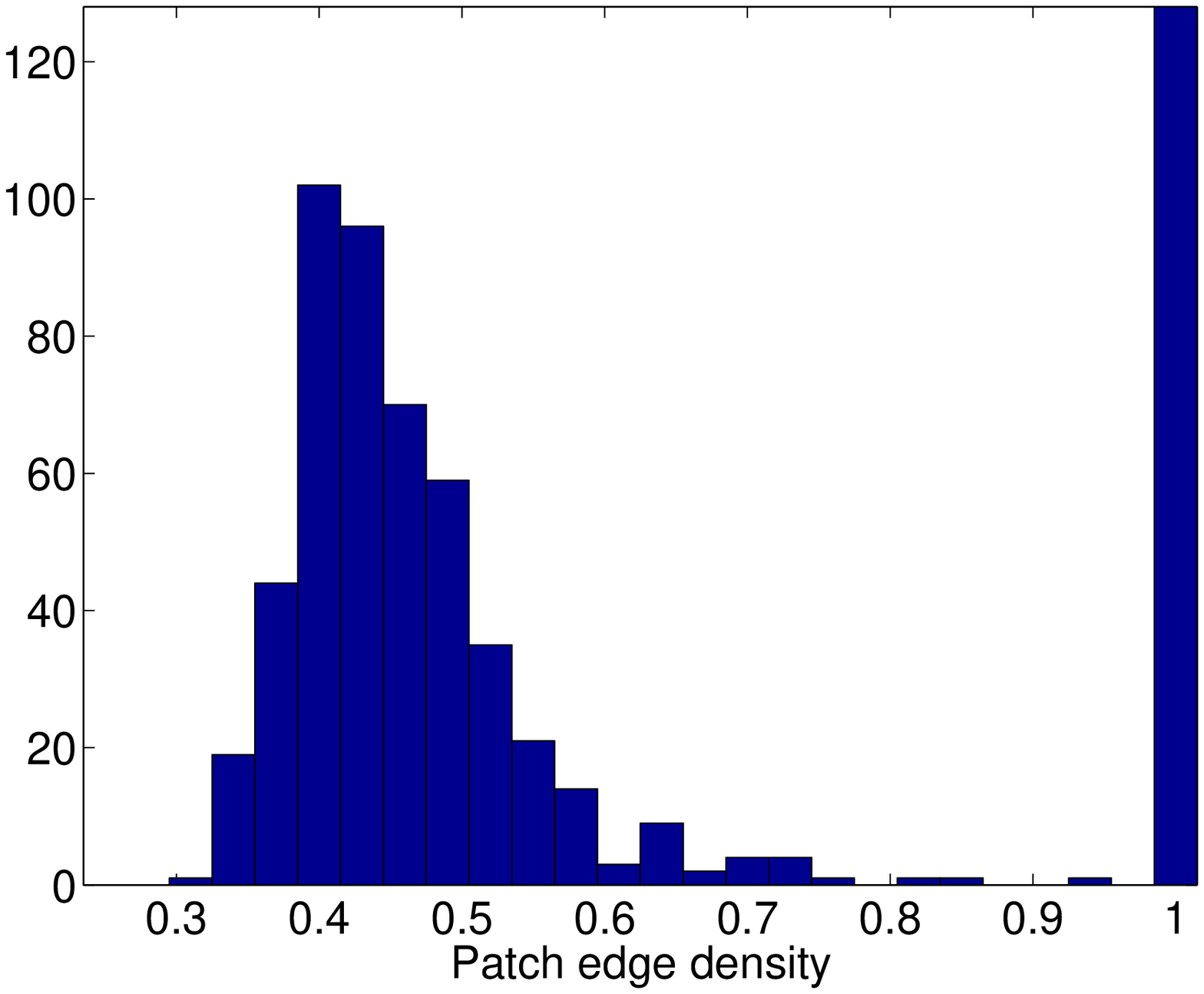}
\end{center}
\caption{Histogram of patch sizes (left) and edge density (right) in the BRIDGE-DONUT graph, $n=500$ and $deg = 18$. Note that a large number of the resulting patches are of size $4$, thus complete graphs on four nodes ($K_4$), which explains the same large number of patches with edge density $1$.}
\label{fig:patch_info_1}
\end{figure}

The remaining part of this subsection is a brief survey of recent SDP relaxations for the graph localization problem
\cite{BiswasYe,sdpAngle,biswas_stress_sdp,biswas,univRig}. A solution
$p_1, \ldots , p_n \in \mathbb{R}^3$ can be computed by minimizing the following error function
\begin{equation}
 \min_{p_1,\ldots,p_n\in \mathbb{R}^3 } \sum_{(i,j)\in E}  \left( \| p_i - p_j \|^2 - d_{ij}^2 \right)^2.
\label{E1}
\end{equation}
While the above objective function is not convex  over the constraint set, it can be relaxed into an SDP \cite{biswas_stress_sdp}. Although SDP can be generally solved (up to a given accuracy) in polynomial time, it was pointed out in \cite{biswas} that the objective function (\ref{E1}) leads to a rather expensive SDP, because it involves fourth order polynomials of the coordinates. Additionally, this approach is rather sensitive to noise, because large errors are amplified by the objective function in (\ref{E1}), compared to the objective function in (\ref{E2}) discussed below.

Instead of using the objective function in (\ref{E1}), \cite{biswas} considers the SDP relaxation of the following penalty function
\begin{equation}
   \min_{p_1,\ldots,p_n \in \mathbb{R}^3}  \sum_{(i,j) \in E} \left| \parallel p_i - p_j  \parallel^2 - d_{ij}^2  \right|.
\label{E2}
\end{equation}
In fact, \cite{biswas} also allows for possible non-equal weighting of the summands in (\ref{E2})  and for possible anchor points. The SDP relaxation of (\ref{E2}) is faster to solve than the relaxation of (\ref{E1}) and it is usually more robust to noise. Constraining the solution to be in $\mathbb{R}^3$ is non-convex, and its relaxation by the SDP often leads to solutions that belong to a higher dimensional Euclidean space, and thus need to be further projected to $\mathbb{R}^3$. This projection often results in large errors for the estimation of the coordinates. A regularization term for the objective function of the SDP was suggested in \cite{biswas} to assist it in finding solutions of lower dimensionality and preventing nodes from crowding together towards the center of the configuration.

\subsection{ Additional Information Specific to the Molecule Problem }
\label{sec:sub_infomol}

In this section we discuss several additional constraints specific to the molecule problem, which are currently not being exploited by 3D-ASAP. While our algorithm can benefit from any existing molecular fragments and their known reflection, there is still information that it does not take advantage of, and which can further improve its performance. Note that many of the remarks below can be incorporated in the pre-processing step of embedding the patches, described in the previous section.

The most important piece of information missing from our 3D-ASAP formulation is the distinction between the ``good" edges (bond lengths) and the ``bad" edges (noisy NOEs). The current implementations of the FULL-SDP and SNL-SDP algorithms do not incorporate such hard distance constraints.

One other important information which we are ignoring is given by the residual dipolar couplings (RDC) measurements that give noisy angle information ($cos^2(\theta)$) with respect to a global orientation \cite{Bax}.

Another approach is to consider an energy based formulation that captures the interaction between atoms in a readily computable fashion, such as the Lennard-Jones potential. One may then use this information to better localize the patches, and prefer patches that have lower energy.

The minimum distance constraint, also referred to as the ``hard sphere" constraint, comes from the fact that any two atoms cannot get closer than a certain distance $\kappa \approx 1$ Angstrom. Note that such lower bounds on distances can be easily incorporated into the SDP formulation.

Another observation one can make use of is set of non-edges of the measurement graph, i.e., the distances corresponding to the missing edges cannot be smaller than the sensing radius $\rho$. Two  remarks are in place however; under the current noise model it is possible for true distances smaller than the sensing radius not to be part of the set of available measurements, and vice-versa, it is possible for true distances larger than the sensing radius to become part of the distance set. However, since this constraint is not as certain as the hard sphere constraint, we recommend using the latter one.


Finally, one can envisage that significant other information can be reduced to distance constraints and incorporated into the approach described here for the calculation of structures and complexes. Such development could significantly speed such calculations if it incorporates larger molecular fragments based on modeling, similarly of chemical shift data etc., as done with computationally intensive experimental energy methods, e.g., HADDOCK \cite{HADDOCK}.


\subsection{Aligning patches}
\label{sec:sub_align}
Given two patches $P_k$ and $P_l$ that have at least four nodes in common, the registration process finds the optimal 3D rigid motion of $P_l$ that aligns the common points (as shown in Figure \ref{fig:registration_example}).
A closed form solution to the registration problem in any dimension was given in \cite{Horn}, where the best rigid transformation between two sets of points is obtained by various matrix manipulations and eigenvalue/eigenvector decomposition.

\begin{figure}[h]
\begin{center}
\includegraphics[width=1.0\columnwidth,keepaspectratio=true]{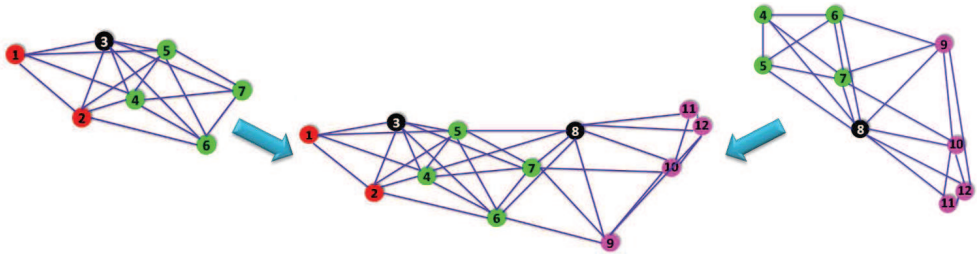}
\end{center}
\caption{Optimal alignment of two patches that overlap in four nodes. The alignment provides a measurement for the ratio of the two group elements in Euc(3). In this example we see that a reflection was required to properly align the patches.}
\label{fig:registration_example}
\end{figure}

Since alignment requires at least four overlapping nodes, the K4 patches that are fully contained in larger patches are initially discarded. Other patches may also be discarded if they do not intersect any other patch in at least four nodes. The nodes belonging to such patches but not to any other patch would not be localized by ASAP.

As expected, in the case of the geometric graph model, the overlap is often small, especially for small values of $\rho$. It is therefore crucial to have robust alignment methods even when the overlap size is small. We refer the reader to Section 6 of \cite{ASAP} for other methods of aligning patches with fewer common nodes in $\mathbb{R}^2$, i.e. the combinatorial method and the link method which can be adjusted for the three dimensional case. The combinatorial score method makes use of the underlying assumption of the geometric graph model. Specifically, we exploit the information in the non-edges that correspond to distances larger than the sensing radius $\rho$, and use this information for estimating both the relative reflection and rotation for a pair of patches that overlap in just three nodes (or more). The link method is useful whenever two patches have a small overlap, but there exist many cross edges in the measurement graph that connect the two patches. Suppose the two patches $P_k$ and $P_l$ overlap in at least one vertex, and call a \textit{link edge} an edge $(u,v) \in E$ that connects a vertex $u$ in patch $P_k$ (but not in $P_l$) with a vertex $v$ in patch $P_l$ (but not in $P_k$). Such link edges can be incorporated as additional information (besides the common nodes) into the registration problem that finds the best alignment between a pair of patches. The right plot in Figure \ref{fig:patch_info_2} shows a histogram of the intersection sizes between patches in the BRIDGE-DONUT graph that overlap in at least $4$ nodes.

\begin{figure}[h]
\begin{center}
\includegraphics[width=0.49\columnwidth]{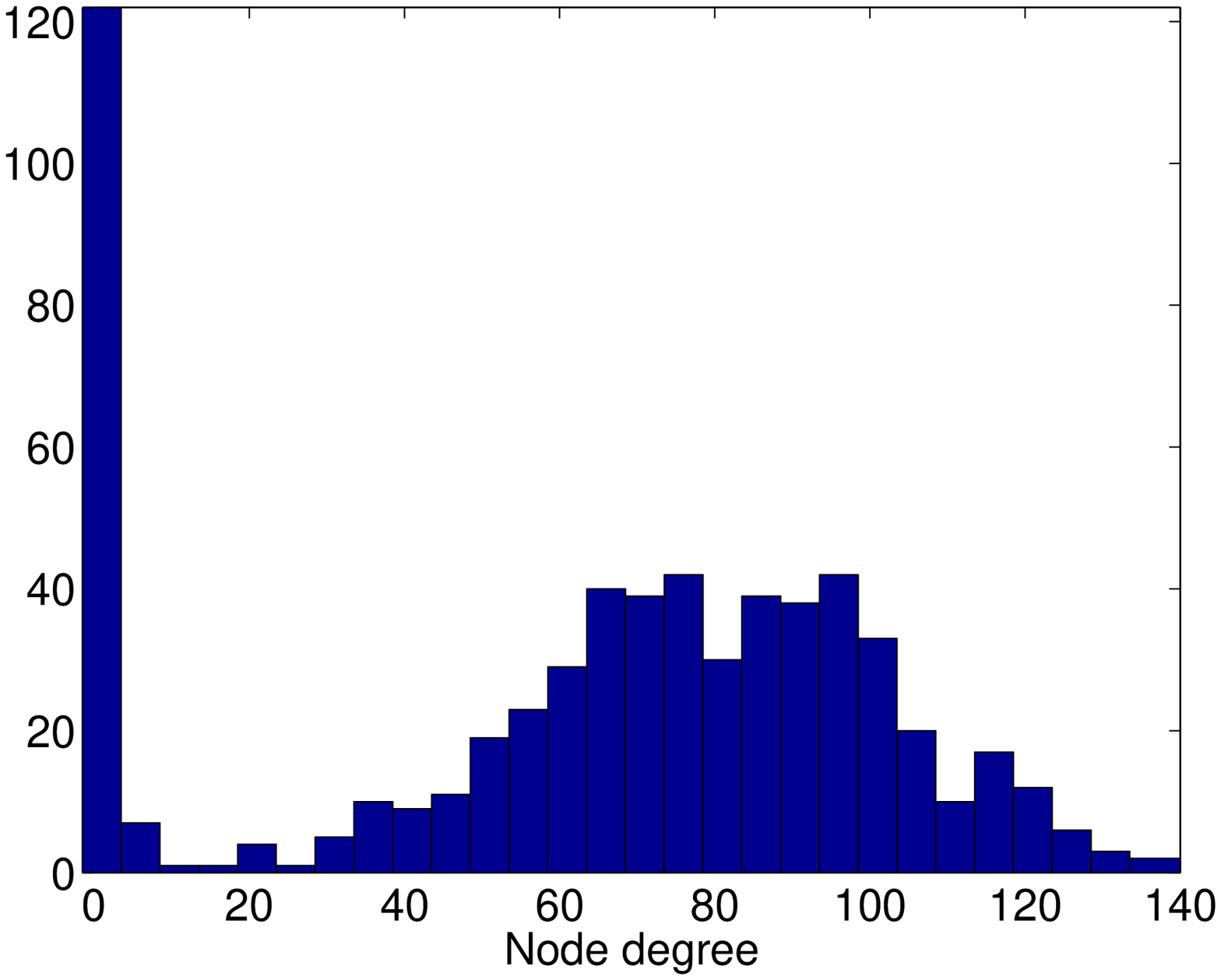}
\includegraphics[width=0.49\columnwidth]{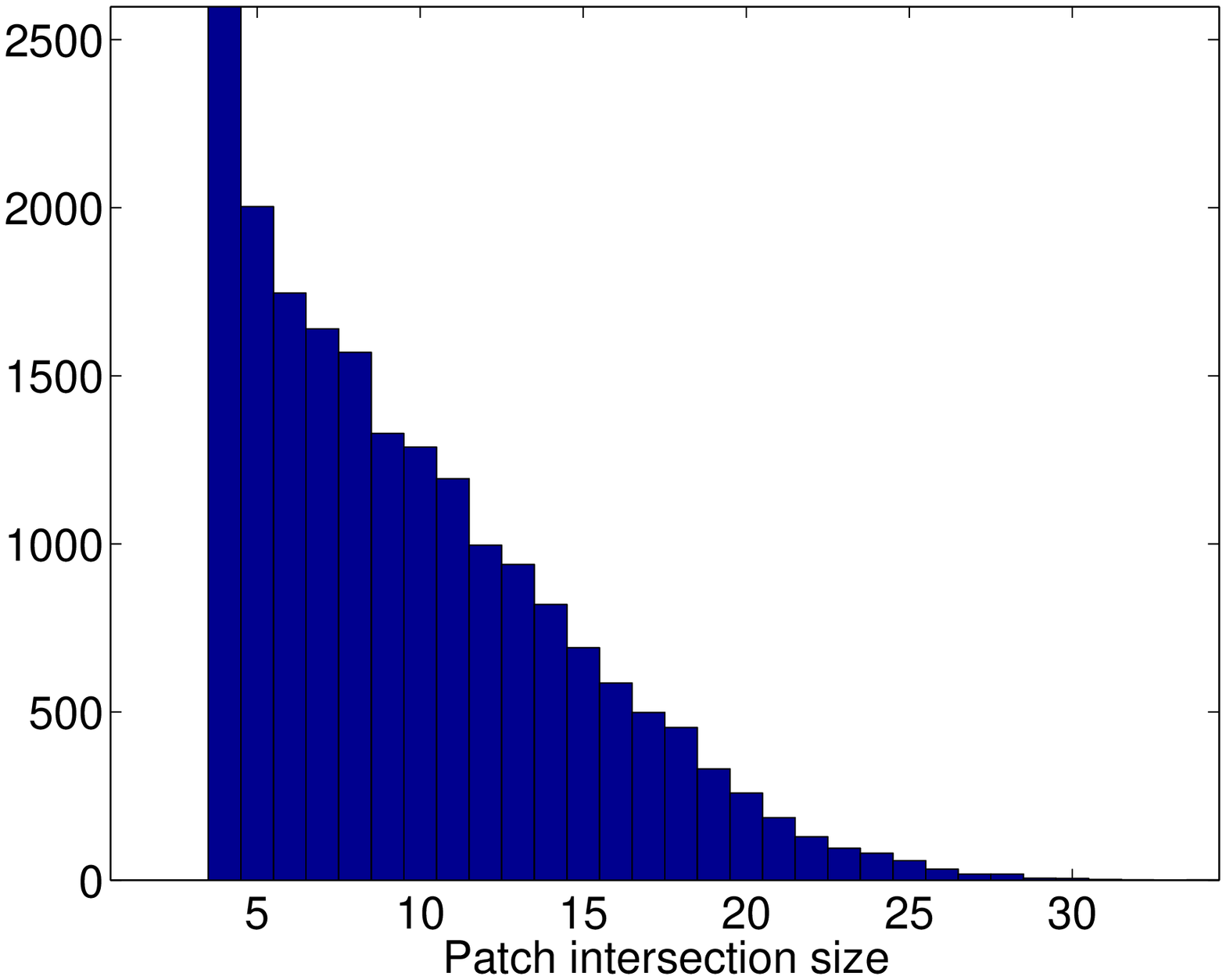}
\end{center}
\caption{Histogram of the node degrees of patches in the patch graph $G^P$ (left) and the intersection size of patches (right), in the BRIDGE-DONUT graph with $n=500$ and $deg = 18$. $G^P$ has $N=615$ nodes (i.e. patches) and average degree $24$, meaning that, on average, a patch overlaps (in at least 4 nodes) with $24$ other patches.}
\label{fig:patch_info_2}
\end{figure}

\section{Spectral-Partitioning-ASAP (3D-SP-ASAP)}
\label{sec:sp-asap}

In this section we introduce 3D-Spectral-Partitioning-ASAP (3D-SP-ASAP), a variation of the 3D-ASAP algorithm, which uses spectral partitioning as a preprocessing step for the localization process.

3D-SP-ASAP combines ideas from both DISCO \cite{DISCO} and ASAP.
The philosophy behind DISCO is to recursively divide large problems into smaller problems, which can ultimately be solved by the traditional SDP-based localization methods. If the number of atoms in the current group is not too large, DISCO solves the atom positions via SDP, and refines the coordinates by using gradient descent;  otherwise, it breaks the current group of atoms into smaller subgroups, solves each subgroup recursively, aligns and combines them together, and finally it improves the coordinates by applying gradient descent. The main question that arises is how to divide a given problem into smaller subproblems.  DISCO chooses to divide a large group of nodes into exactly two subproblems, solves each problem recursively and combines the two solutions. In other words, it builds a binary tree of problems, where the leaves are problems small enough to be embedded by SDP. However, not all available information is being used when considering only a single spanning tree of the graph of patches.  The 3D-ASAP approach fuses information from different spanning trees via the eigenvector computation. However, compared to the number of patches used in DISCO, 3D-ASAP generates many more patches, since the number of patches in 3D-ASAP is linear in the size of the network. This can be considered as a disadvantage, since localizing all the patches is often the most time consuming step of the algorithm. 3D-SP-ASAP tries to reduce the number of patches to be localized while using the patch graph connectivity in its full.

When dividing a graph into two smaller subgraphs, one wishes to minimize the number of edges between the two subgraphs, since in the localization of the two subgraphs the edges across are being left out. Simultaneously, one wishes to maximize the number of edges within the subgraphs, because this makes the subgraphs more likely to be globally rigid and easier to localize. In general, the graph partitioning problem seeks to decompose a graph into $K$ disjoint subgraphs (clusters), while minimizing the number of cut edges, i.e., edges with endpoints in different clusters.  Given the number of clusters $K$, the Min-Cut problem is an optimization problem that computes a partition $\mathcal{P}_1, \ldots ,\mathcal{P}_K$ of the vertex set, by minimizing the number of cut edges
\begin{equation}
 \mbox{Cut}(\mathcal{P}_1, \ldots ,\mathcal{P}_K) = \sum_{i=1}^{K} E(\mathcal{P}_i, \overline{\mathcal{P}_i}),
\end{equation}
where $E(X,Y) = \sum_{i \in X, j \in Y} A_{ij}$, and $\overline{X}$ denotes the complement of $X$. However, it is well known that trying to minimize $\mbox{Cut}(\mathcal{P}_1, \ldots ,\mathcal{P}_K)$ favors cutting off weakly connected individual vertices from the graph, which leads to poor partitioning quality since we would like clusters to consist of a relatively large number of nodes. To penalize clusters $\mathcal{P}_i$  of small size, Shi and Malik \cite{shimalik} suggested minimizing the normalized cut defined as
\begin{equation}
 \mbox{NCut}( \mathcal{P}_1, \ldots ,\mathcal{P}_K ) = \sum_{i=1}^K \frac{\mbox{Cut}(\mathcal{P}_i,\overline{\mathcal{P}_i})}{\mbox{Vol}(\mathcal{P}_i)},
\end{equation}
where
 $\mbox{Vol}(P_i) = \sum_{i \in P_i} deg(i)$,
and $deg(i)$ denotes the degree of node $i$ in $G$.

Although minimizing NCut over all possible partitions of the vertex set $V$ is an NP-hard combinatorial optimization problem  \cite{wagner}, there exists a spectral relaxation that can be computed efficiently \cite{shimalik}. We use this spectral clustering method to partition the measurement graph in the molecule problem. The gist of the approach is to use the classical K-means clustering algorithm on the Laplacian eigenmap embedding of the set of nodes. If $A$ is the adjacency matrix of the graph $G$, and $D$ is a diagonal matrix with $D_{i,i} =  deg(i), i = 1,\ldots,n$, then the Laplacian eigenmap embedding of node $i$ in $\mathbb{R}^k$ is given by $ (\phi_1(i),\phi_2(i),\ldots,\phi_K(i))$, where $\phi_j$ is the $j^{th}$ eigenvector of the matrix $D^{-1} A$. For an extensive literature survey on spectral clustering algorithms we refer the reader to \cite{luxburg}. We remark that other clustering algorithms (e.g., \cite{Hespanha}) may also be used to partition the graph.

The approach we used for localization in conjunction with the above normalized spectral clustering algorithm is as follows (3D-SP-ASAP):
\begin{enumerate}
 \item  We first decompose the measurement graph into $K$ partitions $\mathcal{P}_1, \ldots ,\mathcal{P}_K$, using the normalized spectral clustering algorithm.
 \item  We extend each partition $\mathcal{P}_i$, $i=1,\ldots,K$ to include its 1-hop neighborhood, and denote the new patches by $ P_i$, $i=1,\ldots,K$.
 \item  For every pair of patches $P_i$ and $P_j$ which have nodes in common or are connected by link edges \footnote{edges with endpoints in different patches}, we build a new (link) patch which contains all the common points and link edges.  The vertex set of the new patch consists of the nodes that are common to both $P_i$ and  $P_j$, together with the endpoints of the link edges that span across the two patches. Note that the new list of patches contains the extended patches built in Step (2) $P_1, \ldots ,P_K$, as well as the newly built patches $P_{K+1},\ldots,P_{L}$.
 \item  We extract from each patch $P_i$, $i=1,\ldots,L$ the WUL subgraph, and embed it using the FULL-SDP algorithm for noiseless data, and the SNL-SDP algorithm for noisy data.
 \item  Synchronize all available patches using the eigenvector synchronization algorithm used in ASAP.
\end{enumerate}

Note that when the extended patches from Step (2) are highly overlapping, Step (3) of the algorithm  should be omitted, for reasons detailed below related to the robustness of the embedding. The reason for having Steps (2), and possibly (3), is to be able to align nearby patches. Without Step (2) patches will have disjoint sets of nodes, and the alignment will be based only on the link edges, which is not robust for high levels of noise. Without Step (3) the existing patches may have little overlap, in which case  we expect the alignment involving link edges not to be robust at high levels of noise. By building the link patches, we provide 3D-SP-ASAP more accurate pairwise alignments. Note that embedding link patches is less robust to noise, due to their bipartite-like structure, especially when the bipartitions are very loosely connected to each other (also confirmed by our computations involving link patches). In addition, having to localize a larger number of patches may significantly increase the running time of the algorithm. However, for our numerical experiments with 3D-SP-ASAP conducted on the BRIDGE-DONUT graph, the extended partitions built in Step (2) were highly overlapping, and allowed us to localize the entire network without the need to build the link patches in Step (3).

The advantage of combining a spectral partitioning algorithm with 3D-ASAP is a decrease in running time as shown in Tables \ref{tab:SP_times} and \ref{tab:times}, due to a significantly smaller number of patches that need to be localized. Note that the graph partitioning algorithm is extremely fast, and partitions the BRIDGE-DONUT measurement graph in less than half a second.  Table \ref{tab:ANE_DONUT} shows the ANE reconstruction errors for the BRIDGE-DONUT graph, when we partition the measurement graph into $K=8$ and $K=25$ clusters. For $K=8$, some of the extended partitions become very large, containing as many as $150$ nodes, and SNL-SDP does a very poor job at embedding such large patches when the distance measurements are noisy. By increasing the number of partitions to $K=25$, the extended partitions contain less nodes, and are small enough for SNL-SDP to localize accurately even for high levels of noise. Note that the ANE errors for 3D-SP-ASAP with $K=25$ are comparable with those of ASAP, while the running time is dramatically reduced (by an order of magnitude, for the BRIDGE-DONUT example with $\eta=35\%$). Figures  \ref{fig:spectral_part_PACM} and \ref{fig:spectral_part_DONUT} show various $K$-partitions of the PACM and the BRIDGE-DONUT graphs.


\begin{figure}[h]
\begin{center}
\subfigure[K=2]{\includegraphics[width=0.24\columnwidth]{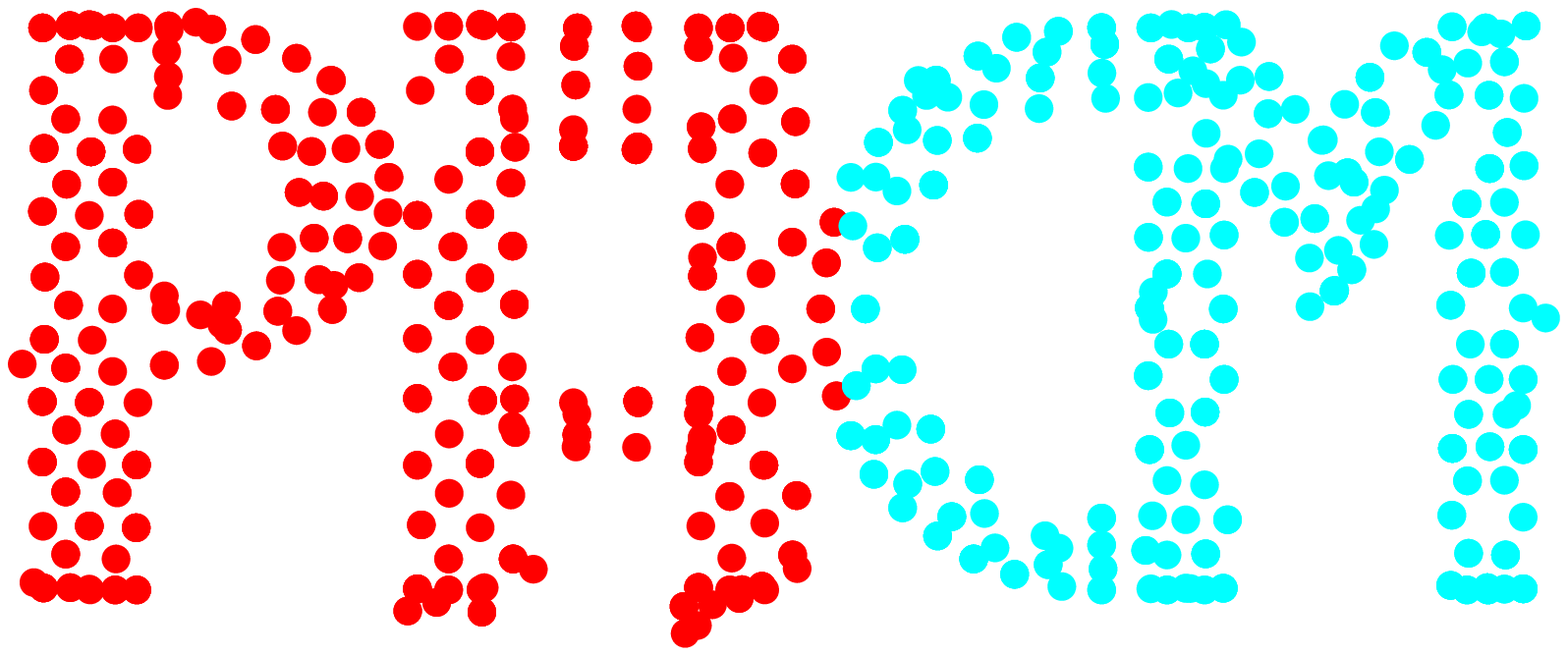}}
\subfigure[K=3]{\includegraphics[width=0.24\columnwidth]{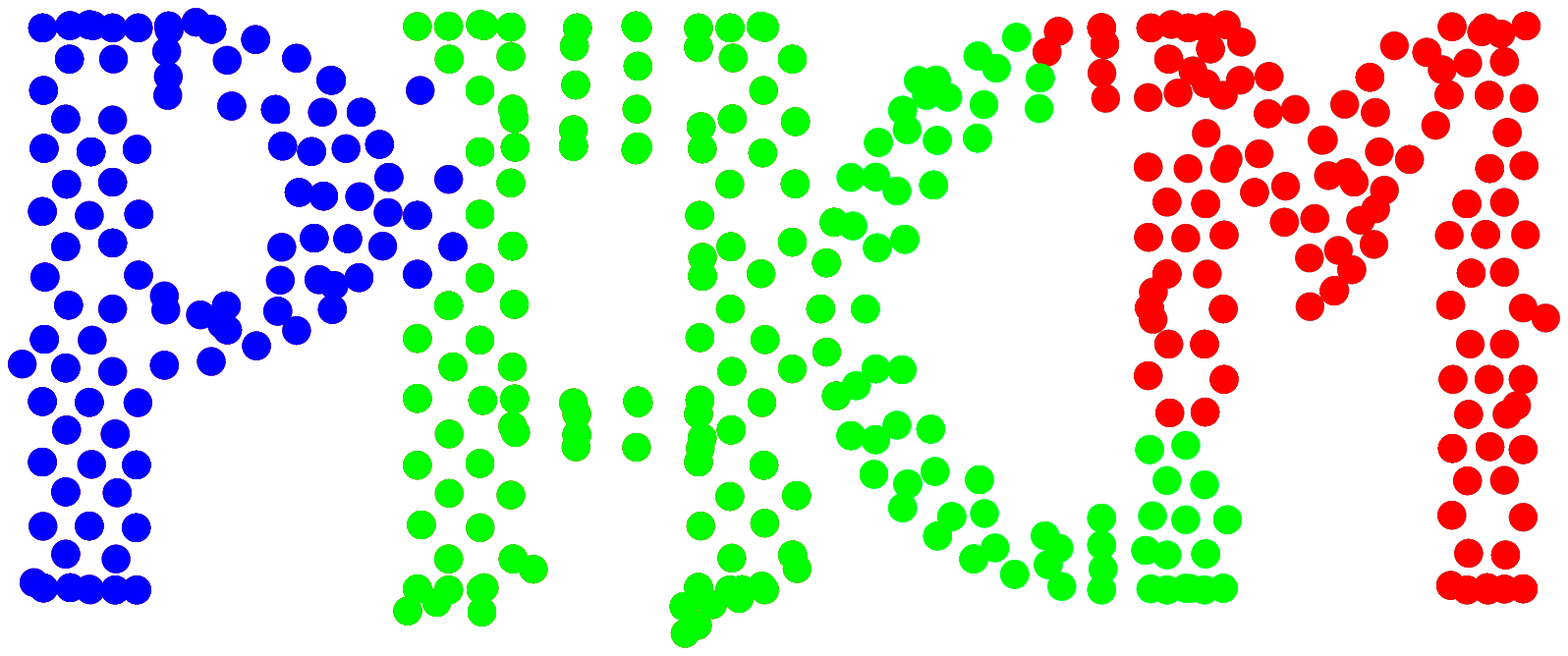}}
\subfigure[K=4]{\includegraphics[width=0.24\columnwidth]{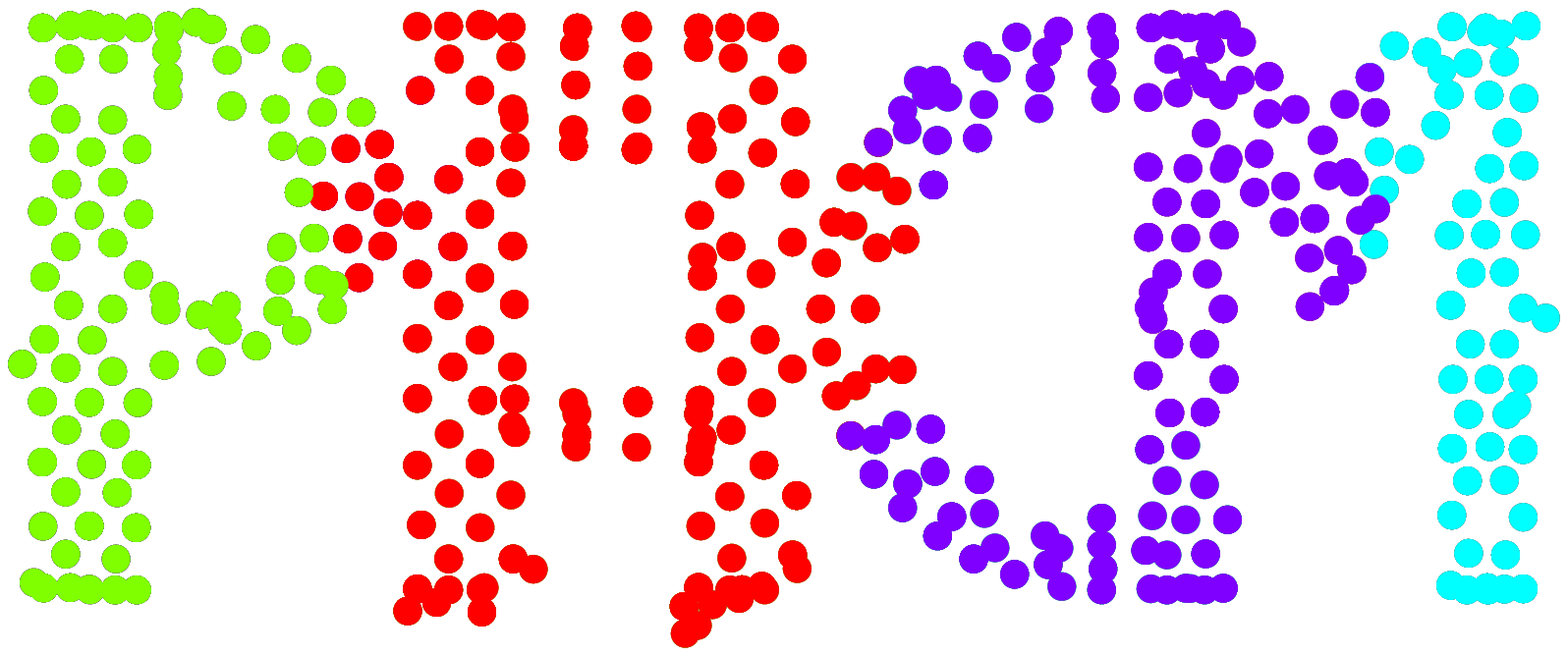}}
\subfigure[K=5]{\includegraphics[width=0.24\columnwidth]{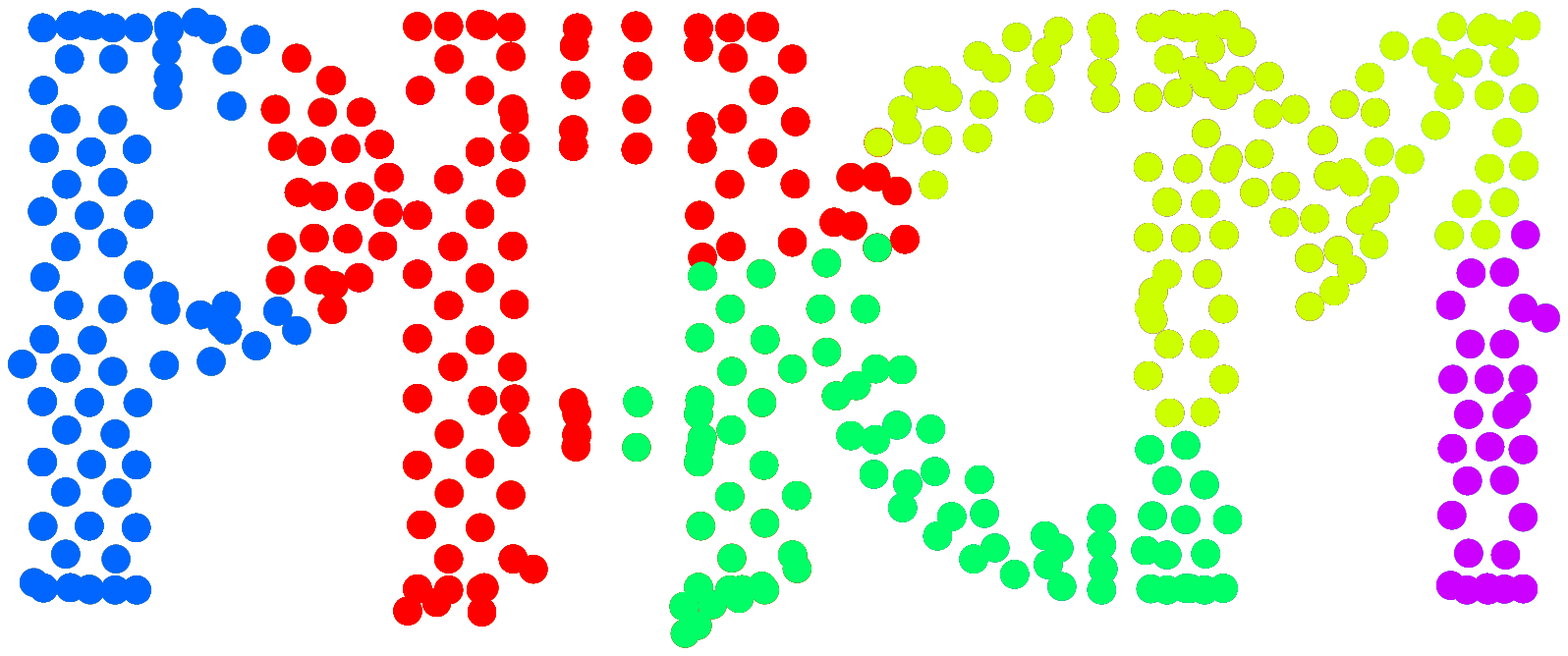}}
\subfigure[K=6]{\includegraphics[width=0.24\columnwidth]{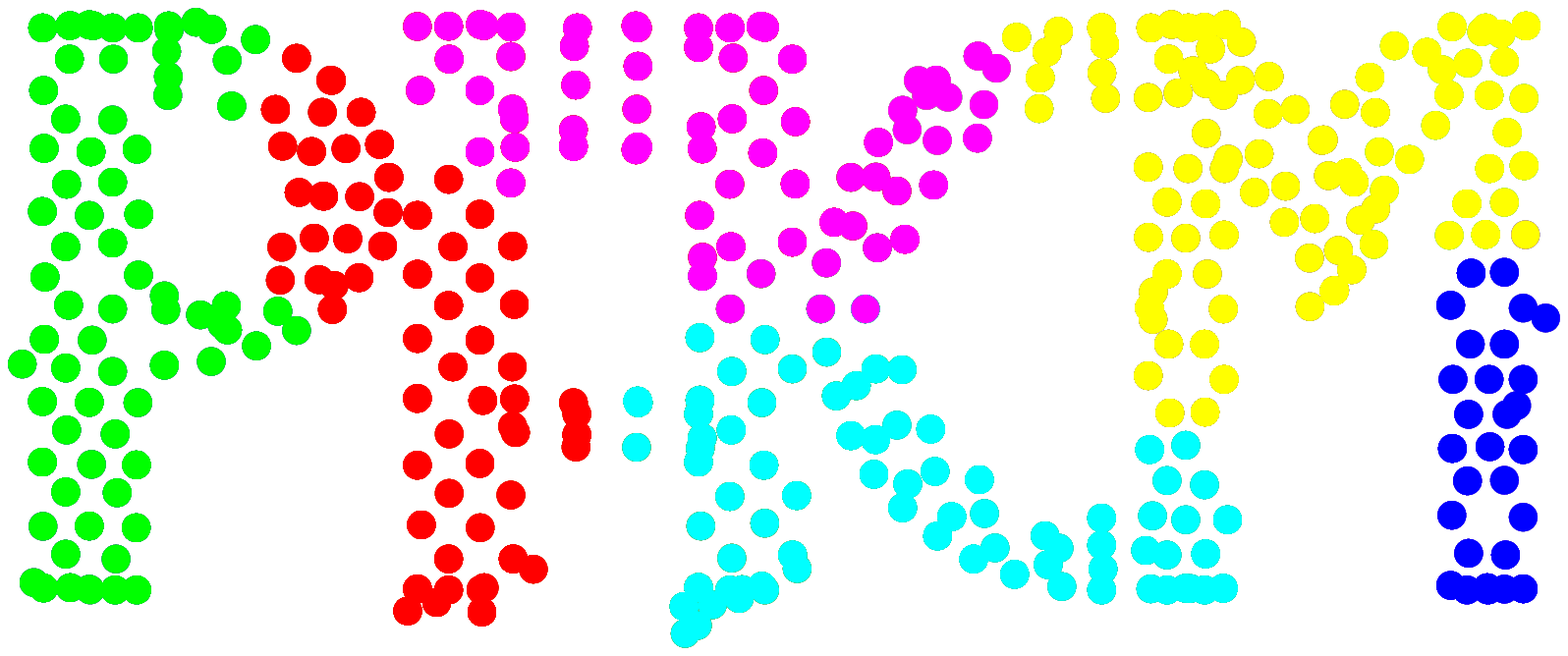}}
\subfigure[K=7]{\includegraphics[width=0.24\columnwidth]{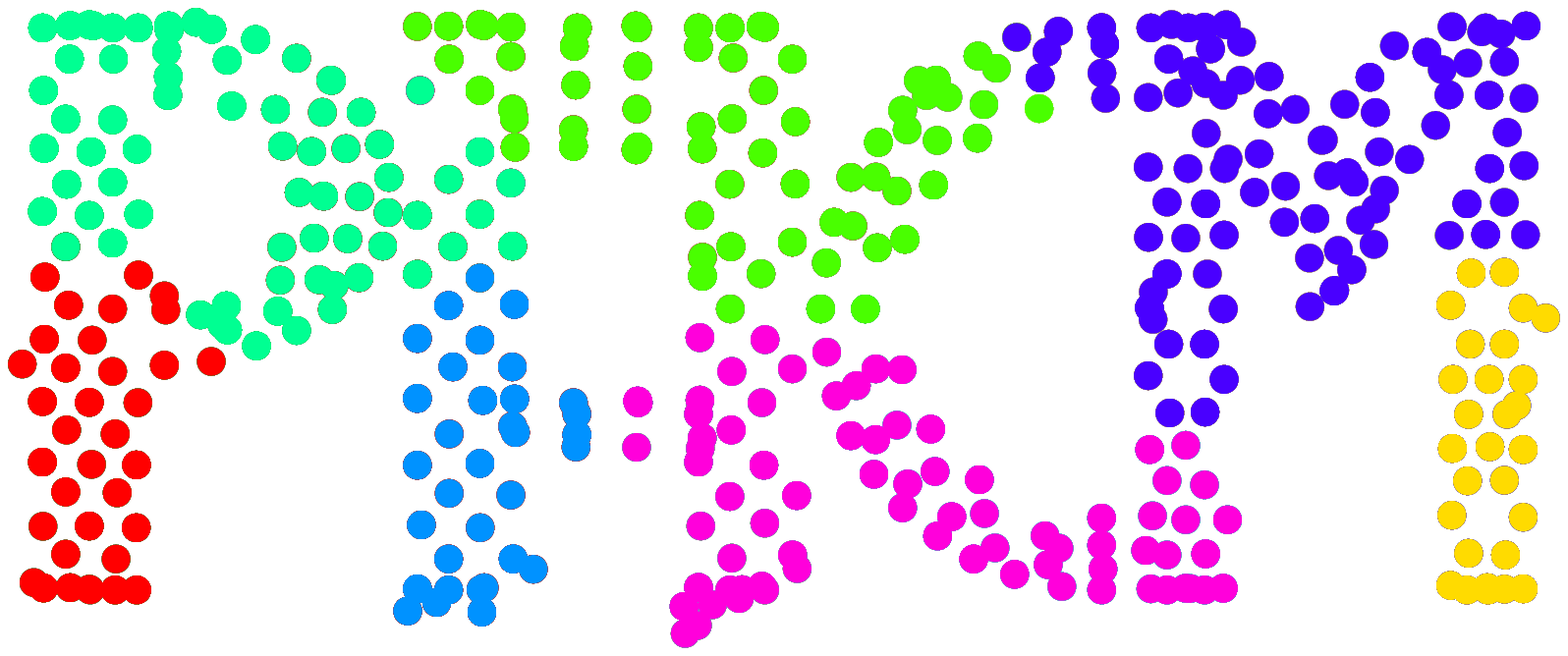}}
\subfigure[K=8]{\includegraphics[width=0.24\columnwidth]{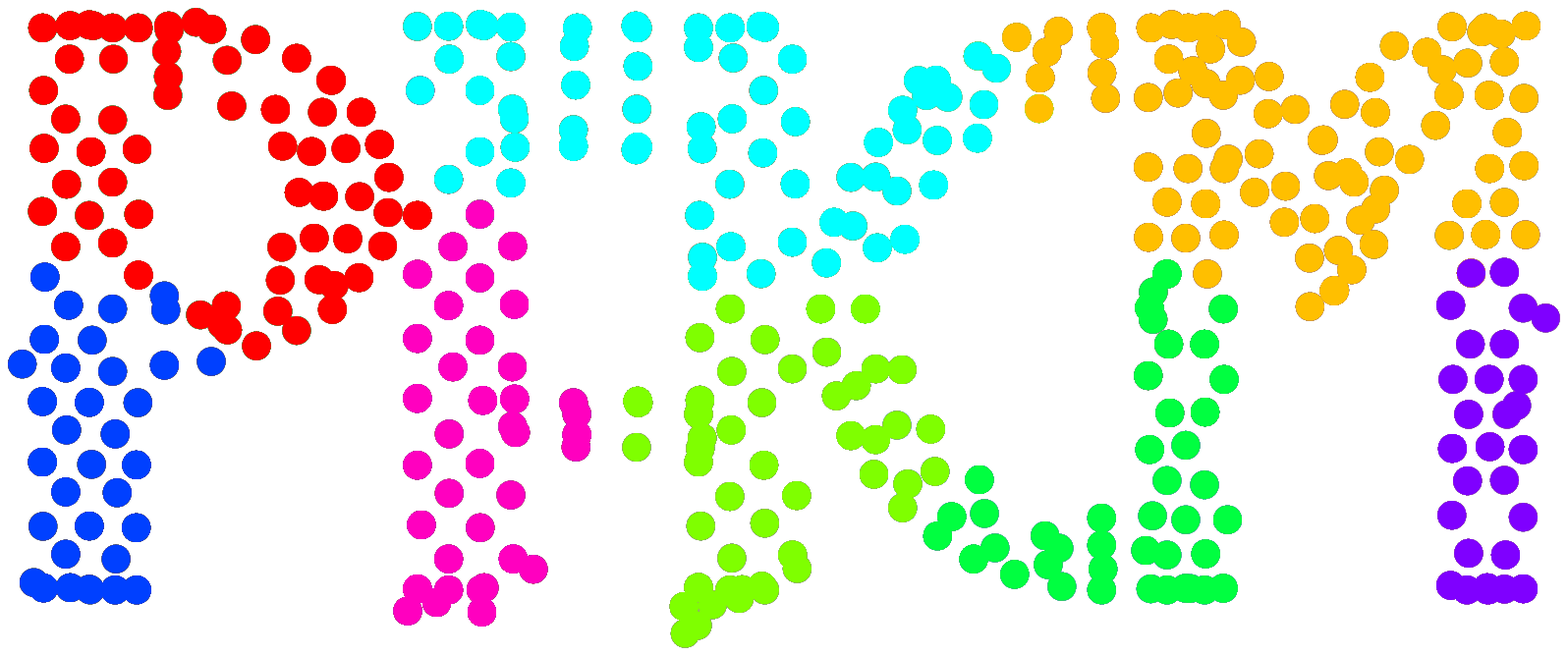}}
\subfigure[K=9]{\includegraphics[width=0.24\columnwidth]{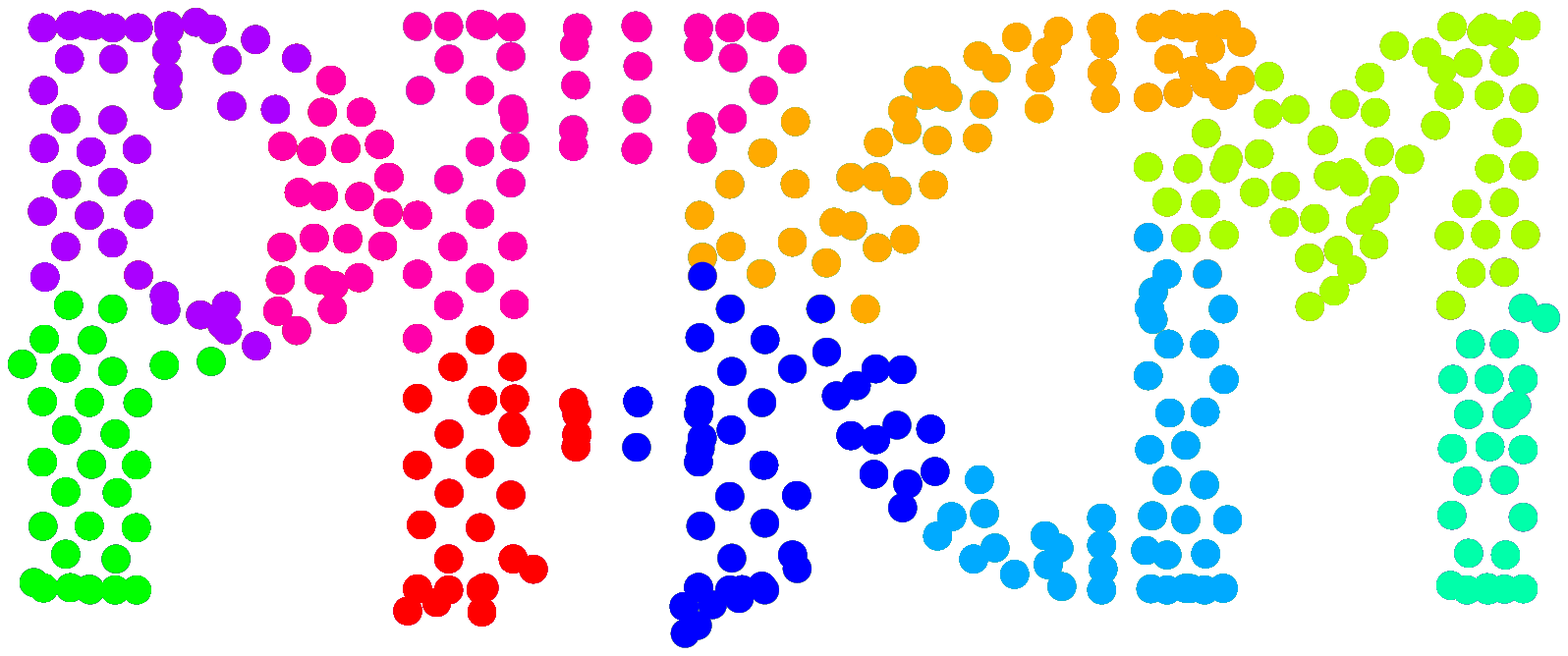}}
\end{center}
\caption{ Partitions of the PACM graph ($K$ is the number of partitions)}
\label{fig:spectral_part_PACM}
\end{figure}

\begin{figure}[h]
\begin{center}
\subfigure[K=2]{\includegraphics[width=0.24\columnwidth]{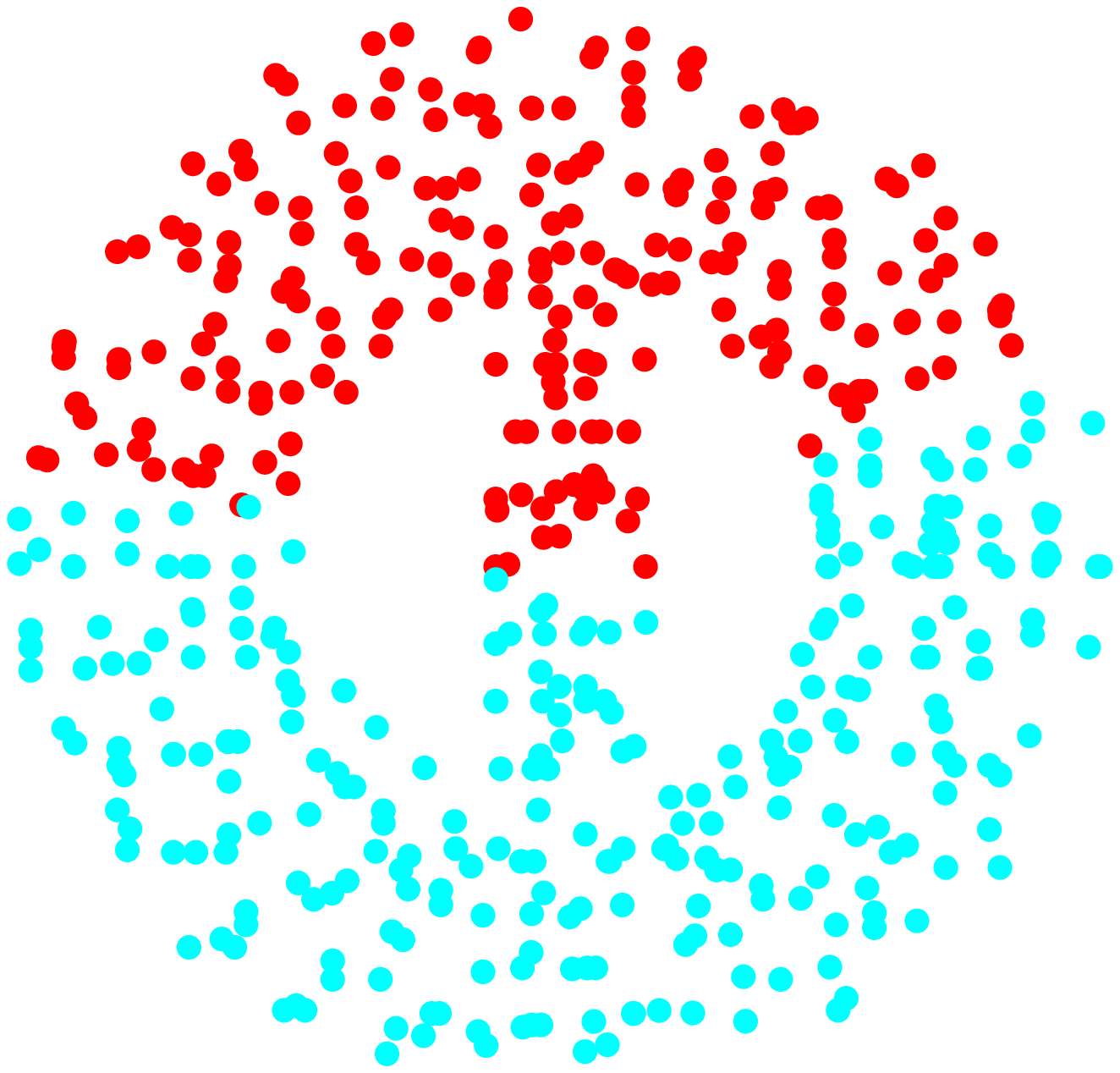}}
\subfigure[K=3]{\includegraphics[width=0.24\columnwidth]{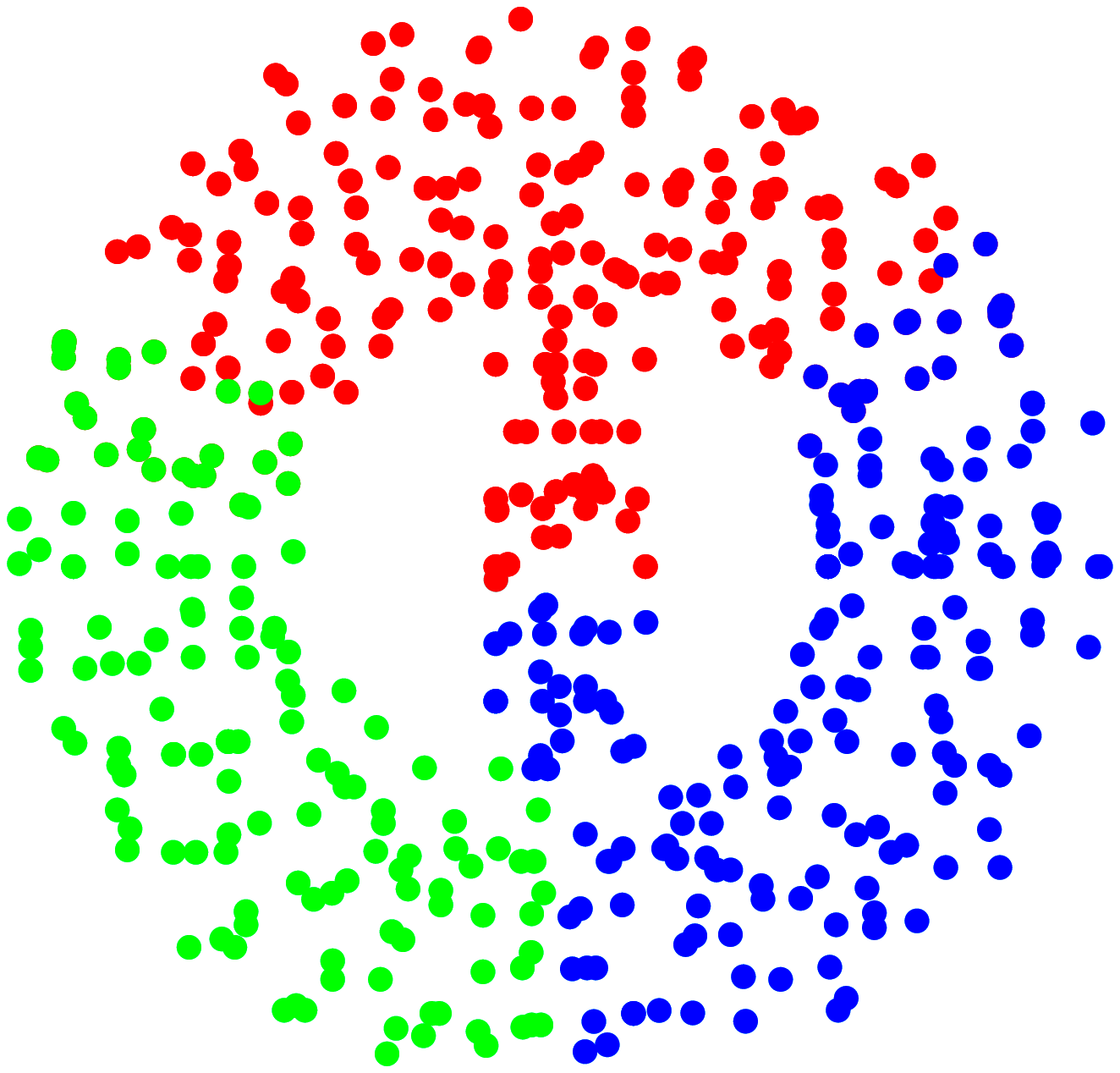}}
\subfigure[K=4]{\includegraphics[width=0.24\columnwidth]{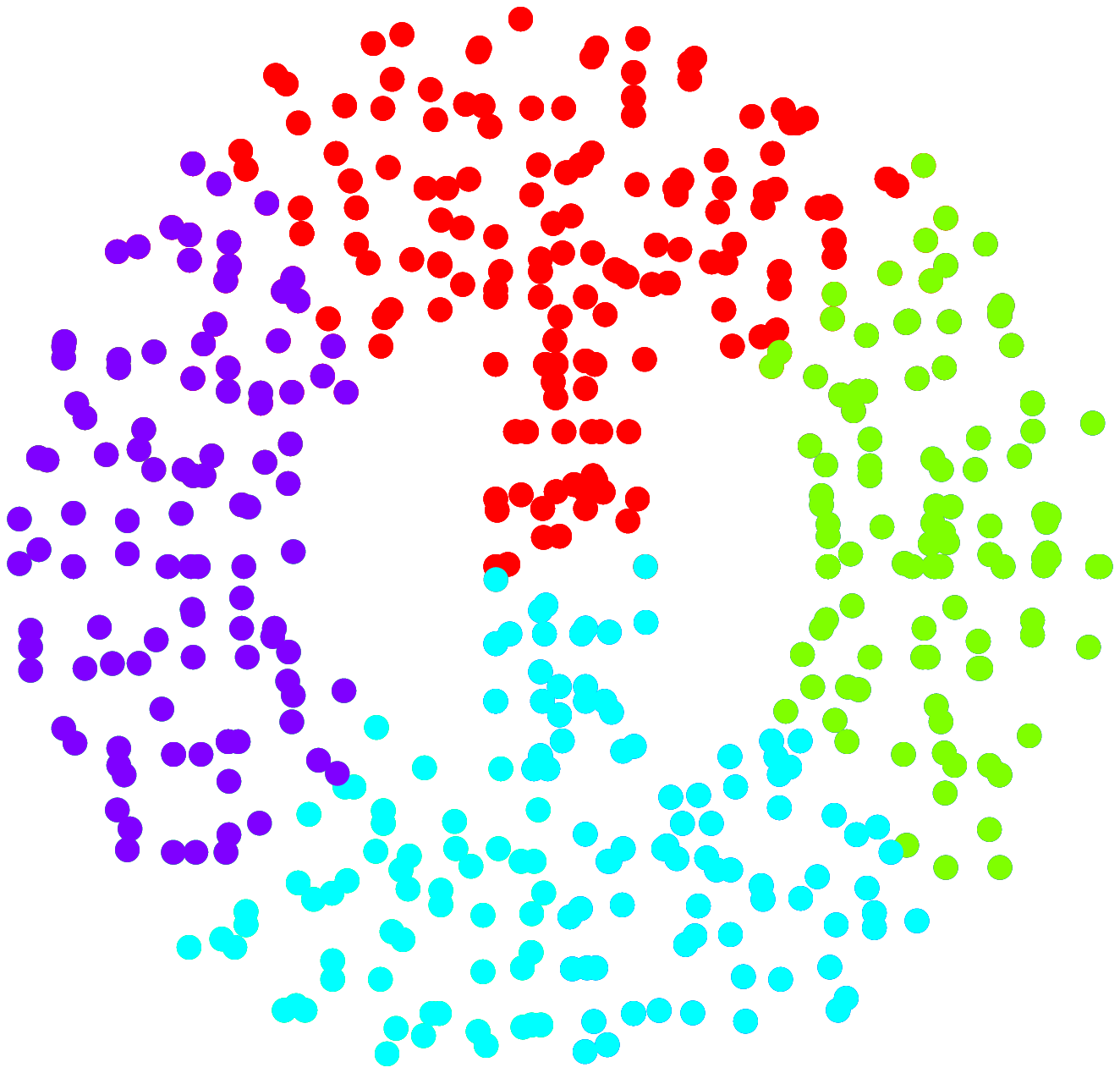}}
\subfigure[K=5]{\includegraphics[width=0.24\columnwidth]{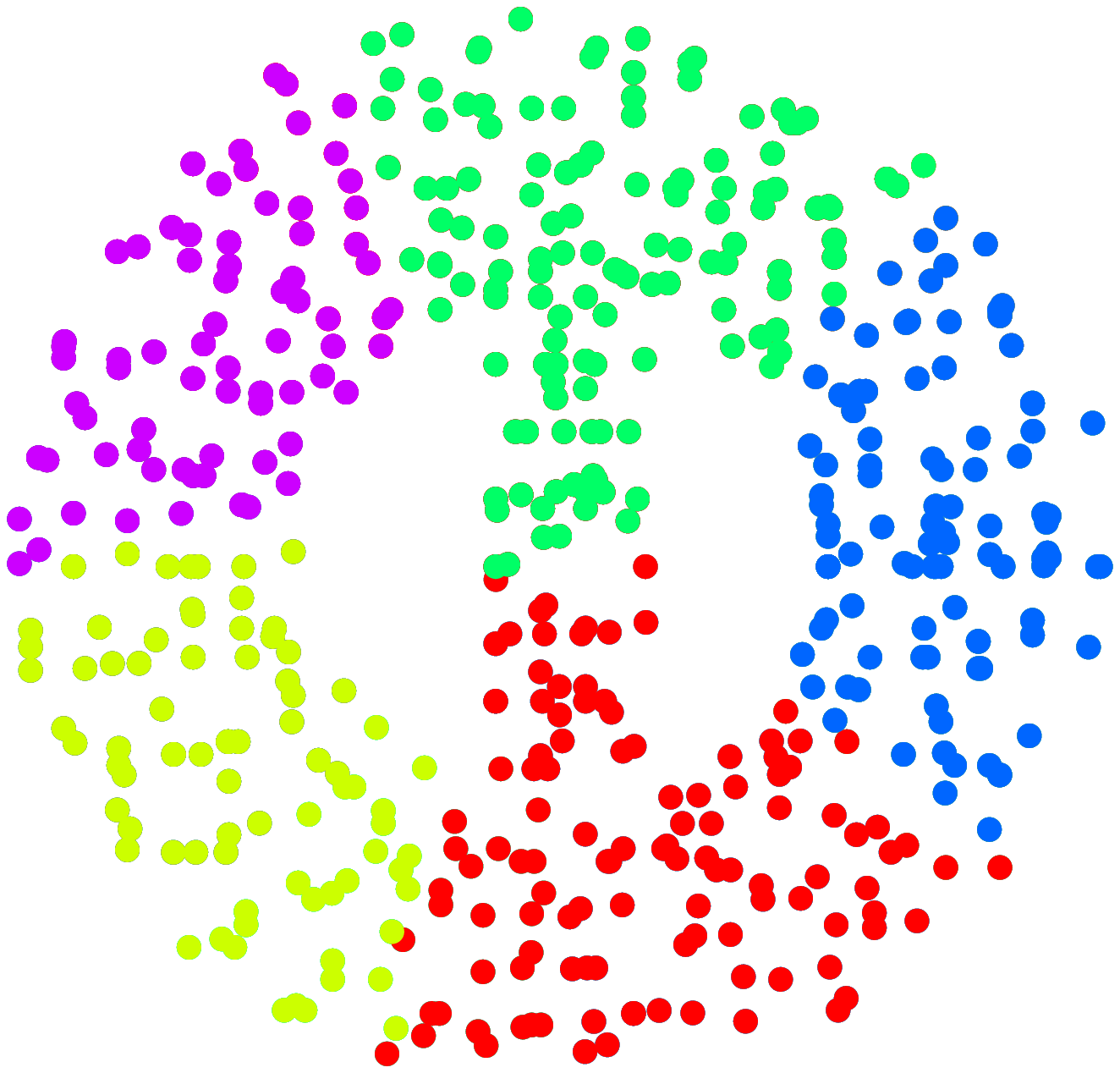}}
\subfigure[K=6]{\includegraphics[width=0.24\columnwidth]{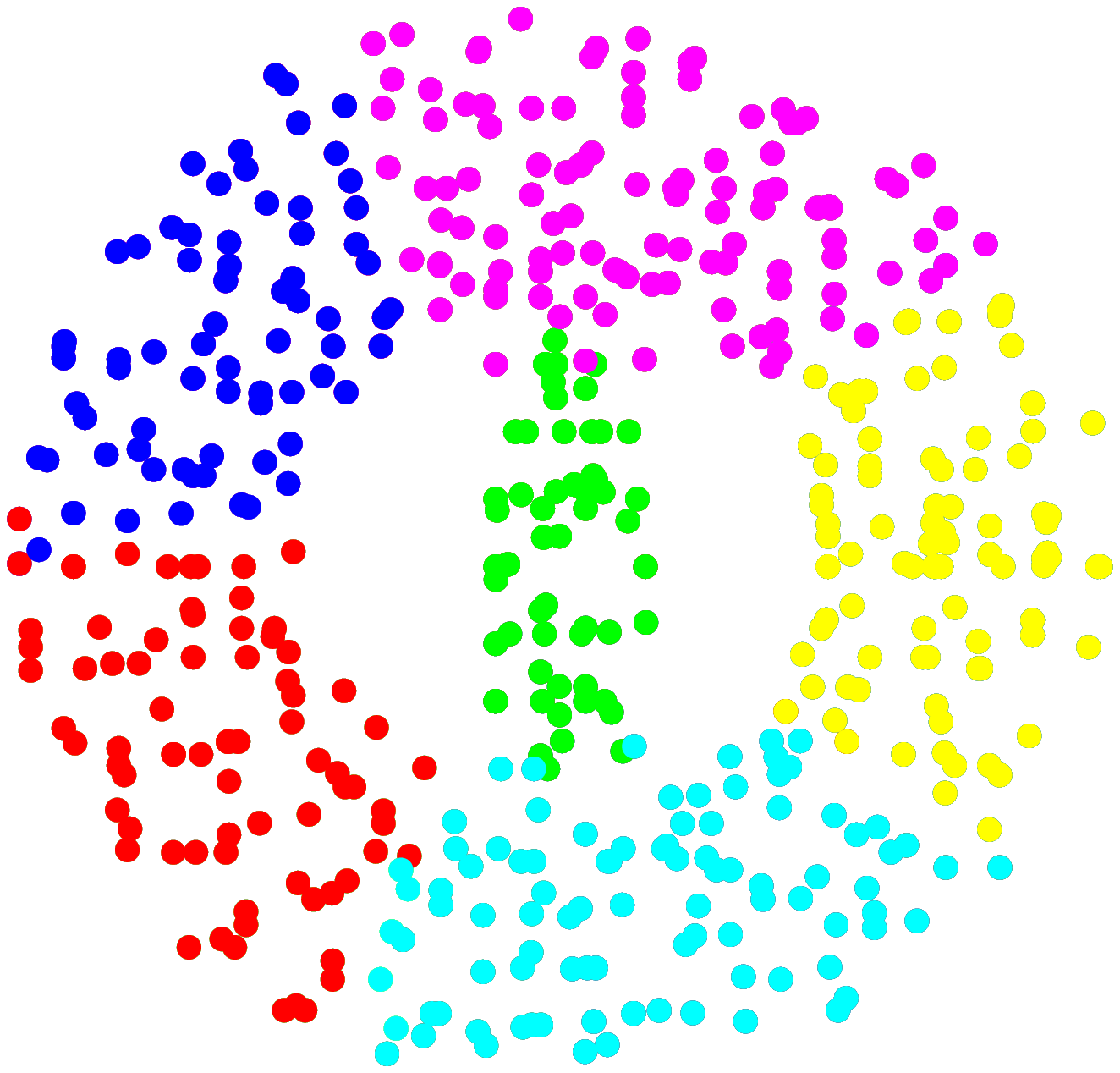}}
\subfigure[K=7]{\includegraphics[width=0.24\columnwidth]{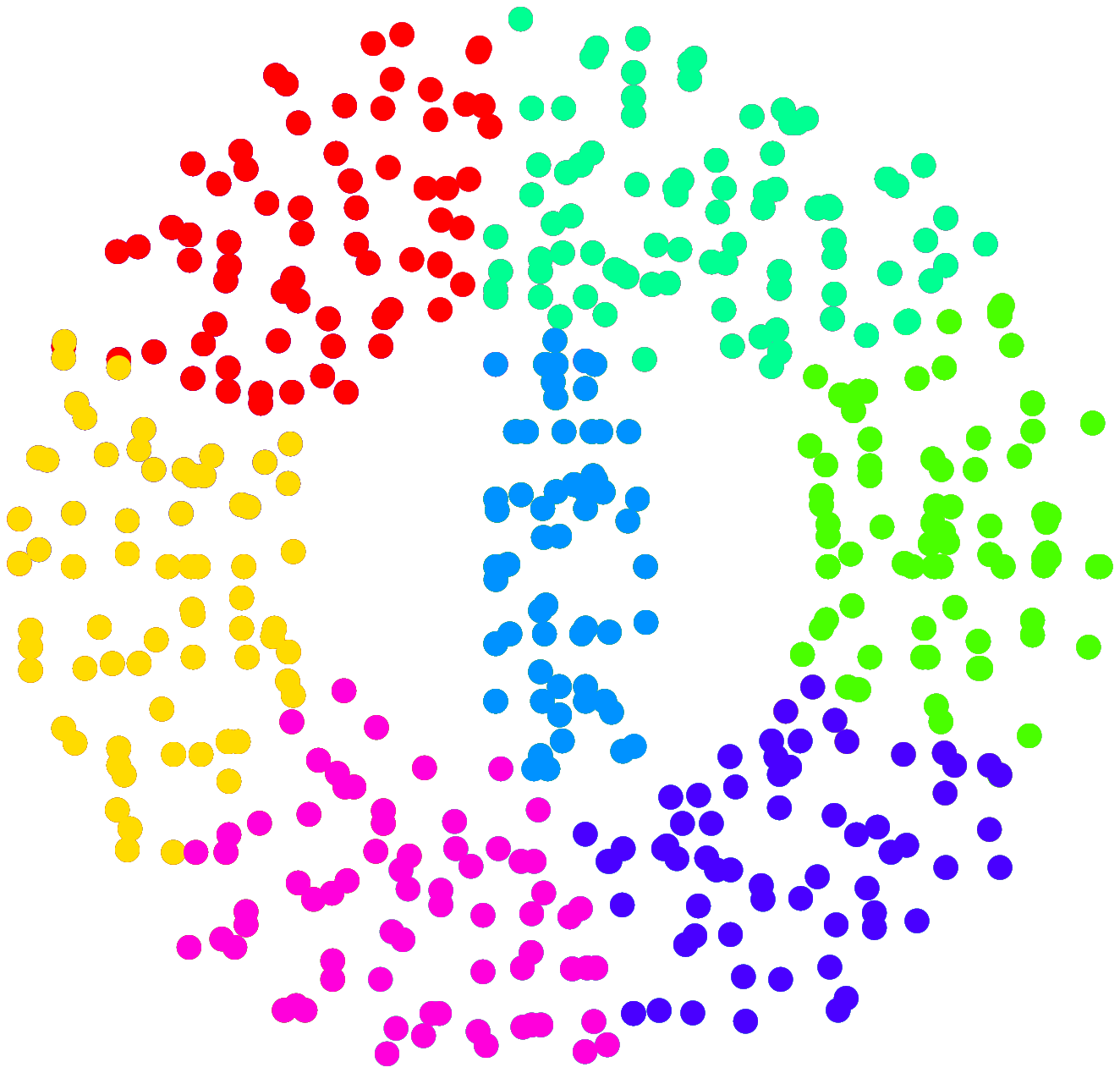}}
\subfigure[K=8]{\includegraphics[width=0.24\columnwidth]{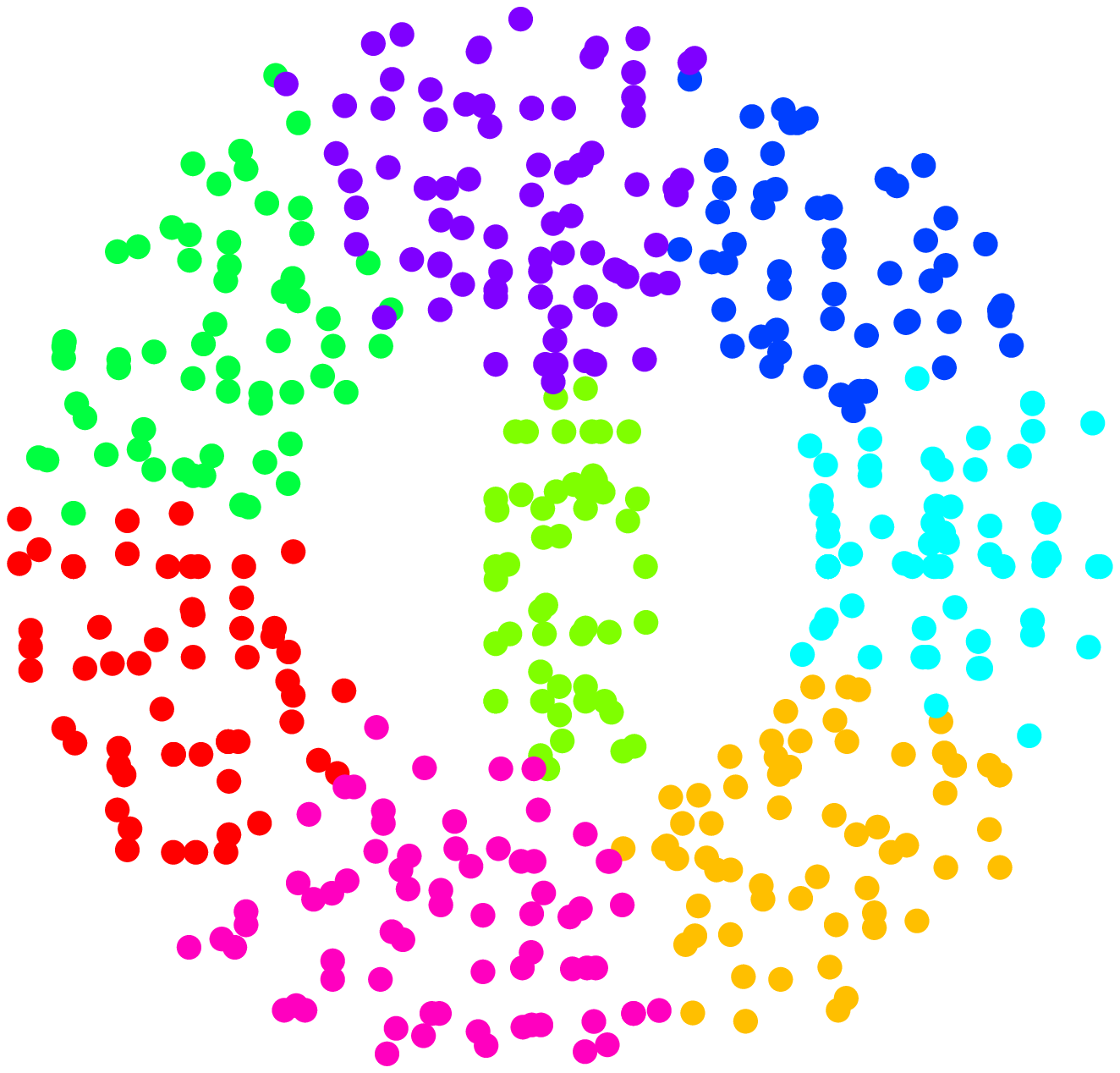}}
\subfigure[K=9]{\includegraphics[width=0.24\columnwidth]{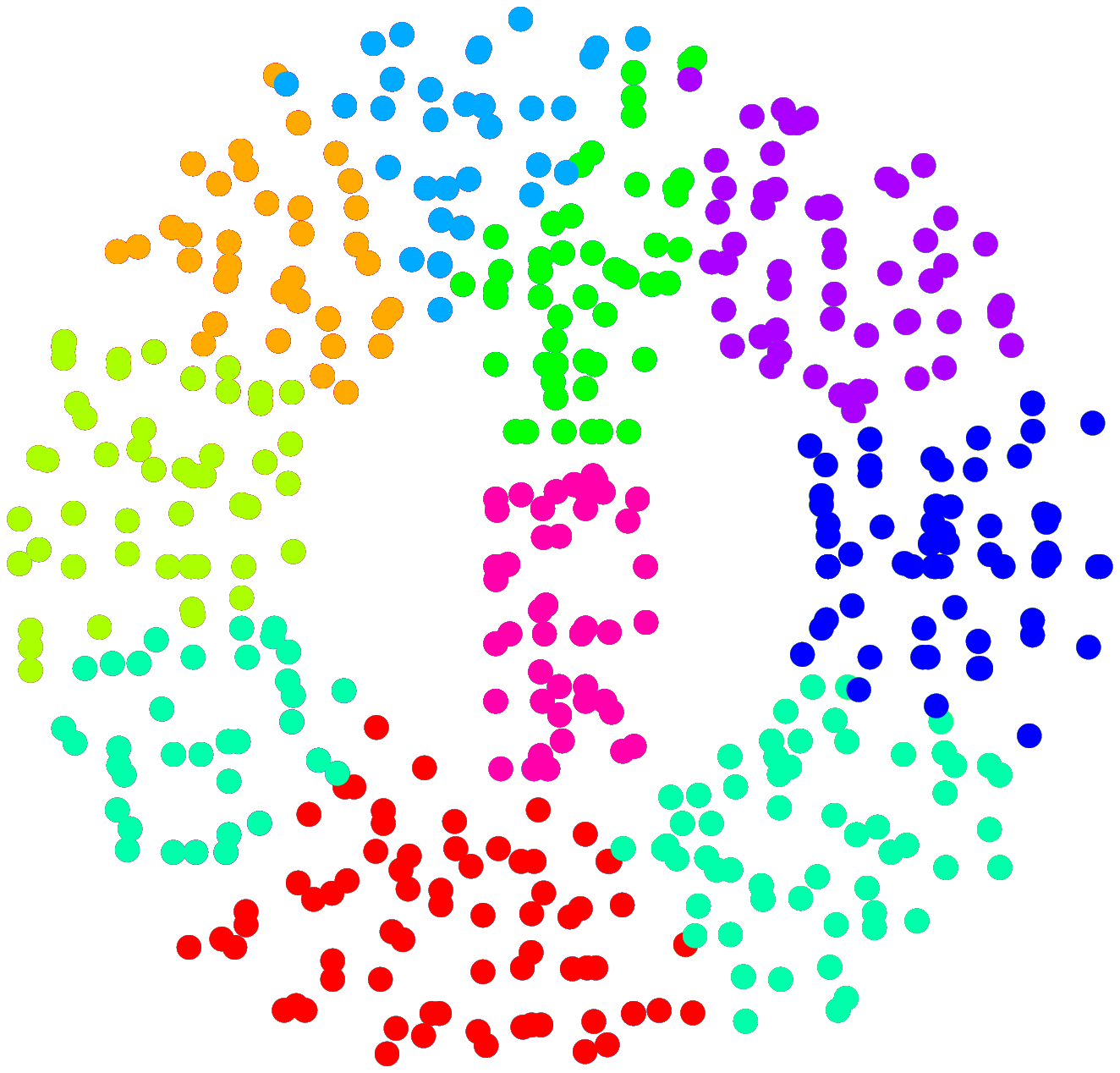}}
\end{center}
\caption{ Partitions of the BRIDGE-DONUT graph ($K$ is the number of partitions)}
\label{fig:spectral_part_DONUT}
\end{figure}

\section{The Median Denoising Algorithm (MDA) and a Rescaling Heuristic}
\label{MDA}

In this section we describe a method to improve the localization of the individual patches prior to the patch alignment and global synchronization steps. The main observation here is the following: suppose a pair of nodes appear in several different patches, then each patch provides a different estimation for their distance, and all these estimators can be averaged together to provide a perhaps more accurate estimator. The improved reestimated distances are then used to localize the individual patches again, this time using the cMDS algorithm since patches no longer contain any missing distances. The second part of this section introduces a simple rescaling procedure, that increases the accuracy of the final reconstruction.

Denote by $C_k$ ($k=1,\ldots,N$) the set of all pairwise edges within patch $P_k$ (so if $P_k$ has $t$ nodes then $C_k$ contains all  ${t\choose 2}$ edges). Let $\mathcal{C}$ denote the set of all edges that appear in at least one patch, i.e., $ \mathcal{C} = \cup_{k=1}^{N} C_k$. Note that $ \mathcal{C}$ contains edges $(i,j)$ that were initially missing from the measurement set (i.e. $(i,j) \notin E$) but for which we now have an estimate due to the initial localization of the patches (using, say, SDP). After the initial embedding, we have at least one estimated distance for each edge $(i,j) \in \mathcal{C}$, but because patches overlap, most of the edges appear in more than one patch. Denote by $d_{ij}^{1}, d_{ij}^{2}, \ldots $ the set of estimates for the (possibly missing) distance $(i,j)$, where $ d_{ij}^{k}$ is the
distance between nodes $i$ and $j$ in the SDP embedding of patch $P_k$. To obtain a more accurate estimate of $d_{ij}$ we take the median of all the above estimates, and denote the updated value by
\begin{equation}
\widetilde{d}_{ij} =  \operatorname{median} (d_{ij}^{1}, d_{ij}^{2}, \ldots)
\end{equation}

If the noise level in the originally measured distances is small, then we expect the estimates $d_{ij}^{1}, d_{ij}^{2}, \ldots $ to have small variance and take values close to the true distance. However, for higher levels of noise, many of the estimates can be very inaccurate (either highly overestimating or underestimating the true distances), and we choose the median value to approximate the true distance. Tables  \ref{tab:denoising_UNITCUBE} and \ref{tab:denoising_1d3z} show that MDA is indeed effective and reduces the noise level in the distances, usually by at least a few percentages. Note that the SDP improves by itself on the initial distance measurements (those that are available), and the MDA decreases their noise even further.

To take advantage of the more accurate updated distance values, we recompute the embeddings of all $N$ patches using the cMDS algorithm, since there exists an estimate for all pairwise distances within a patch. Whenever a hard constrained distance (a ``good" edge) is present, such as a bond length, we use its true distance in the cMDS embeddings, and ignore the noisy value returned by the SDP since hard constrained distances are not enforced in the SDP.

Tables  \ref{tab:denoising_UNITCUBE} and  \ref{tab:denoising_1d3z}
show the average ANE of the patches after the SDP embeddings ($\overline{ANE}_{old}$), and after the cMDS embeddings ran on the updated distances ($\overline{ANE}_{new}$).
Note that the new patch embeddings are significantly more accurate, in some cases the ANE decreasing by as much as $25\%$. We also experimented with an iterative version of this denoising algorithm, where at each round
we run the cMDS and recompute the median of the updated distances. However, subsequent iterations of this procedure did not improve the accuracy furthermore.

\begin{table}[h]
\begin{minipage}[b]{0.99\linewidth}
\centering
\begin{tabular}{|c||c|c||c|c|c||c|c|}
\hline
& & &  \multicolumn{3}{c||}{Existing distances} &  \multicolumn{2}{c|}{Missing distances}  \\
\cline{4-8}
$\eta$  & $\overline{ANE}_{old}$ & $\overline{ANE}_{new}$& $\bar{\eta}$  & $\bar{\eta}_{SDP}$ & $\bar{\eta}_{MDA}$ & $\bar{\eta}_{SDP}$ & $\bar{\eta}_{MDA}$ \\
\hline\hline
	   $10\%$   &   0.04 &    0.03 &        5.06   &                3.75  &      3.42  &       3.37  &       3.13 \\
           $20\%$   &   0.10  &   0.078  &      10.28    &               8.22   &       7.5  &       8.91  &       8.12 \\
           $30\%$   &   0.18  &     0.14  &      15.95    &              14.09   &     13.03  &      17.99  &      16.83 \\
           $40\%$   &   0.27   &   0.22   &     21.74     &             20.66    &    19.31   &     29.09   &     27.72  \\
           $50\%$   &   0.37  &    0.32  &      29.11    &              28.76   &      27.5  &      42.14  &      40.92 \\
\hline
\end{tabular}
\caption{ Denoising distances at various levels of noise, for the UNITCUBE data set. $\overline{ANE}_{old}$ and $\overline{ANE}_{new}$ denote the average ANE of all patches after the SDP embedding respectively after the cMDS embedding with the distances denoised by MDA. $\bar{\eta}$ is the empirical noise level (averaged over all noisy distances), $\bar{\eta}_{SDP}$ is the noise level after the SDP embedding, $\bar{\eta}_{MDA}$ is the noise level after running MDA. We show the noise levels individually for both the existing and missing distances. Note that $\bar{\eta} \neq \eta $ since the former denotes the average absolute value deviation from the true distance as a percentage of the true distance, while the latter is the parameter in the uniform distribution in (\ref{eta}) that is used to define the noise model (in fact, $\mathbb{E}[\bar{\eta}] = \frac{\eta}{2}$).}
\label{tab:denoising_UNITCUBE}
\end{minipage}
\hspace{1.5cm}
\end{table}

\begin{table}[h]
\begin{minipage}[b]{0.99\linewidth}
\centering
\begin{tabular}{|c||c|c||c|c|c|c||c|c|}
\hline
& & &  \multicolumn{4}{c||}{Existing distances} &  \multicolumn{2}{c|}{Missing distances}  \\
\cline{4-9}
$\eta$  & $\overline{ANE}_{old} $ & $\overline{ANE}_{new}$ & $\bar{\eta}$  & $\bar{\eta}_{SDP}$ & $\bar{\eta}_{SDP,GOOD}$ & $\bar{\eta}_{MDA}$ & $\bar{\eta}_{SDP}$ & $\bar{\eta}_{MDA}$ \\
\hline\hline
           $0\% $ &  0.00048  &  0.00039   &        0 &         0.02   &      0.14    &        0 &        0.38     &   0.32 \\
           $10\%$ &  0.00316  &  0.00266    &     4.91 &        3.42    &     1.79     &    3.13  &        3.2      &   2.99 \\
           $20\%$ &  0.00590  &  0.00506    &     9.81 &        7.01    &     3.68     &    6.44  &       6.18      &   5.81 \\
           $30\%$ &  0.00957  &  0.00803    &    14.74 &       10.75    &     5.59     &    9.88  &       9.51      &   8.95 \\
           $40\%$ &  0.01292  &  0.01120     &   19.65  &       14.5     &    7.62      &  13.42   &     13.55       &   12.8 \\
           $50\%$ &  0.01639  &  0.01462     &   24.57  &       18.1     &    9.82      &  16.93   &     18.24       &   17.35 \\
\hline
\end{tabular}
\caption{ Denoising distances at various levels of noise for the ubiqutin (PDG 1d3z, \cite{cornilescu}). $\bar{\eta}_{SDP,GOOD}$ denotes the average noise of the hard constraint (``good") distances after the SDP embeddings. Note that the current implementation of the SNL-SDP algorithm used for the patch embeddings does not allow for hard distance constraints.}
\label{tab:denoising_1d3z}
\end{minipage}
\hspace{1.5cm}
\end{table}

In the remaining part of this section, we discuss the issue of scaling of the distances after different steps of the algorithm, and propose a simple rescaling  heuristic that improves the overall reconstruction.

We denote the true distance measurements by $l_{ij} = \| p_i - p_j \|, (i,j) \in E$, and the empirical measurements by $d_{ij} = l_{ij} + \epsilon_{ij}, (i,j) \in E$, where $ \epsilon_{ij}$ is random independent noise, as introduced in (\ref{eta}). We generically refer to $\tilde{d}_{ij}$ as the distance between nodes $i$ and $j$, after different steps of the algorithm, as indicated by the columns of Tables \ref{tab:rescaling_UNITCUBE} and \ref{tab:rescaling_1d3z}. We denote by $\delta$ the average scaling of the distances with respect to their ground truth values, and similarly by $\kappa$ the average empirical noise of the distance measurements
$$ \delta = \operatorname{mean} \left(  \frac{l_{ij}}{\tilde{d}_{ij}}, (i,j) \in E \right),\;\;\;\;\;\;\;
    \kappa = \operatorname{mean} \left( \left| \frac{ \tilde{d}_{ij} - l_{ij}}{l_{ij}} \right|, (i,j) \in E \right)$$

The first column of the tables contains the scaling and noise values of the initial distance measurements, taken as input by ASAP. The fact that $\kappa$ is approximately half the value of the noise level $\eta$ stems from the fact that the measurements are uniformly distributed in $[ (1-\eta) l_{ij},(1+\eta) l_{ij}]$. Note that the initial scaling is quite significant, and this is a consequence of the same noise model and the geometric graph assumption. Distances that are scaled down by the noise will still be available to 3D-ASAP, while distances that are scaled up and become larger than the sensing radius will be ignored. Therefore, on average, the empirical distances are scaled down and $\delta$ is greater than 1.

The second column gives the scaling and noise values after the MDA algorithm followed by cMDS to recompute the patches. Note that after this denoising heuristic, while the scaling $\delta$ remains very similar to the original values, the noise $\kappa$ decreases considerably especially for the lower levels of noise.

The third column computes the scaling, noise and ANE values after the least squares step for estimating the translations, where by ANE we mean the average normalized error introduced in (\ref{AAvgNE}). Note that both the scaling and the noise levels significantly increase after integrating all patches in a globally consistent framework. To correct the scaling issue, we scale up all the available distances (thus taking into account only the edges of the initial measurement graph) by $\delta^{*}$, computed as
$$ \delta^{*} = \operatorname{mean} \left(  \frac{d_{ij}}{ \tilde{d}_{ij}^{LS} },  (i,j) \in E \right)$$
where $d_{ij}$ denotes the initial distances, and $\tilde{d}_{ij}^{LS}$ the distances after the least-square step. Note that, as a consequence, the ANE error of the reconstruction decreases significantly.

Finally, we refine the solution by running gradient descent, on the initial distances $d_{ij}$ scaled up by $\delta^{*}$, which further improves the scaling, noise and ANE values.
\begin{table}[h]
\begin{minipage}[b]{0.99\linewidth}
\centering
\begin{tabular}{|c|c|c||c|c||c|c|c||c|c||c|c|c|}
\hline
& \multicolumn{2}{c||}{Original} &  \multicolumn{2}{c||}{MDA-cMDS}  &  \multicolumn{3}{c||}{Least squares} & \multicolumn{2}{c||}{Rescaled}  & \multicolumn{3}{c|}{Gradient Descent}  \\
\cline{2-13}
$\eta  \%$ & $\delta$ & $\kappa \%$ & $\delta$& $ \kappa \%$ & $\delta$& $\kappa \%$ &   ANE & $\delta^{*}$&    ANE & $\delta$& $\kappa \%$ & ANE \\
\hline\hline
$10$ &      1.01 &   5.1 &   1.01 &  3.4 &   1.02 &   3.6 &   0.05 &   1.01 &   0.05 &   1.00 &   3.0 &   0.04 \\
$20$ &     1.05 &  10.3 &   1.04 &   7.5 &   1.08  &   8.8 &   0.12 &   1.04 &   0.10 &   1.02 &   6.1 &   0.07 \\
$30$ &     1.12 &  16.0 &   1.11 &  13.0 &   1.21  &  17.0 &   0.25 &   1.09 &   0.20 &   1.04 &   9.9 &   0.16 \\
$40$ &     1.22 &  21.8 &   1.21 &  19.3 &   1.43  &  27.8 &   0.36 &   1.21 &   0.26 &   1.04 &  13.9 &   0.19 \\
$45$ &      1.29 &  26.0 &   1.28 &  23.8 &   1.65 &  35.6 &   0.45 &   1.32 &   0.30 &   1.02 &  17.8 &   0.26 \\
$50$ &      1.38 &  29.2 &   1.36 &  27.5 &   1.92 &  43.8 &   0.53 &   1.47 &   0.35 &   0.99 &  22.4 &   0.32 \\
\hline
\end{tabular}
\caption{ Average scaling and noise of the existing edges (and where appropriate the ANE of the reconstruction) for the UNITCUBE graph, at various steps of the algorithm: after the original measurements,
after the SDP-MDA-cMDS denoising step, after the least squares in Step 3, after rescaling of the available distances by the empirical scaling factor $\delta^{*}$,
and after running gradient descent on the rescaled (originally available) distances.}
\label{tab:rescaling_UNITCUBE}
\end{minipage}
\hspace{1.5cm}
\end{table}

\begin{table}[h]
\begin{minipage}[b]{0.99\linewidth}
\centering
\begin{tabular}{|c|c|c||c|c||c|c|c||c|c||c|c|c|}
\hline
& \multicolumn{2}{c||}{Original} &  \multicolumn{2}{c||}{MDA-cMDS}  &  \multicolumn{3}{c||}{Least squares} & \multicolumn{2}{c||}{Rescaled}  & \multicolumn{3}{c|}{Gradient Descent}  \\
\cline{2-13}
$\eta \%$ & $\delta$ & $\kappa \%$ & $\delta$& $\kappa \%$ & $\delta$& $\kappa \%$ &   ANE & $\delta^{*}$&    ANE & $\delta$& $\kappa \%$ & ANE \\
\hline\hline
$10$ &     1.00   &   0.03  &    1.00   &   0.02   &   1.02    &  0.04   &   0.01    &  1.02   &        4e-3 &   0.99  &    0.01   &       3e-3 \\
$20$ &    1.01    &  0.06   &   1.00    &  0.04    &  1.04     & 0.07    &  0.01     & 1.03    &  0.01    &  0.99   &   0.02    &  0.01   \\
$30$ &     1.02   &   0.08  &    1.01   &   0.06   &   1.06    &  0.10 &   0.02    &  1.05   &   0.01   &   0.99  &    0.03   &   0.01   \\
$40$ &    1.03    &  0.11   &   1.01    &  0.08    &  1.10   & 0.14    &  0.02     & 1.09    &  0.02    &  0.99   &   0.05    &  0.01   \\
$50$ &    1.06    &  0.14   &   1.02    &  0.10  &  1.18     & 0.19    &  0.04     & 1.16    &  0.03    &  0.98   &   0.08    &  0.01   \\
\hline
\end{tabular}
\caption{ Average scaling and noise of all edges (good+NOE) of the ubiqutin (PDG 1d3z) (and where appropriate the ANE of the reconstruction) at various steps of the algorithm: after the original measurements,
after the SDP-MDA-cMDS denoising step, after the least squares in Step 3, after rescaling of the available distances by the empirical scaling factor $\delta^{*}$,
and after running gradient descent on the rescaled NOE distances.}
\label{tab:rescaling_1d3z}
\end{minipage}
\hspace{1.5cm}
\end{table}

\section{Synchronization with molecular fragments}
\label{sync_anchors}

In the molecule problem, as mentioned earlier, there are molecular fragments whose local configuration is known in advance up to an element of the special Euclidean group SE(3) rather than the Euclidean group E(3). In other words, these are small structures that need to be translated and rotated with respect to the global coordinate system, but no reflection is required. As a result, for their corresponding patches we know the corresponding group element in $\mathbb{Z}_2$. This motivates us to consider the problem of synchronization over $\mathbb{Z}_2$ when ``molecular fragment" information is available, and refer to it
from now on as SYNC($\mathbb{Z}_2$). We propose and compare four methods for solving  SYNC($\mathbb{Z}_2$): two relaxations to a quadratically constrained quadratic program (QCQP), and two semidefinite programming (SDP) formulations.

Mathematically, the synchronization problem over the group $\mathbb{Z}_2 = \{\pm 1\}$ can be stated as follows: given $k$ group elements (anchors) $\mathcal{A} = \{ a_1,\ldots,a_k\}$ and a set of (possibly) incomplete (noisy) pairwise group measurements $z_{ij} = a_i x_j^{-1}$ and $z_{ij} = x_i x_j^{-1}$, $(i,j) \in E(G)$, find the unknown group elements (sensors) $\mathcal{S} = \{ x_{1}, \ldots ,x_{l} \} $. We may sometimes abuse notation and denote all nodes by  $ x = \{ x_{1}, \ldots , x_{l}, x_{l+1}, \ldots, x_{N}\}$, where $N=l+k$, with the understanding that the last $k$ elements denote the anchors. Whenever we say that indices $i \in \mathcal{S}$ and $j \in \mathcal{A} $, it should be understood that we refer to sensor $x_i \in \mathcal{S}$, respectively anchor $a_j \in \mathcal{A}$. In the molecule problem, $k$ denotes the number of patches whose reflection is known a priori , and the goal is to estimate the reflection of the remaining $l=N-k$ patches.

In the absence of anchors (when $k=0$ and $N=l$), the synchronization problem over $\mathbb{Z}_2$ was considered in \cite{ASAP}, following the approach for angular synchronization introduced in \cite{sync}. The goal is to estimate $N$ unknown group elements $z_1,\ldots,z_N  \in \mathbb{Z}_2$, whose pairwise group relations are captured in a sparse $N \times N$ matrix $Z=(z_{ij})$ where
\begin{equation}
z_{ij} = \left\{
     \begin{array}{rl}
   z_i  z_j^{-1} & \;\; \text{ if the measurement is correct}    \\
 - z_i  z_j^{-1} & \;\; \text{ if the measurement is incorrect}  \\
     \end{array}
   \right.
\label{noiseModelZ22}
\end{equation}
The set $E^P$ of pairs $(i,j)$ for which a pairwise measurement exists (correct or incorrect) can be realized as the edge set of the patch graph $G^P=(V^P,E^P)$, where vertices correspond to patches, and edges to the measurements. We denote the original solution by $z_1, \ldots, z_N$ and our approximated solution by $ x_1, \ldots, x_N$. Our task is to estimate $ x_1, \ldots, x_N$ such that we satisfy as many pairwise group relations as possible. To that end, we start by considering the problem of maximizing the following quadratic form
\begin{equation} \label{maxZ2}
 \max_{ x_1,\ldots,x_N \in \mathbb{Z}_2^N} \sum_{i,j=1}^{N} x_i {Z}_{ij} x_j = \max_{ x_1,\ldots,x_N \in \mathbb{Z}_2^N} x^T Z x,
\end{equation}
whose maximum,  in the noise-free case, is attained when $x=z$.
Note that for the correct set of reflections $z_1, \ldots, z_N$, each correct group measurement $z_{ij} = z_i z_j^{-1}$ contributes $ z_i {Z}_{ij} z_j = + 1 $ to the sum in (\ref{maxZ2}), while each incorrect measurement will add a $-1$ to the summation.
Our task now is to solve the quadratic integer optimization problem in (\ref{maxZ2}), which is unfortunately a non-convex optimization problem. Since NP-hard problems, such as the maximal clique problem, can be formulated as an integer quadratic problem, the maximization in (\ref{maxZ2}) is itself NP-hard. We therefore make use of the following relaxation
\begin{equation} \label{relmaxZ2_2}
\max_{ {\sum_{i=1}^{N} |x_i|^2=N}} \; \sum_{i,j=1}^{N} x_i {Z}_{ij} x_j = \max_{ \parallel x \parallel^2 = N} x^T Z x
\end{equation}
whose maximum is achieved when $x=v_1$, where $v_1$ is the normalized top eigenvector of $Z$, satisfying $Z v_1 = \lambda_1 v_1$ and $ \| v_1 \|^2=N$, with $\lambda_1$ being the largest eigenvalue. Thus an approximate solution to the maximization problem in (\ref{maxZ2}) is given by the top eigenvector of the symmetric matrix $Z$. In practice, for increased robustness to noise, we use the following alternative relaxation
\begin{equation} \label{relmaxZ2_D}
\max_{  x^T D x = 2m } x^T Z x,
\end{equation}
where $D_{ii} = \sum_{j=1}^{N} |Z_{ij}|$, $m$ denotes the number of edges, and hence $2m$ is the sum of the node degrees. The solution to the maximization problem in (\ref{relmaxZ2_D}) is given by the top eigenvector of the normalized matrix $ D^{-1}Z$. We refer the reader to \cite{sync,ASAP} for an analysis of the eigenvector method and its  connection with the normalized discrete graph Laplacian.



\subsection{Synchronization by relaxing a QCQP}
\label{sub:qcqp}

When anchor information is available, meaning that we know a priori  some of the group elements which we refer to as anchors, we follow a similar approach to equations (\ref{maxZ2}, \ref{relmaxZ2_2})  that motivated the eigenvector synchronization method. Similarly, we are interested in finding the set of unknown elements that maximize the number of satisfied pairwise measurements, but this time, under the additional  constraints imposed by the anchors, i.e. $x_i = a_i, i \in \mathcal{A}$. Unfortunately, maximizing the quadratic form $x^{T}Zx$ under the anchor constraints, is no longer an eigenvector problem. Our approach is to combine under the same objective function both the contribution of the sensor-sensor pairwise measurements (as a quadratic term) and the contribution of the anchors-sensor pairwise measurements (as a linear term). To that end, we start by formulating the synchronization problem as a least squares problem, by minimizing the following quadratic form
\begin{eqnarray}
\min_{ x }  \sum_{(i,j) \in E} (x_i - Z_{ij}x_j)^2  \nonumber
& = &  \min_{ x }  \sum_{(i,j) \in E} z_i^2 + Z_{ij}^2 x_j^2  - 2 Z_{ij} x_i x_j \nonumber \\
& = &  \min_{ x }  \sum_{(i,j) \in E}  x_i^2 + x_j^2  - 2 Z_{ij} x_i x_j \nonumber \\
& = &  \min_{ x }  \sum_{i=1}^n d_i x_i^2  -  \sum_{(i,j) \in E} 2 Z_{ij} x_i x_j \nonumber \\
& = &  \min_{ x }  x^T D x - x^T Z x  \nonumber \\
& = &  \min_{ x } x^T( D - Z) x
\label{LS_DZ}
\end{eqnarray}
To account for the existence of anchors, we first write the matrices $Z$ and $D$ in the following block format \\
$$ Z = \left[ \begin{array}{cc}
S & U  \\
U^{T} & V \\
\end{array} \right],\;\;\;\;\;\;\;\;\;\;\;\;\; D =  \left[ \begin{array}{cc}
D_S & 0  \\
0 & D_V \\
\end{array} \right]  $$
where $S_{l \times l}$, $U_{l \times k}$ and $V_{ k \times k}$ denote the sensor-sensor, sensor-anchor respectively anchor-anchor measurements, and $D$ is a diagonal matrix with $D_{ii} = \sum_{j=1}^{N} |Z_{ij}|$. Note that $V$ is a matrix with all nonzero entries, since the measurement between any two anchors is readily available. Similarly, we write (with a slight abuse of notation) the solution vector in the form $x = [s \; a]^T$, where $s$ denotes the signs of the sensor nodes, and $a$ the signs of the anchor nodes. The quadratic function minimized in (\ref{LS_DZ}) can now be written in the following form
\begin{equation}
 \left
[ \begin{array}{cc} s^T  &  a^T  \\  \end{array} \right]
\left[ \begin{array}{cc}
D_S - S & - U  \\
- U^{T} & D_V - V \\
\end{array} \right]
\left[ \begin{array}{c}
				s  \\
				a \\
				\end{array} \right]
 = s^T (D_S - S)  s - 2 s^T U a + a^T (D_V - V) a
\label{blockform}
\end{equation}
The vector $(Ua)_{l \times 1}$  can be interpreted as the anchor contribution in the estimation of the  sensors. Note that in the case when the sensor-anchors measurements should be trusted more than the sensor-sensor measurements, the anchor contribution can be weighted accordingly, and equation (\ref{blockform}) becomes $z^T (D_S-S) z - 2 \gamma  z^T U a + a^T (D_V - V) a $, for a given weight $\gamma$. Since $ a^T (D_V - V) a $ is a (nonnegative) constant,
we are interested in minimizing the integer quadratic form $$ \underset{z \in \mathbb{Z}_2^l}{\text{ minimize }} z^T (D_S - S)  z - 2 z^T U a.$$ Unfortunately, solving such a problem is NP-hard, and thus we introduce the following relaxation to a quadratically constrained quadratic program (QCQP)
\begin{equation}
	\begin{aligned}
	& \underset{z = (z_1, \ldots, z_l)}{\text{minimize}}
	& & z^T (D_S - S)  z - 2 z^T U a  &  \\
	& \text{subject to}
	& &  z^T  z = l
	\end{aligned}
\label{maximization_zTz}
\end{equation}
We proceed by considering the Lagrangian function
\begin{equation}
	\phi(z,\lambda) =  z^T (D_S - S)  z -  2 z^T U a + \lambda (z^T z - l)
\label{lagrange}
\end{equation}
where $\lambda$ is the Lagrange multiplier. Differentiating (\ref{lagrange}) with respect to $z$, we are led to the following equation
\begin{equation}
	 2 (D_S - S) z -  2  U a + 2 \lambda z = 0  \nonumber
\label{diff_lagrang1}
\end{equation}
which can be written as
\begin{equation}
	 (D_S - S + \lambda I) z = U a
\label{diff_lagrang}
\end{equation}
Using the fact that $z^T z = l$, (\ref{diff_lagrang}) becomes
\begin{equation}
      (U a)^T  (D_S - S + \lambda I)^{-2} (U a) = l
\label{lambda_eq}
\end{equation}
which we solve for $\lambda$. We refer the reader to \cite{golub} for a more detailed analysis of eigenvalue solutions to such quadratically constrained optimization problems.
Finally, the solution to the minimization problem formulated in (\ref{maximization_zTz}) is given by
\begin{equation}
  z^* = ( D_S - S + \lambda I )^{-1} (U a)
\label{solquadlin}
\end{equation}
For simplicity, we choose to solve the nonlinear matrix equality (\ref{lambda_eq}) in MATLAB using the $lsqnonlin$ command, whose success in finding $\lambda$ is contingent upon a good initialization. To that end, we consider the eigendecomposition of the symmetric positive definite matrix
    $$D_S-S = \sum_{i=1}^{l} \lambda_i \phi_i \phi_i^T,$$
where $(D-S) \phi_i = \lambda_i \phi_i$, for $i =1,\ldots,l$, and expand the vectors  $Ua$ and $z$ in the form  $Ua = \sum_{i=1}^l \beta_i \phi_i $ and $z =  \sum_{i=1}^l \alpha_i \phi_i$.
Rewriting equation (\ref{diff_lagrang}),  $(D_S - S) z + \lambda z = U a$,  in terms of the above expansion yields
\begin{equation}
\sum_{i=1}^l \lambda_i  \alpha_i \phi_i + \lambda \sum_{i=1}^l \alpha_i \phi_i = \sum_{i=1}^l \beta_i \phi_i.
\end{equation}
Thus, for every $i=1,\dots,l$ it must hold true that
\begin{equation}
\lambda_i  \alpha_i  + \lambda  \alpha_i = \beta_i,
\label{constr}
\end{equation}
i.e. $\alpha_i = \frac{\beta_i}{\lambda + \lambda_i}$. Since $z^Tz=l$ and $\sum_{i=1}^l \alpha_i^2 = l$, the Lagrange multiplier $\lambda$ must satisfy
\begin{equation}
 \sum_{i=1}^l  \frac{\beta_i^2}{ (\lambda + \lambda_i)^2} = l
\label{funclam}
\end{equation}
The condition that $\lambda + \lambda_i > 0$ ensures that the solution search for $\lambda$ lies outside the set of singularities of (\ref{funclam}). Thus $\lambda >  - \lambda_0$, where $\lambda_0$ is the smallest eigenvalue of the matrix $D_S-S$. Note that the row sums of $D_S-S$ are non-negative since the diagonal entries in $D_S$ are the degree of the sensor nodes in the entire graph, taking into account the sensor-anchor edges as well, not just the sensor-sensor edges. Furthermore, $D_S-S$ is a positive semidefinite matrix (thus $\lambda_0 \geq 0$), as it can be seen from the following identity
\begin{eqnarray}  \label{LS_zZ}
x^T(D_S-S)x
& = &  \sum_{i=1}^l x_i^2 d_i - \sum_{(i,j) \in E_{SS}} 2 x_i x_j S_{ij} \geq \sum_{(i,j) \in E_{SS} , S_{ij}= \pm 1} (x_i-S_{ij} x_j)^2  \geq 0  \nonumber
\end{eqnarray}
where $E_{SS}$ denotes the set of sensor-sensor edges.

We also consider the following formulation similar to (\ref{maximization_zTz}), but we replace the constraint
 $z^T z = l$ with $z^T D_S z = \Delta$,
where $\Delta = \sum_{i=1}^l d_i $ is the sum of the degrees of all sensor nodes. Note that the following change of variable $\bar{z}=D_S^{1/2}z$ brings the new optimization problem to a form similar to (\ref{maximization_zTz})
\begin{equation}
	\begin{aligned}
	& \underset{\bar{z} }{\text{minimize}}
	& & \bar{z}^T D_S^{-1/2} (D_S - S) D_S^{-1/2} \bar{z} - 2 \bar{z}^T D_S^{-1/2} U a  &  \\
	& \text{subject to}
	& &  \bar{z}^T \bar{z} = \Delta
	\end{aligned}
\label{maximization_zTDz}
\end{equation}
with the corresponding Lagrangian $\bar{\lambda}$ satisfying $\bar{\lambda} > - \bar{\lambda}_0$, where $\bar{\lambda}_0$ is the smallest eigenvalue of the matrix $D_S^{-1/2} (D_S - S) D_S^{-1/2}$.

\subsection{Synchronization by SDP}
An alternative approach to solving SYNC($\mathbb{Z}_2$) is by using semidefinite programming. The objective function in (\ref{maxZ2}) can be written as
\begin{equation}
  \sum_{i,j=1}^{N} x_i {Z}_{ij} x_j = Trace(Z \Upsilon)
\end{equation}
where $\Upsilon$ is the $N \times N$ symmetric rank-one matrix with $\pm 1$ entries
\begin{equation}
\Upsilon_{ij} = \left\{
     \begin{array}{rl}
x_i  x_j^{-1} & \;\; \text{ if } i,j \in \mathcal{S} \\
x_i  a_j^{-1} & \;\; \text{ if } i \in \mathcal{S}, j \in \mathcal{A} \\
a_i  a_j^{-1} & \;\; \text{ if } i,j \in \mathcal{A} \\
     \end{array}
   \right.
\label{upsilondef}
\end{equation}
Note that $\Upsilon$ has ones on its diagonal $ \Upsilon_{ii} =1, \forall i=1,\ldots,N$, and the anchor information gives another layer of hard constraints. The SDP relaxation of (\ref{maxZ2}) in the presence of anchors becomes
\begin{equation}
	\begin{aligned}
	& \underset{\Upsilon \in \mathbb{R}^{N \times N}}{\text{maximize}}
	& & Trace(Z \Upsilon) \\
	& \text{subject to}
	& & \Upsilon_{ii} =1, i=1,\ldots,N \\
	& & &  \Upsilon_{ij} = a_i  a_j^{-1}, \;\; \text{ if } i,j \in \mathcal{A} \\
		& & &   \Upsilon \succeq 0
	\end{aligned}
\label{SDP_max}
\end{equation}
where the maximization is taken over all semidefinite positive real-valued matrices  $\Upsilon \succeq 0$.  While $\Upsilon$ as defined in (\ref{SDP_max}) has rank one, the solution of the SDP is not necessarily of rank one. We therefore compute the top eigenvector of that matrix, and estimate $x_1,\ldots,x_s$ based on the sign of its first $s$ entries.

Alternatively, to reduce the number of unknowns in (\ref{SDP_max}) from $N=l+k$ to $l$, one may consider the following relaxation
 \begin{equation}
		\begin{aligned}
		& \underset{ \Upsilon \in \mathbb{R}^{l \times l}; x \in \mathbb{R}^{l} }{\text{maximize}}
		& & Trace(S \Upsilon)  + 2 x^T U a \\
		& \text{subject to}
		& & \Upsilon_{ii} =1, \forall i=1,\ldots,l  \\
		& & &  \left[ \begin{array}{cc}
				\Upsilon  & x  \\
				 x^T & 1 \\
				\end{array} \right]  \succeq 0
		\end{aligned}
	\label{SDP_max_XY}
\end{equation}
Ideally, we would like to enforce the constraint $\Upsilon = x x^T$, which guarantees that  $\Upsilon$ is indeed a rank-one solution. However, since rank constrains do not lead to convex optimization problems, we relax this constraint via Schur's lemma to $\Upsilon \succeq z z^T$. This last matrix inequality is equivalent \cite{boyd94} to the last constraint in the SDP formulation in (\ref{SDP_max_XY}). Finally, we obtain estimators $\hat{z}_1, \ldots,\hat{z}_l$ for the sensors by setting $\hat{z}_i= \text{sign}(x_{i}), \forall i=1,\ldots,l$.

\subsection{Comparison of algorithms for SYNC($\mathbb{Z}_2$)}
Figure \ref{fig:comp_sync_anch} compares the performance of the algorithms we proposed for synchronization in the presence of molecular fragment information. The adjacency graph of available pairwise measurements is an Erd\H{o}s-R\'{e}nyi graph $G(N,p)$ with $N=75$ and $p=0.2$ (i.e., a graph with $N$ nodes, where each edge is present with probability $p$, independent of the other edges). We show numerical experiments for four scenarios, where we vary the number of anchors $k=\{5,15,30,50\}$ chosen uniformly at random from the $N$ nodes. As the number of anchors increases, compared to the number of sensors $s=N-k$, the performance of the four algorithms is essentially similar. Only when the number of anchor nodes is small ($k=5$), the SDP-Y formulation shows superior results, together with SDP-XY and QCQP with constraint $z^T D z = \Delta$, while the QCQP with constraint $z^T z=s$
performs less well. In practice, one would choose the QCQP formulation with constraint $z^T D z = \Delta$, since the SDP based methods are computationally expensive as the size of the problem increases. This was also our method of choice for computing the reflection of the patches in the localization of the ubiqutin (PDG 1d3z), when molecular fragment information was available and the reflection of many of the patches was known a priori .

\begin{figure}[h]
\centering
\includegraphics[width=0.49\columnwidth]{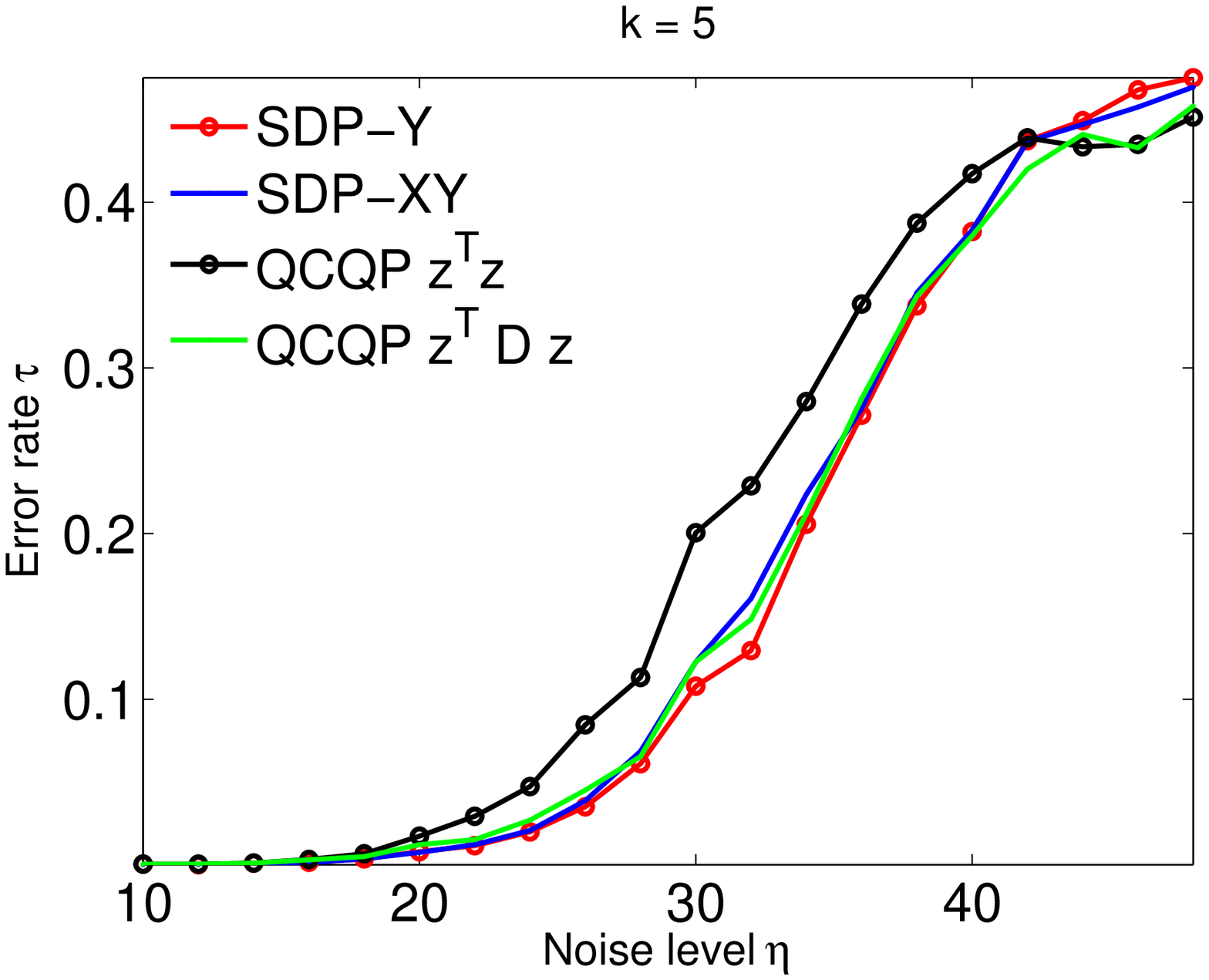}
\includegraphics[width=0.49\columnwidth]{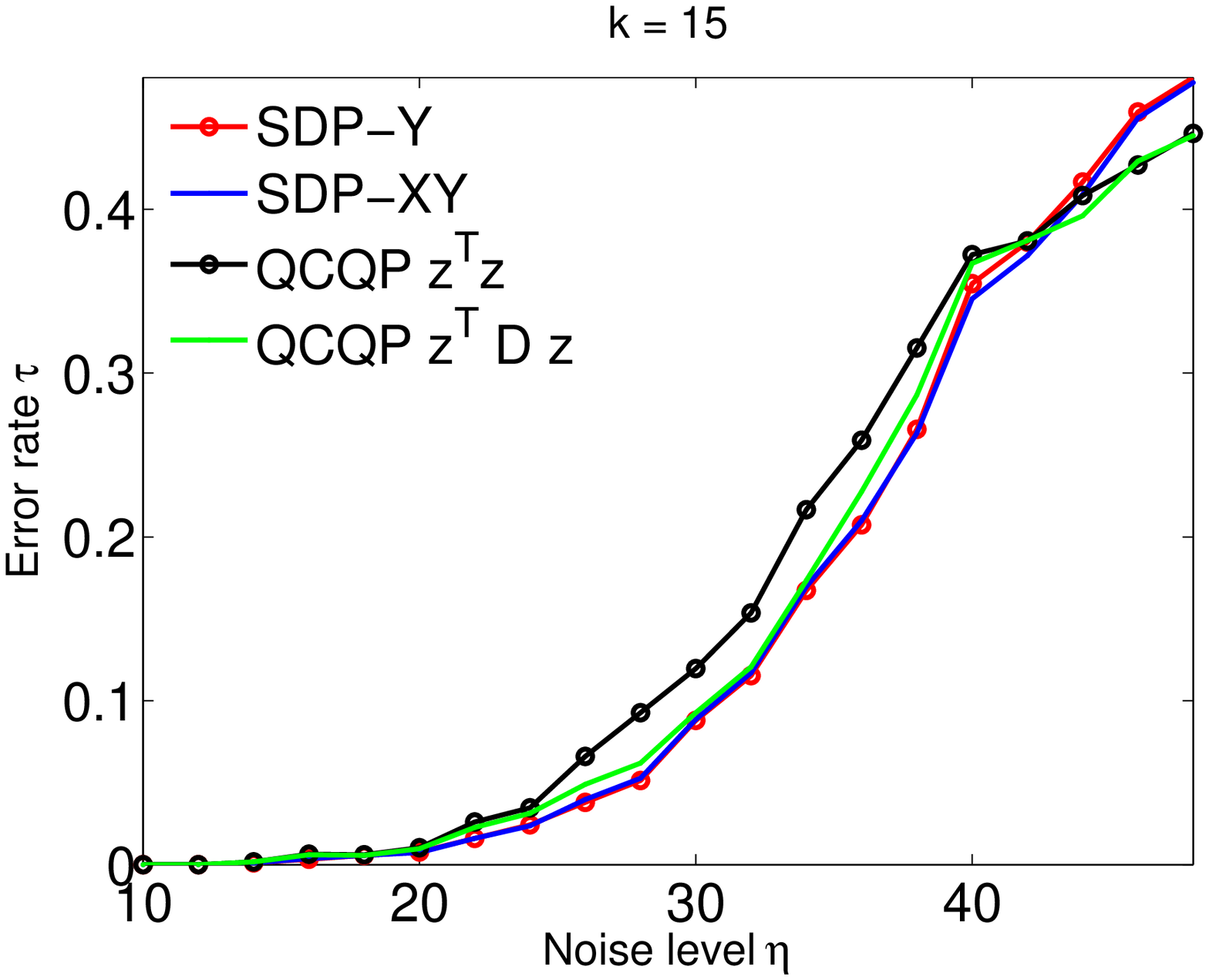}
\includegraphics[width=0.49\columnwidth]{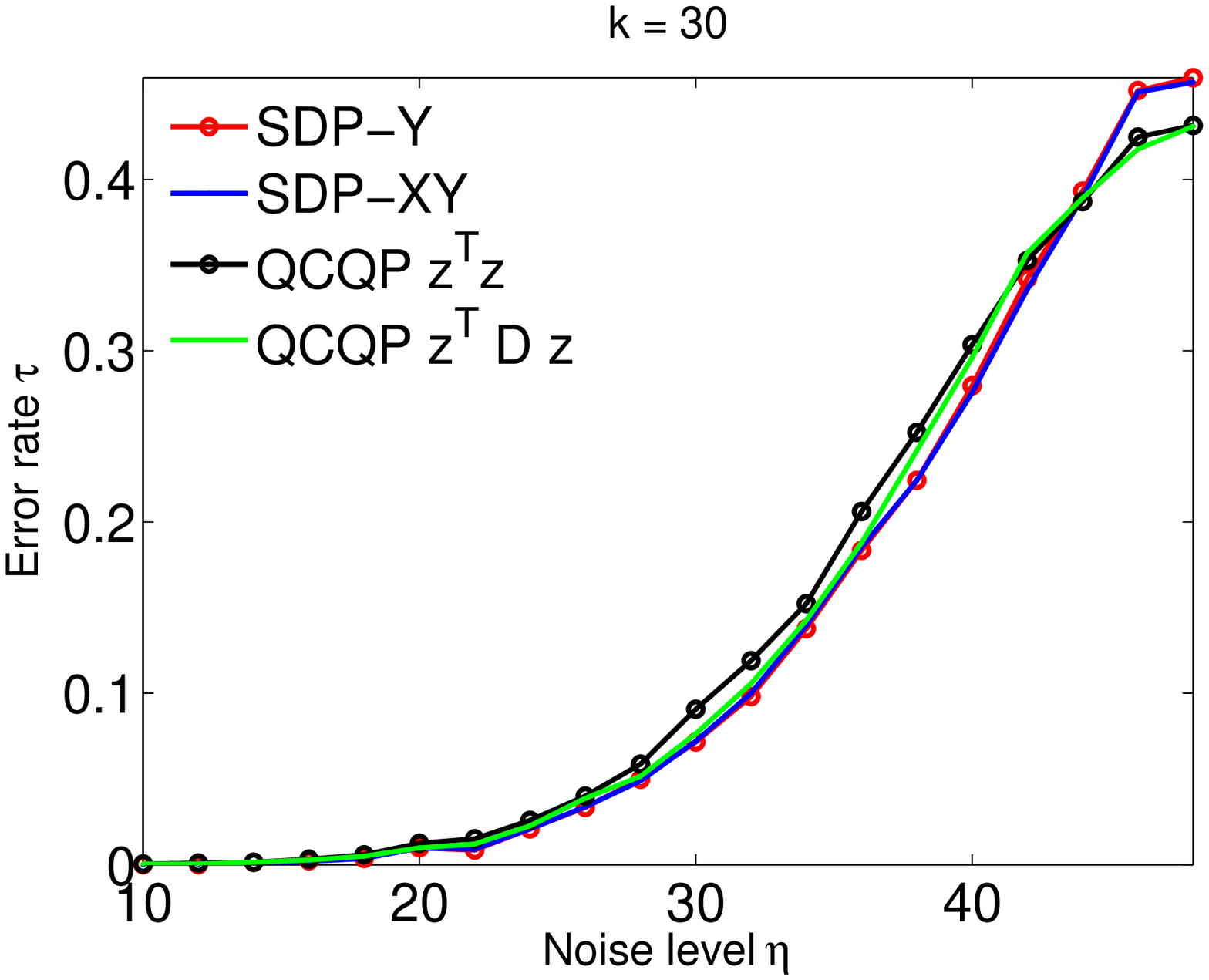}
\includegraphics[width=0.49\columnwidth]{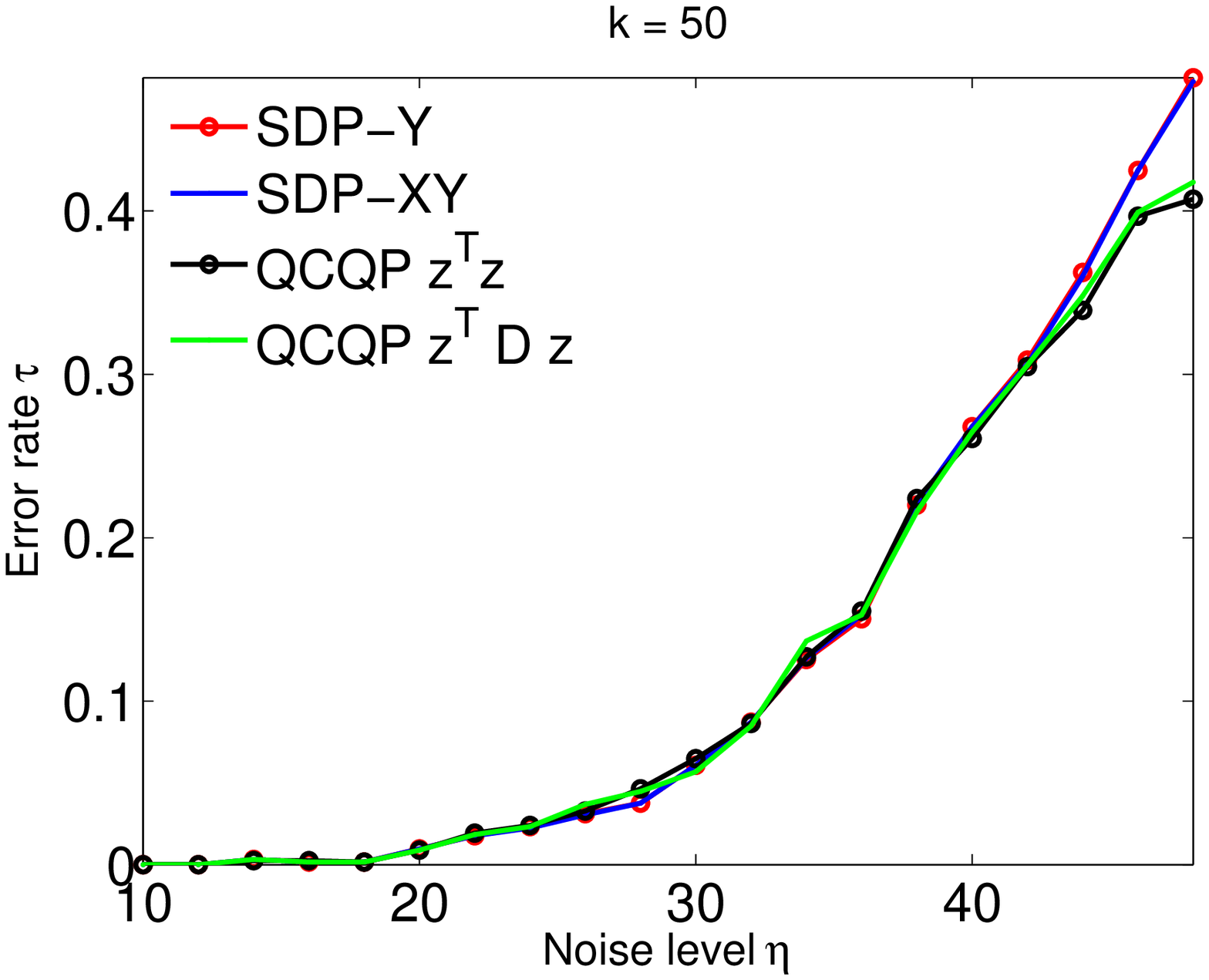}
\caption{ Comparison, in terms of robustness to noise, of the four algorithms we proposed for SYNC($\mathbb{Z}_2$): the two QCQP formulations using the two different constraints: $z^Tz=s$ and $z^T D z = \Delta$ as they appear in equations (\ref{maximization_zTz}) and (\ref{maximization_zTDz}), the two SDP-based formulations SDP-Y and SDP-XY as formulated in (\ref{SDP_max}) and (\ref{SDP_max_XY}). The adjacency graph of available pairwise measurements is an Erd\H{o}s-R\'{e}nyi  graph $G(N,p)$ with $N=75$ and $p=0.2$. Also, $k$ denotes the number of anchors, chosen uniformly at random from the $N$ nodes. Results are averaged over 50 runs.}
\label{fig:comp_sync_anch}
\end{figure}

\section{Experimental results}\label{numexp}

We have implemented our 3D-ASAP algorithm and compared its performance with other methods across a variety of measurement graphs, varying parameters such as the number of nodes, average degree (sensing radius) and level of noise.

In our experiments, the noise is multiplicative and uniform, meaning that to each true distance measurement $l_{ij} = \parallel p_i - p_j \parallel$, we add random independent noise $\epsilon_{ij}$ in the range $[-\eta l_{ij}, \eta l_{ij}]$, i.e.,
\begin{eqnarray}
   & d_{ij} = l_{ij} + \epsilon_{ij}   \nonumber \\
   & \epsilon_{ij} \sim Uniform(  [-\eta l_{ij}, \eta l_{ij}])
\label{eta}
\end{eqnarray}
The percentage noise added is $100 \eta$, (e.g., $\eta=0.1$ corresponds to $10\%$ noise).

The sizes of the graphs we experimented with range from $200$ to $1000$ nodes taking different shapes, with average degrees in the range $14-26$, and noise levels up to $50\%$. Across all our simulations, we consider the geometric graph model, meaning that all pairs of sensors within range $\rho$ are connected. We denote the true coordinates of all sensors by the $2\times n$ matrix $P=(p_1 \; \cdots \; p_n)$, and the estimated coordinates by the matrix $\hat{P} = (\hat{p}_1 \; \cdots\;  \hat{p}_n)$.  To measure the localization error of our algorithm we first factor out the optimal rigid transformation between the true embedding $P$ and our reconstruction $\hat{P}$, and then compute the following average normalized error (ANE)
\begin{equation}
ANE=\frac{ \sqrt{ \sum_{i=1}^n  \parallel p_i - \hat{p}_i \parallel^2}}
{\sqrt{ \sum_{i=1}^n  \parallel p_i - p_0\parallel^2} }
= \frac{\parallel  P-\hat{P}  \parallel_{F}}{\parallel P  - p_0 \mb{1}^T \parallel_{F}},
\label{AAvgNE}
\end{equation}
where $p_0 = \frac{1}{n}\sum_{i=1}^n p_i$ is the center of mass of the true coordinates. The normalization factor in the denominator of (\ref{AAvgNE}) ensures that the ANE is not only rigid invariant, but it is also scale free, i.e., it is invariant to scaling the underlying configuration by a constant factor.

The experimental results in the case of noiseless data (i.e. incomplete set of exact distances) should already give the reader an appreciation of the difficulty of the problem. As mentioned before, the graph realization problem is NP-hard, and the most we can aim for is an approximate solution. The main advantage of introducing the notion of a weakly uniquely localizable graph and using it for breaking the original problem into smaller subproblems, can be seen in the experimental results for noiseless data. Note that across all experiments, except for the PACM graph reconstructions, we are able to compute localizations which are at least one order of magnitude more accurate than what those computed by DISCO or SNL-SDP. Note that all algorithms are followed by a refinement procedure, in particular steepest descent with backtracking line search.
\begin{table}[h]
\begin{minipage}[b]{0.55\linewidth}
\centering
\begin{tabular}{|c||c|c|c|c|}
\hline
  $\eta$ & 3D-ASAP   & DISCO  & SNL-SDP   & MVU  \\
\hline\hline
   0\%  &     5e-4 &    8e-3 &    2e-3 &  2e-6\\
  10\%  &     0.04 &    0.06 &    0.04  & 0.07 \\
  20\%  &     0.07 &    0.10 &    0.09  & 0.12\\
  30\%  &     0.16 &    0.17 &    0.17  & 0.26\\
  40\%  &     0.19 &    0.29 &    0.29  & 0.50\\
  45\%  &     0.26 &    0.34 &    0.35  & 0.48\\
  50\%  &     0.32 &    0.42 &    0.43  & 0.61 \\
\hline
\end{tabular}
\caption{Reconstruction errors (measured in $ANE$) for the UNITCUBE graphs with $n=212$ vertices, sensing radius $\rho = 0.3$ and average degrees $deg = 17-25$. \vspace{3.4mm}
}
\label{tab:ANE_UNITCUBE}
\end{minipage} \;\;\;\;\;\;
\begin{minipage}[b]{0.4\linewidth}
\centering
\begin{tabular}{|c||c|c|c|}
\hline
  $\eta$ & 3D-ASAP   & DISCO  \\
\hline\hline
    0 \%  &   0.0001  &   0.0024   \\
   10 \%   &   0.0031  &   0.0029  \\
   20 \%  &   0.0052  &   0.0058   \\
   30 \%  &   0.0069  &   0.0078   \\
   35 \%  &   0.0130  &   0.0104   \\
   40 \%  &   0.0094  &   0.0107   \\
   45 \%  &   0.0168  &   0.0207   \\
   50 \%  &   0.0146  &   0.0151   \\
\hline
\end{tabular}
\caption{Reconstruction errors (measured in $ANE$) for the 1d3z molecule graph with $n=1009$ vertices,
sensing radius $\rho = 5$ angstrom and average degree $deg = 14$.
}
\label{tab:ANE_1d3z}
\end{minipage}
\end{table}

As expected, FAST-MVU performs at its best when the data is a set of random points uniformly distributed in the unit cube. In this case, the top three eigenvectors of the Laplacian matrix used by FAST-MVU  provide an excellent approximation of the original coordinates. Indeed, as the experimental results in Table \ref{tab:ANE_UNITCUBE} show, the FAST-MVU algorithm gives the best precision across all algorithms, returning an $ANE=2e-6$. However, in the case of noisy data, the performance of FAST-MVU degrades significantly, making FAST-MVU the least robust to noise algorithm tested on the UNITCUBE graph. Furthermore, FAST-MVU performs extremely poor on more complicated topologies, even in the absence of noise, as it can be seen in the reconstructions of the PACM and BRIDGE-DONUT graphs shown in Figure \ref{fig:MVU_examples}. For comparison, the original underlying graphs are shown in Figures \ref{fig:RECS_PACM} and \ref{fig:RECS_Donut}. 

\begin{figure}[h]
\begin{center}
\subfigure[PACM  graph with noise $\eta=0\%$]{\includegraphics[width=0.67\columnwidth]{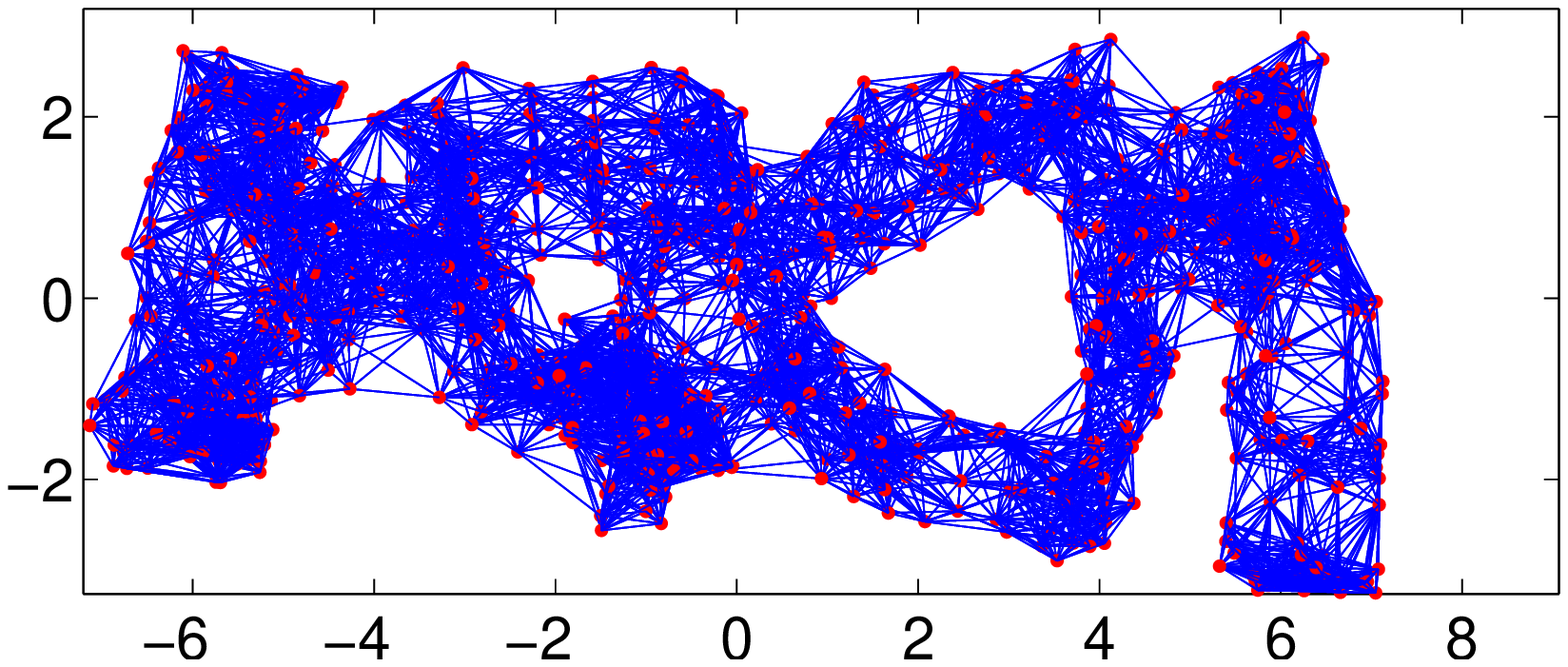}}
\subfigure[DONUT graph with noise $\eta=0\%$]{\includegraphics[width=0.30\columnwidth]{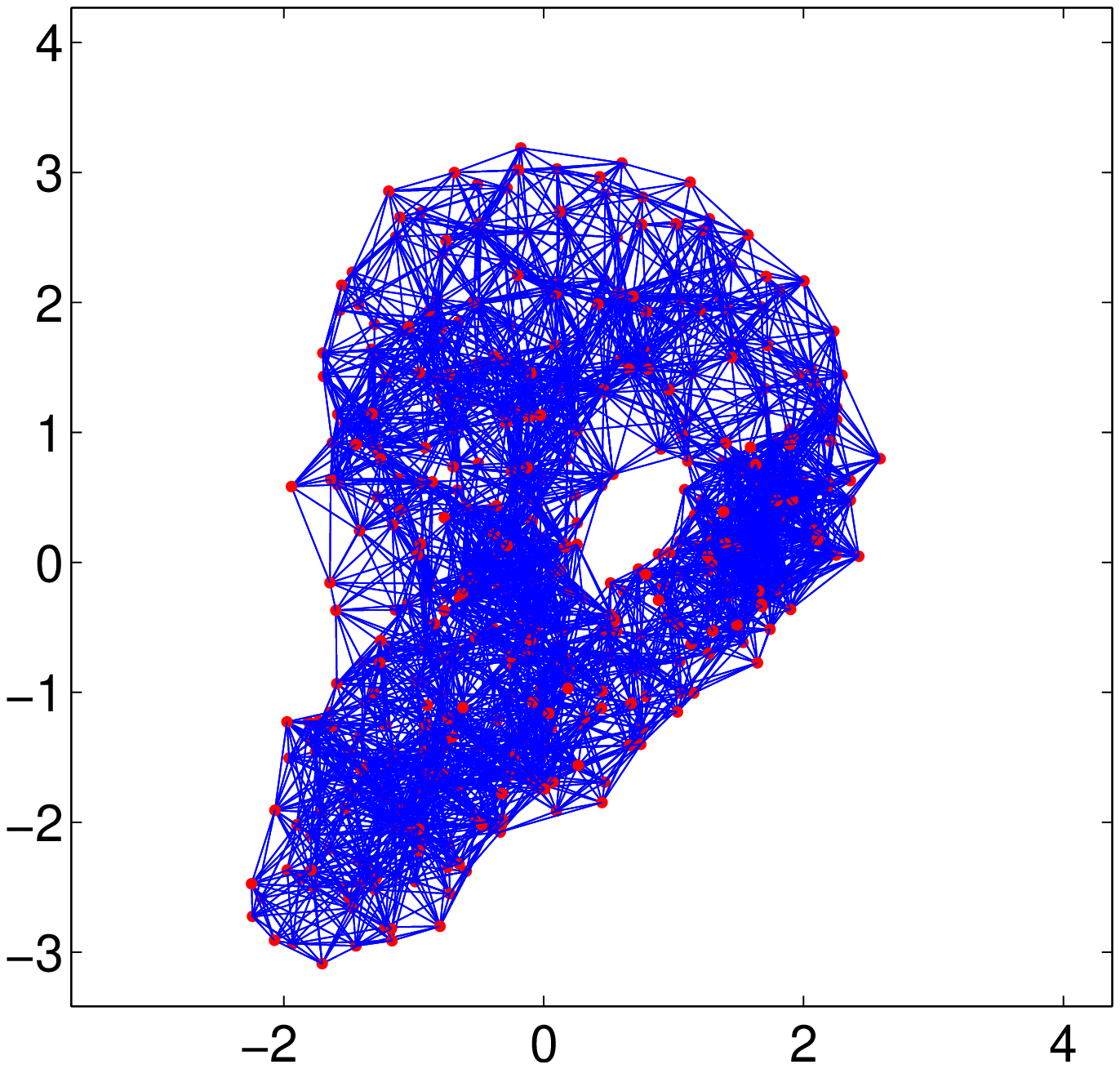}}
\end{center}
\vspace{-5mm}
\caption{ Reconstructions of the PACM and BRIDGE-DONUT graphs using the FAST-MVU algorithm, in the case of noiseless data (compare to the original measurement graphs shown in Figures \ref{fig:RECS_PACM} and \ref{fig:RECS_Donut}.)}
\label{fig:MVU_examples}
\end{figure}

The second data set used in our experiments is a synthetic example of the molecule problem, and represents the ubiqutin protein represented by PDB entry 1d3z \cite{cornilescu}, with 1288 atoms, 686 of which are hydrogen atoms. However, after removing the nodes of degree less than 4,  $n=1009$ atoms remain. For this data set, the available distances can be divided into groups: the first group contains distances inferred from known bond length and torsion angles (``good" edges) which are kept accurate across different levels of noise, and the second group contains distances (NOE edges) which are perturbed with noise $\eta$ at each stage of the experiment. On average, each node is incident with $6$ good edges and $8$ NOE edges, thus the average degree of the entire measurement graph is $deg=14$. In addition, besides knowledge of good and NOE edges, we also use  information from molecular fragments  in the form of subgroups of atoms (molecular fragments) whose coordinates are readily available. We incorporate 280 such molecular fragments, whose average size is close to 5 atoms. To exploit such information, we slightly modify the approach to break up the measurement graph and synchronize the patches. We build patches from molecular fragments by selecting a WUL subgraph from each extended molecular fragment, where by an extended molecular fragment we mean the graph resulting from a molecular fragment and all its 1-hop neighbors. Once we extract patches from all extended molecular fragments, we consider all  remaining nodes (singletons) that are not contained in any of the patches obtained so far. Depending on the noise level, the number of such singleton nodes is between $10$ and $15$.
Finally, we extract WUL patches from the 1-hop neighborhoods of the singleton nodes, following the approach of Section \ref{WULS}. Since  we know a priori  the localization, and in particular the reflection of each molecular fragment, and the reflection of a molecular fragment induces a reflection on the patch that contains it, we have readily available information on the reflection of all patches that contain molecular fragments. In other words, we are solving a synchronization problem over $\mathbb{Z}_2$ when molecular fragment information is available. Using the synchronization method introduced in Section \ref{sub:qcqp}, we compute the reflection of the remaining patches. In terms of the reconstruction error, 3D-ASAP and DISCO return comparable results, slightly in favor of our algorithm. For noiseless NOE distances, we compute a localization that is one order of magnitude more accurate than the solution returned by DISCO.

Another graph that we tested is the BRIDGE-DONUT graph shown in Figure \ref{fig:RECS_Donut}. It has $n=500$ nodes, sensing radius $\rho=0.92$, and average degree in the range $18-25$, depending on the noise level. Table \ref{tab:ANE_DONUT} contains the reconstruction errors for various levels of noise
showing that 3D-ASAP consistently returns more accurate solutions than DISCO.
For the BRIDGE-DONUT graph, we included in our numerical simulations the spectral partitioning 3D-SP-ASAP algorithm, which performed just as well as the 3D-ASAP algorithm (while significantly reducing the running time as shown in Tables \ref{tab:times} and \ref{tab:SP_times}), and consistently outperformed the DISCO algorihtm at all levels of noise.
Note that 3D-SP-ASAP returns very poor localizations when using only $k=8$ partitions; in this case the extended partitions have large size (up to 150 nodes) and SNL-SDP fails to embed them accurately, even at small levels of noise, as shown in Table \ref{tab:ANE_DONUT}. By increasing the number of partitions from $k=8$ to $k=25$ and thus decreasing the size of each partition (and also of the extended partitions), the running time of 3D-SP-ASAP increases slightly from 140 to 186 seconds (at $35\%$ noise), but we are now able to match the accuracy of 3D-ASAP which requires almost ten times more running time (1770 seconds).

Finally, for the PACM graphs in Figure \ref{fig:RECS_PACM}, the network takes the shape of the letters $P,A,C,M$ that form a connected graph on $n=800$ vertices. The sensing radius is $\rho = 1.2$ and the average degree in the range $deg \approx 21-26$. This graph was particularly useful in testing the sensitivity of the algorithm to the topology of the network. In Table \ref{tab:ANE_PACM}, we show the reconstruction errors for various levels of noise
DISCO returns better results for $\eta = 0\%, 10\%, 40\%$,  but less accurate for $\eta=20\%, 30\%, 35\% $. We believe that DISCO returns results comparable to 3D-ASAP (unlike the previous three graphs) because the topology of the PACM graph favors the graph decomposition used by DISCO. Note that at $\eta=0\%$ noise, the 3D-ASAP reconstruction differs from the original embedding mainly in the right handside of the letter M, which is loosely connected to the rest of the letter, and also parts of its underlying graph are very sparse, thus rendering the SDP localization algorithms less accurate. This can also be seen in the spectral partitioning of the PACM graph shown in Figure \ref{fig:spectral_part_PACM}, where for $k$ (number of clusters) as low as 3 or 4, the right side of letter M is picked up as an individual cluster. At higher levels of noise 3D-SP-ASAP proves to be more accurate than DISCO, and it returns results comparable with ASAP.

\begin{table}[h]
\begin{minipage}[b]{0.9\linewidth}
\centering
\begin{tabular}{|c||c|c|c|c|c|}
\hline
  $\eta$ &     3D-ASAP &  3D-SP-ASAP$_{k=8}$& 3D-SP-ASAP$_{k=25}$ & DISCO    \\
\hline\hline
    0\%  &     6e-4 &  4e-4 & 2e-4 		 & 4e-3  \\
   10\%  &     0.04 &  0.59 & 0.03	 & 0.03   \\
   20\%  &     0.05 &  0.48 & 0.04	 & 0.07   \\
   30\%  &     0.11 &  0.52 & 0.14	 & 0.30   \\
   35\%  &     0.12 &  0.33 & 0.15	 & 0.21    \\
   40\%  &     0.24 &  0.42 & 0.19	 & 0.27   \\
   45\%  &     0.28 &  0.71 & 0.32	 & 0.35   \\
   50\%  &     0.23 &  0.35 &  0.35	 & 0.40    \\
\hline
\end{tabular}
\caption{Reconstruction errors (measured in $ANE$) for the BRIDGE-DONUT graph with $n=500$ vertices, sensing radius $\rho = 0.92$ and average degrees $deg = 18-25$. We used $k$ to denote the number of partitions of the vertex set. For $\eta=0$ we embed each of the $k$ patches (extended partitions) using FULL-SDP, while for $\eta > 0$ we used the SNL-SDP algorithm. At $\eta=50\%$,  3D-SP-ASAP$_{k=8}$ localizes only 435 out of the 500 nodes.}
\label{tab:ANE_DONUT}
\end{minipage}
\;\;\;
\begin{minipage}[b]{0.90\linewidth}
\centering
\begin{tabular}{|c||c|c|c|c|}
\hline
  $\eta$ & 3D-ASAP &  3D-SP-ASAP$_{k=40}$   &   DISCO    \\
\hline\hline
    0\%  &    0.05 &  0.05 & 0.02  \\
   10\%  &    0.08 & 0.12  & 0.07  \\
   20\%  &    0.07 & 0.16  & 0.18  \\
   30\%  &    0.27 &  0.15 & 0.45  \\
   35\%  &    0.18 &  0.32 & 0.28  \\
   40\%  &    0.48 &  0.20 & 0.26  \\
\hline
\end{tabular}
\caption{Reconstruction errors (measured in $ANE$) for the PACM graph with $n=800$ vertices, sensing radius $\rho = 1.2$ and average degrees $deg = 21-26$. For 3D-SP-ASAP we used $k=40$ partitions. Note that even for the noiseless case the error is not negligible due to incorrect embeddings of subgraphs that are contained in the right leg of the letter M.   \;\;\;\;\;\;\;\;\; \;\;\;\;\;\;\;\;\;\;\;\;\;\;\;\;\;  \;\;\;\;\;\;\;\;
}
\label{tab:ANE_PACM}
\end{minipage}
\end{table}

\begin{figure}[ht]
\renewcommand{\arraystretch}{1.3}
\centering
\begin{tabular}{@{\hspace{0.1cm}} m{0.04\columnwidth}@{\hspace{0.2cm}}|@{\hspace{0.2cm}}
m{0.45\columnwidth}@{\hspace{0.1cm}} m{0.45\columnwidth}@{\hspace{0.1cm}}  }
 \ \ $\eta$ &   \quad\quad\quad\quad\quad 3D-ASAP  & \quad\quad\quad\quad\quad DISCO   \\ \hline
0\% &
\includegraphics[width=0.43\columnwidth]{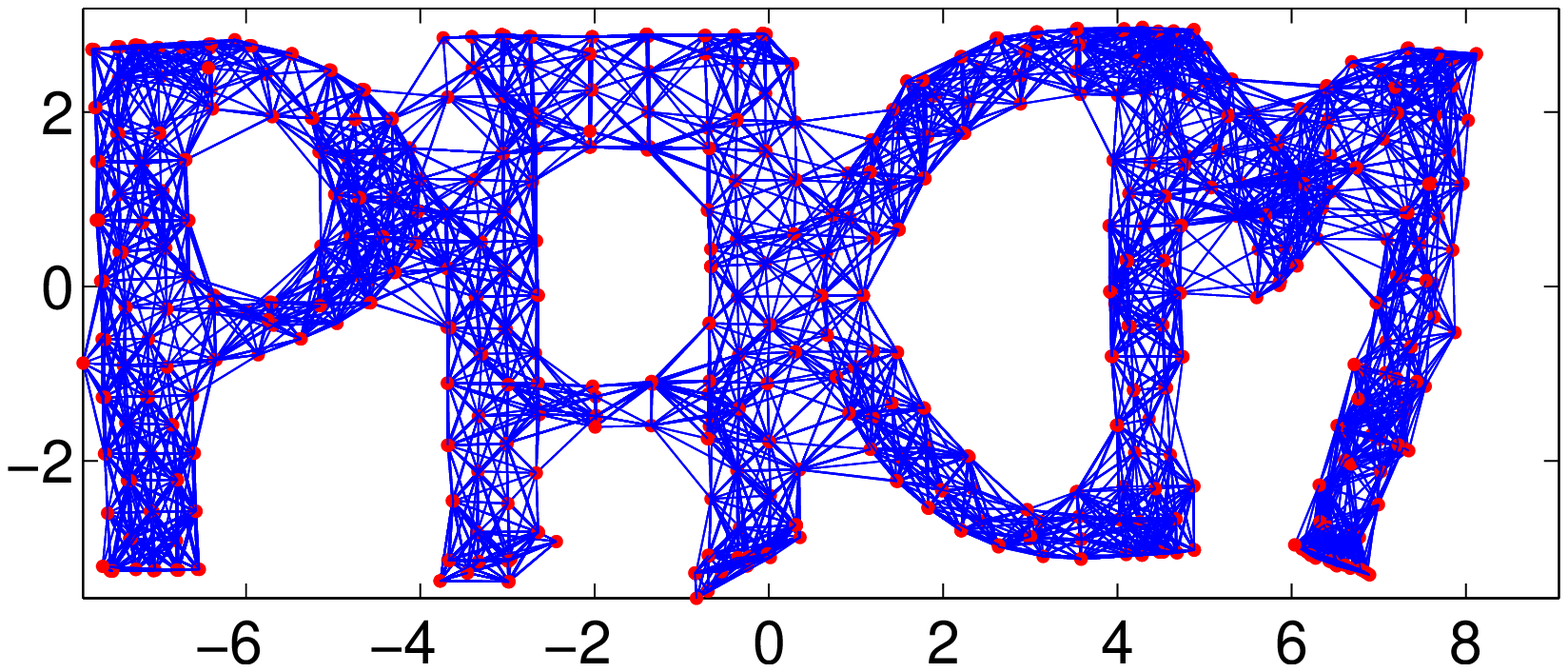}   &
\includegraphics[width=0.43\columnwidth]{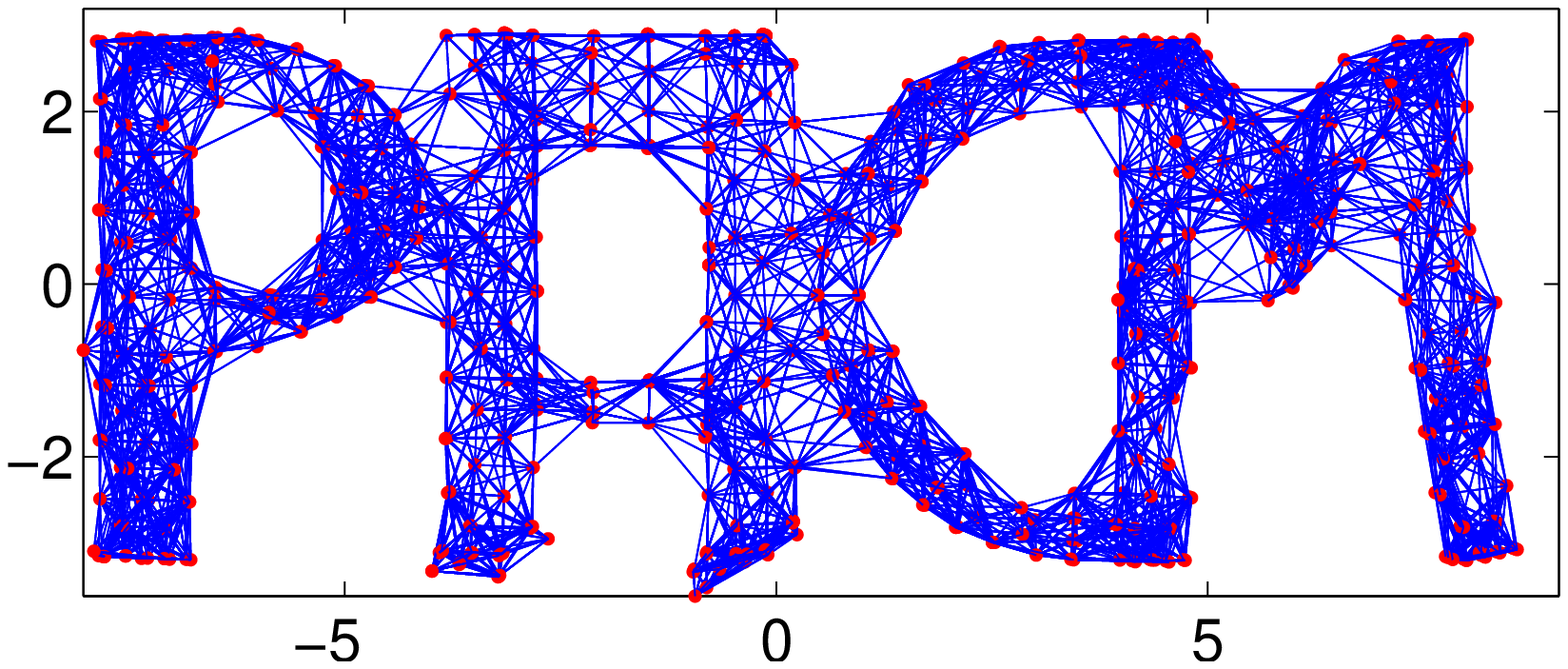}  \\
10\% &
\includegraphics[width=0.43\columnwidth]{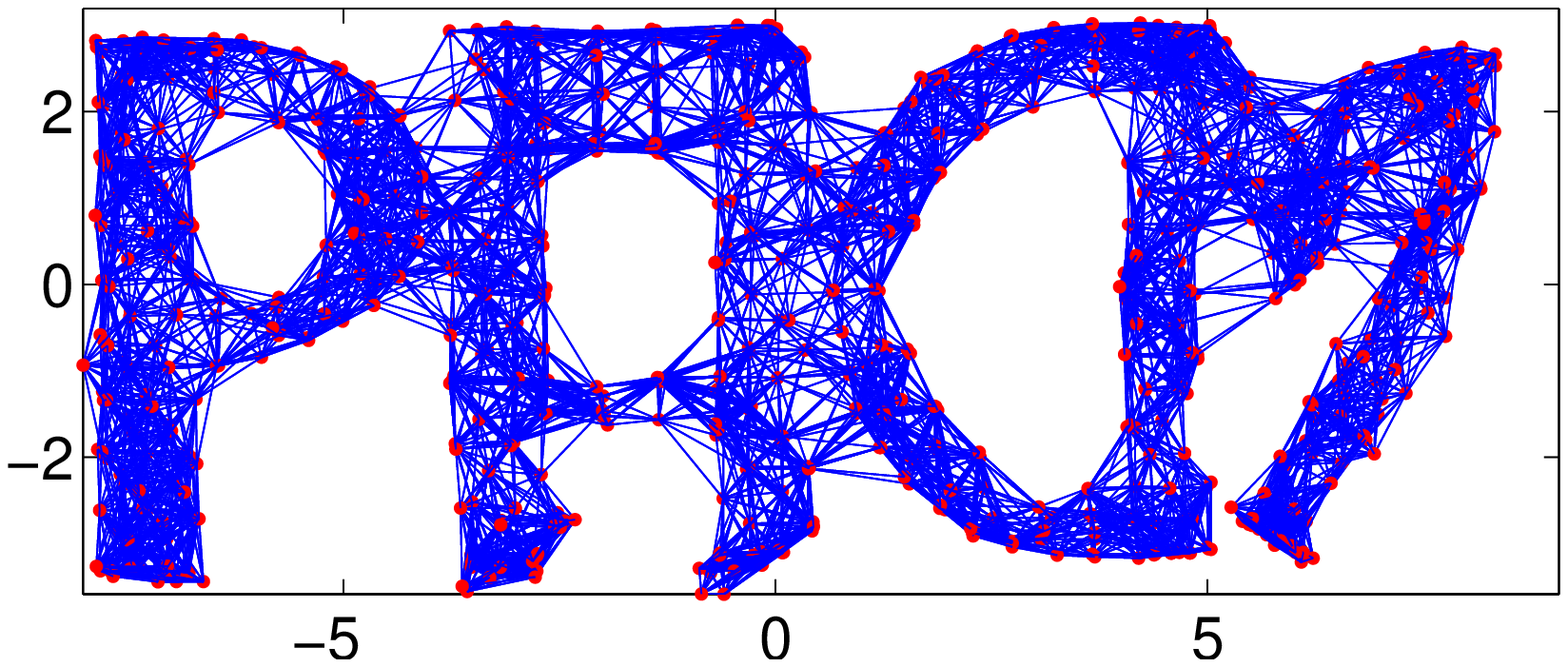}  &
\includegraphics[width=0.43\columnwidth]{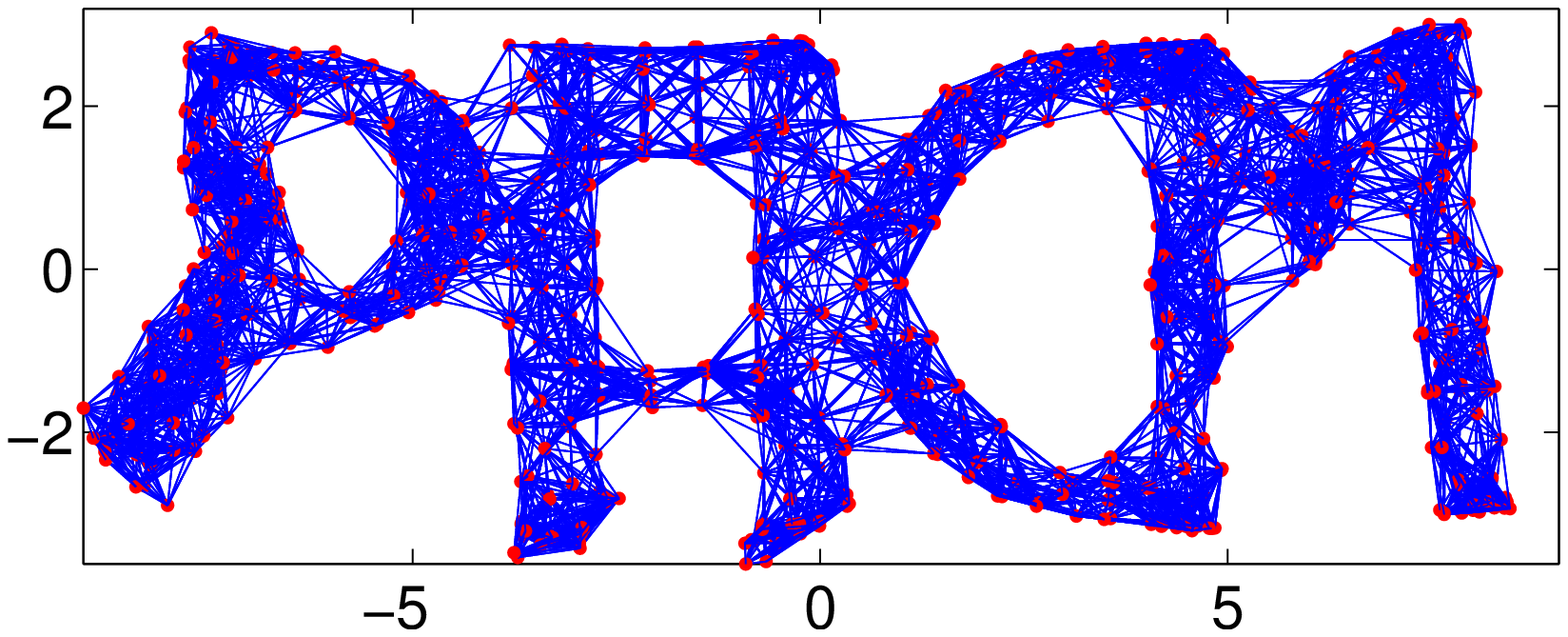}   \\
 25\% &
 \includegraphics[width=0.43\columnwidth]{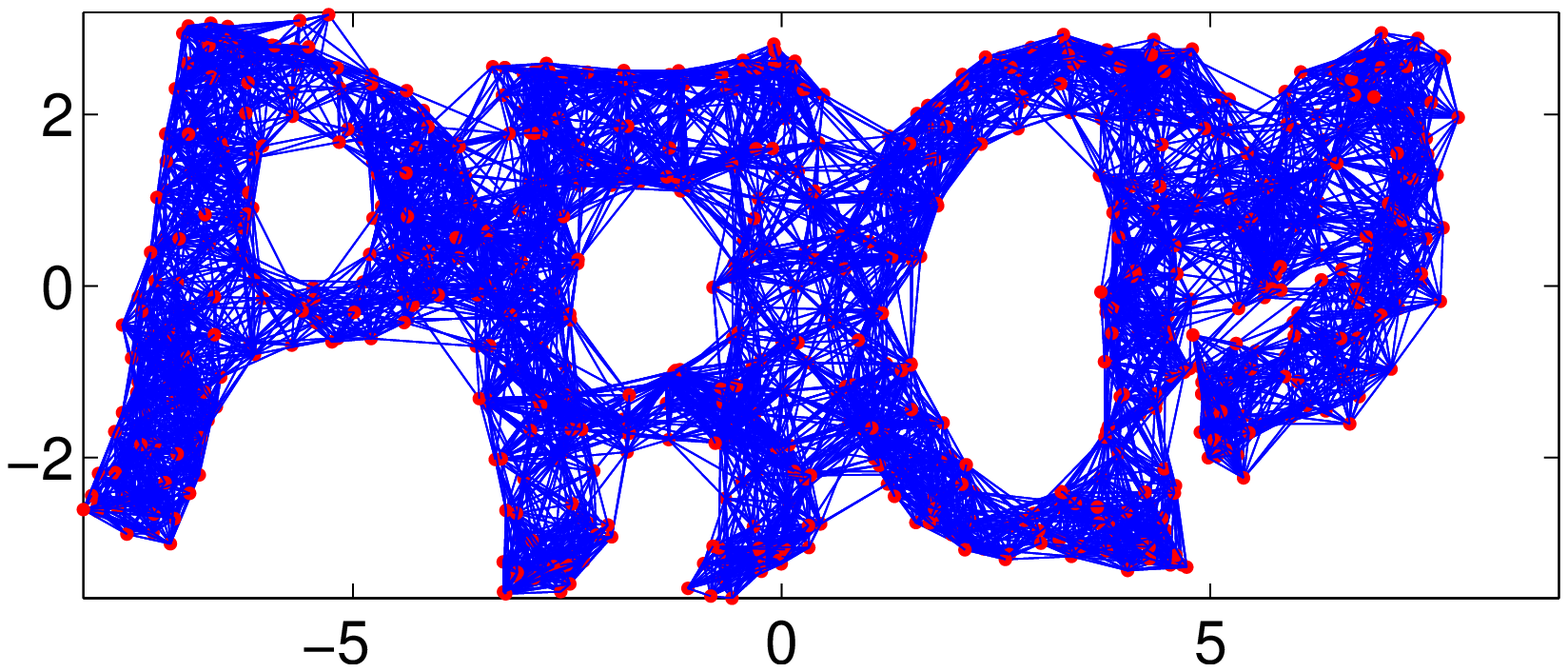} &
 \includegraphics[width=0.43\columnwidth]{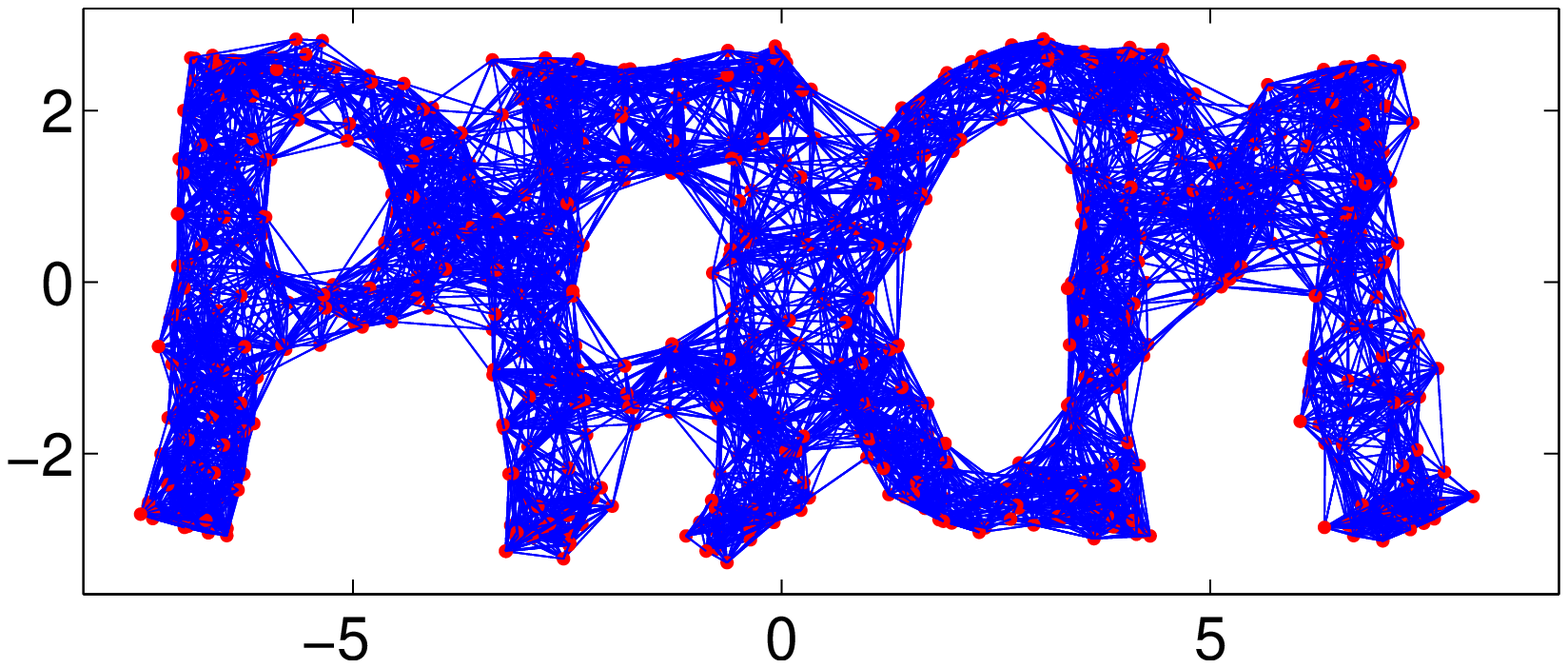}   \\
30\% &
\includegraphics[width=0.43\columnwidth]{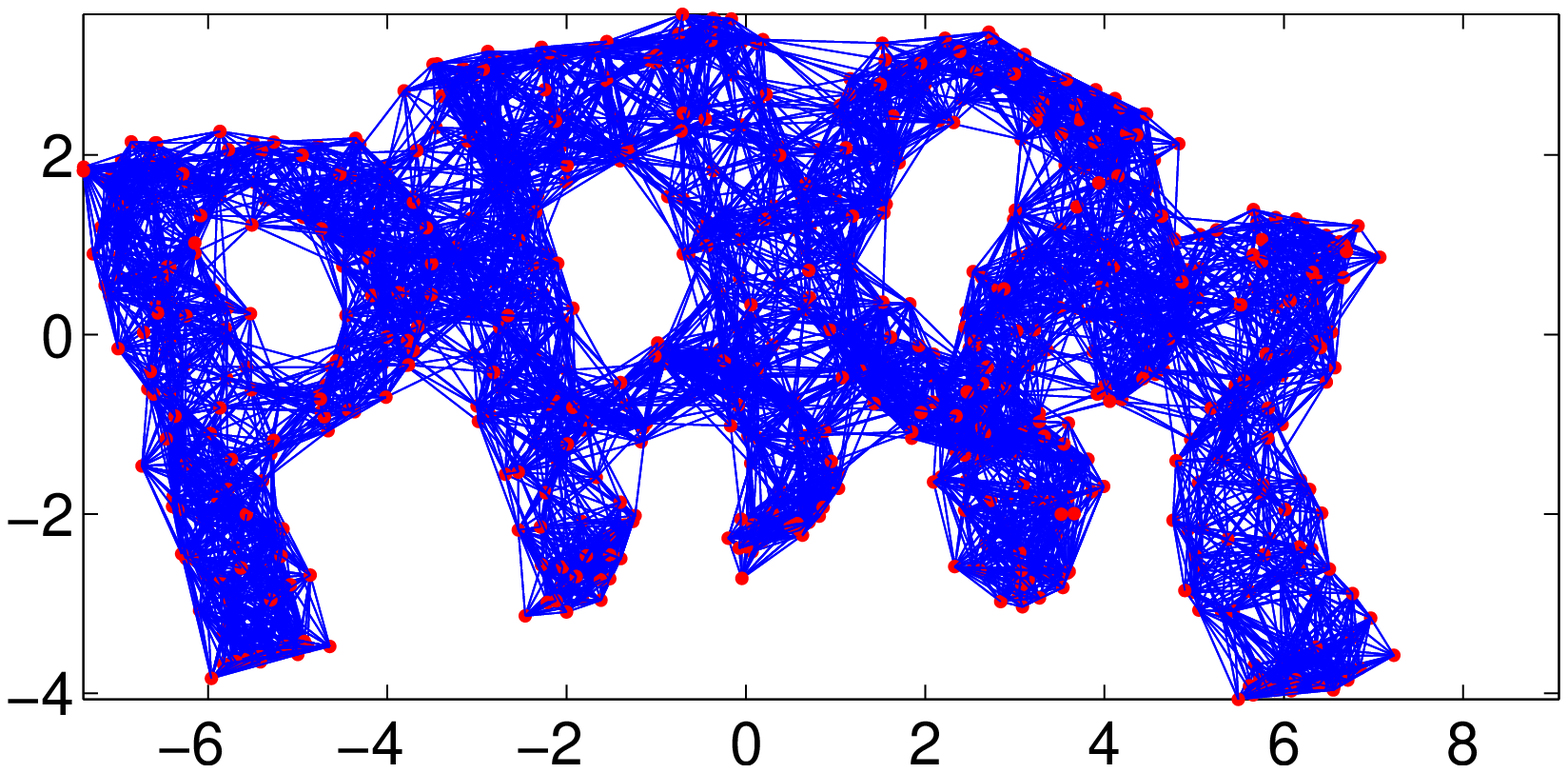} &
\includegraphics[width=0.43\columnwidth]{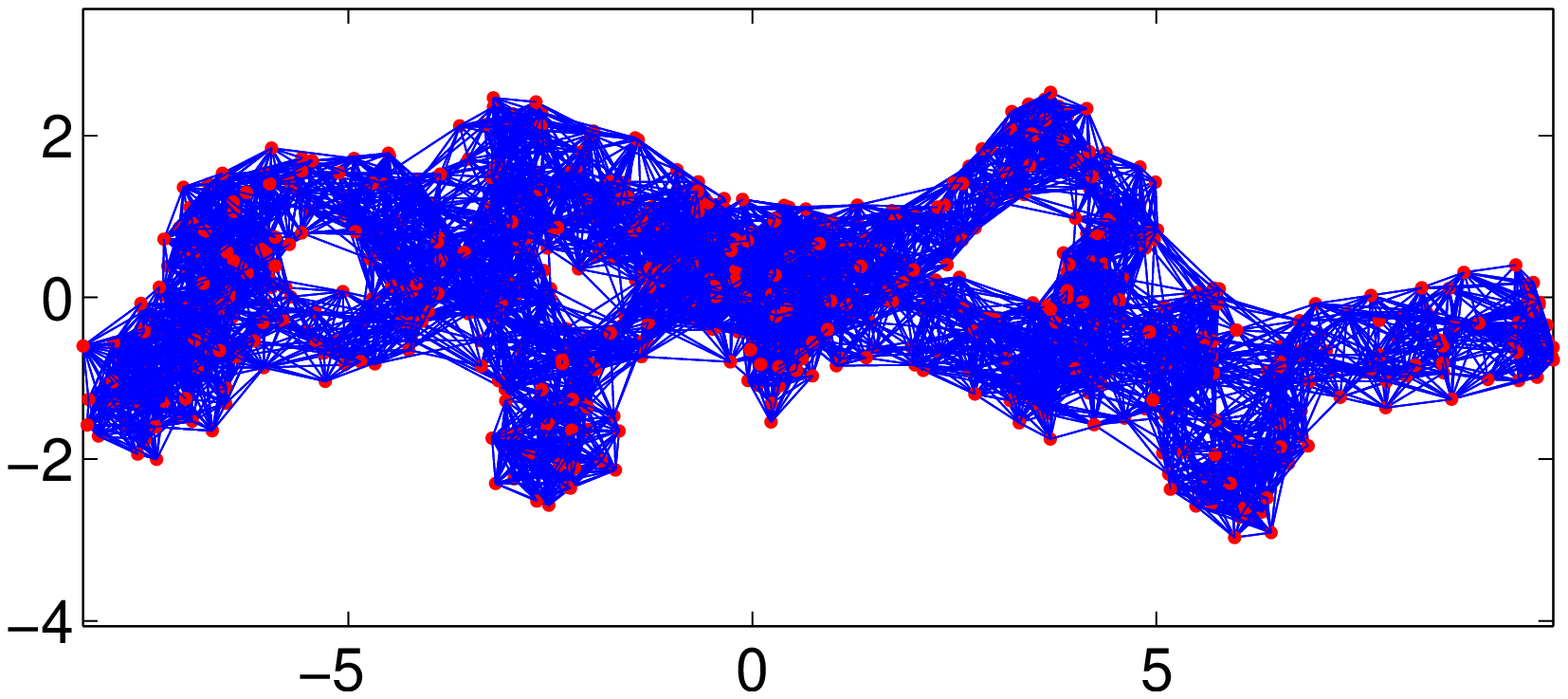}   \\
35\% &
\includegraphics[width=0.43\columnwidth]{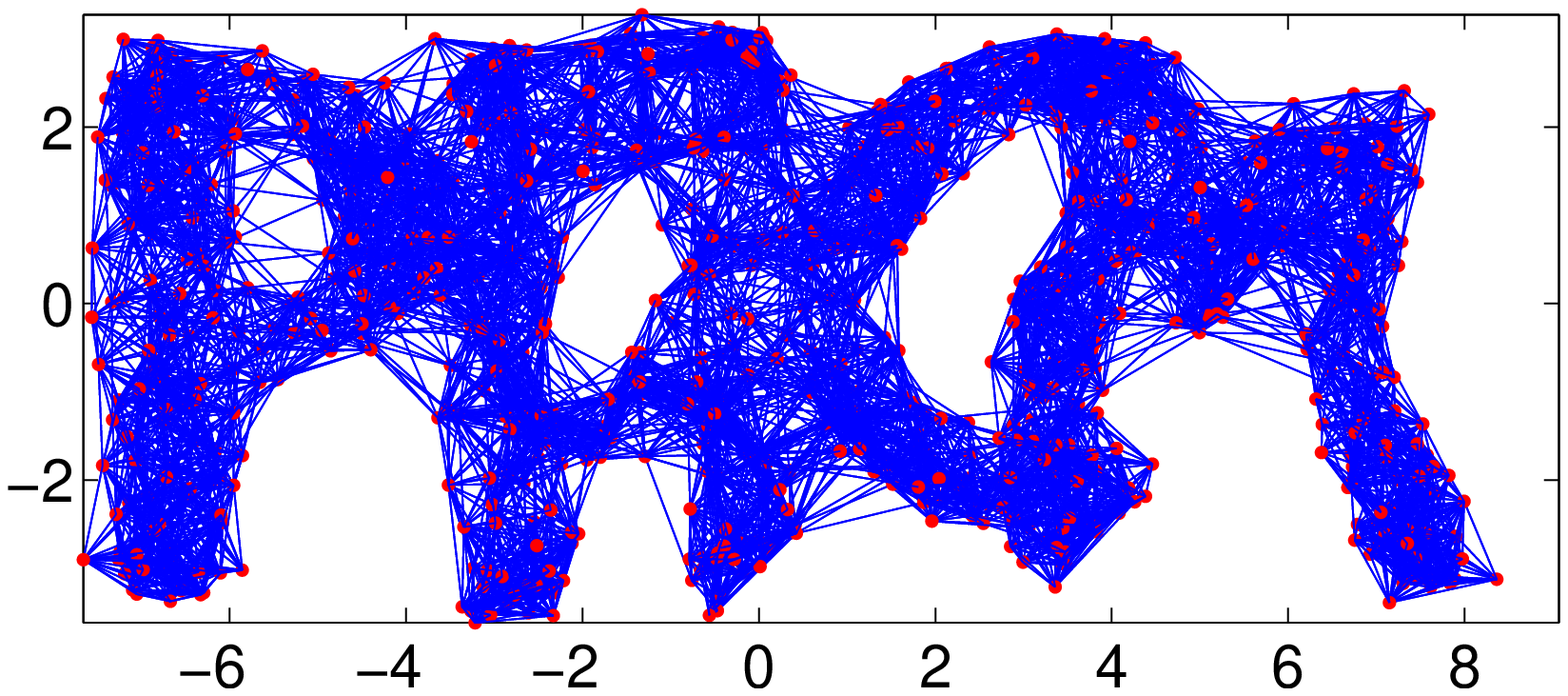} &
\includegraphics[width=0.43\columnwidth]{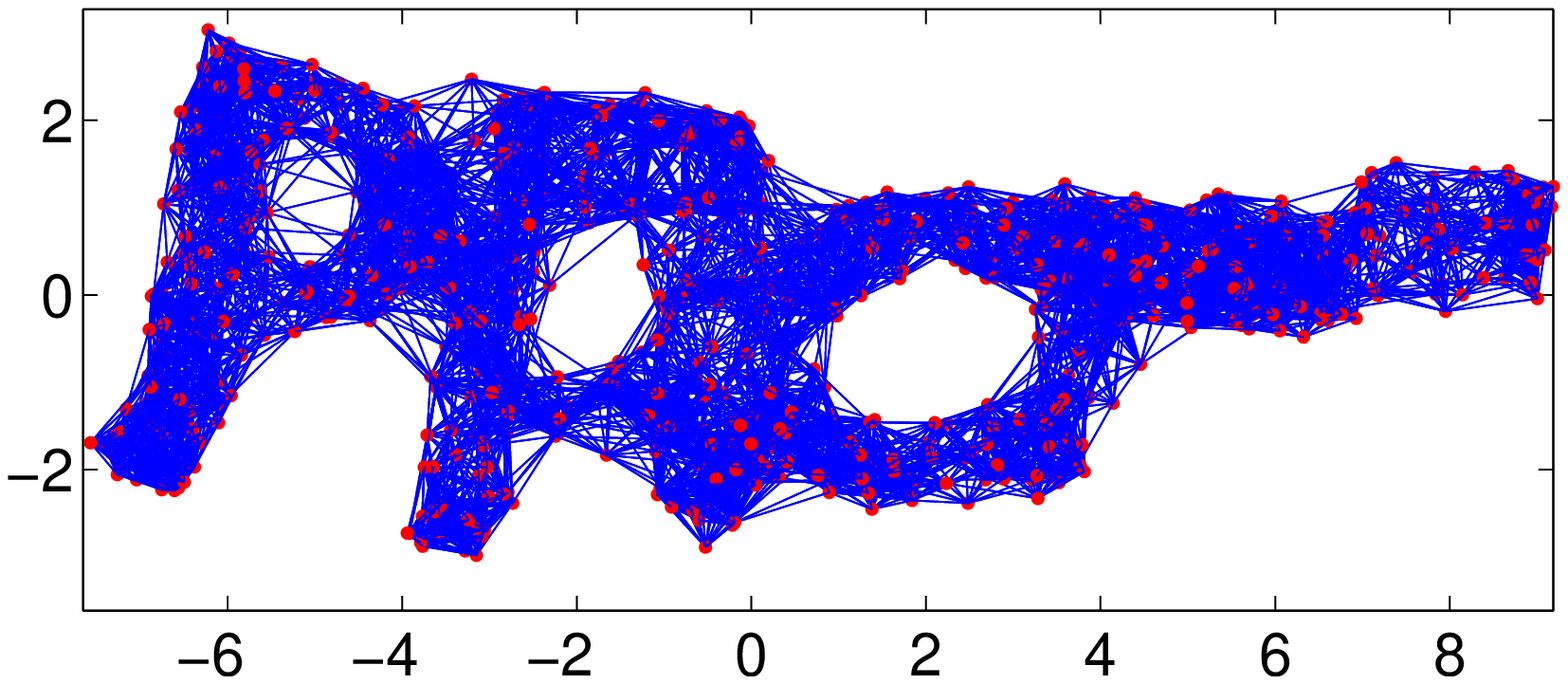}   \\
40\% &
\includegraphics[width=0.43\columnwidth]{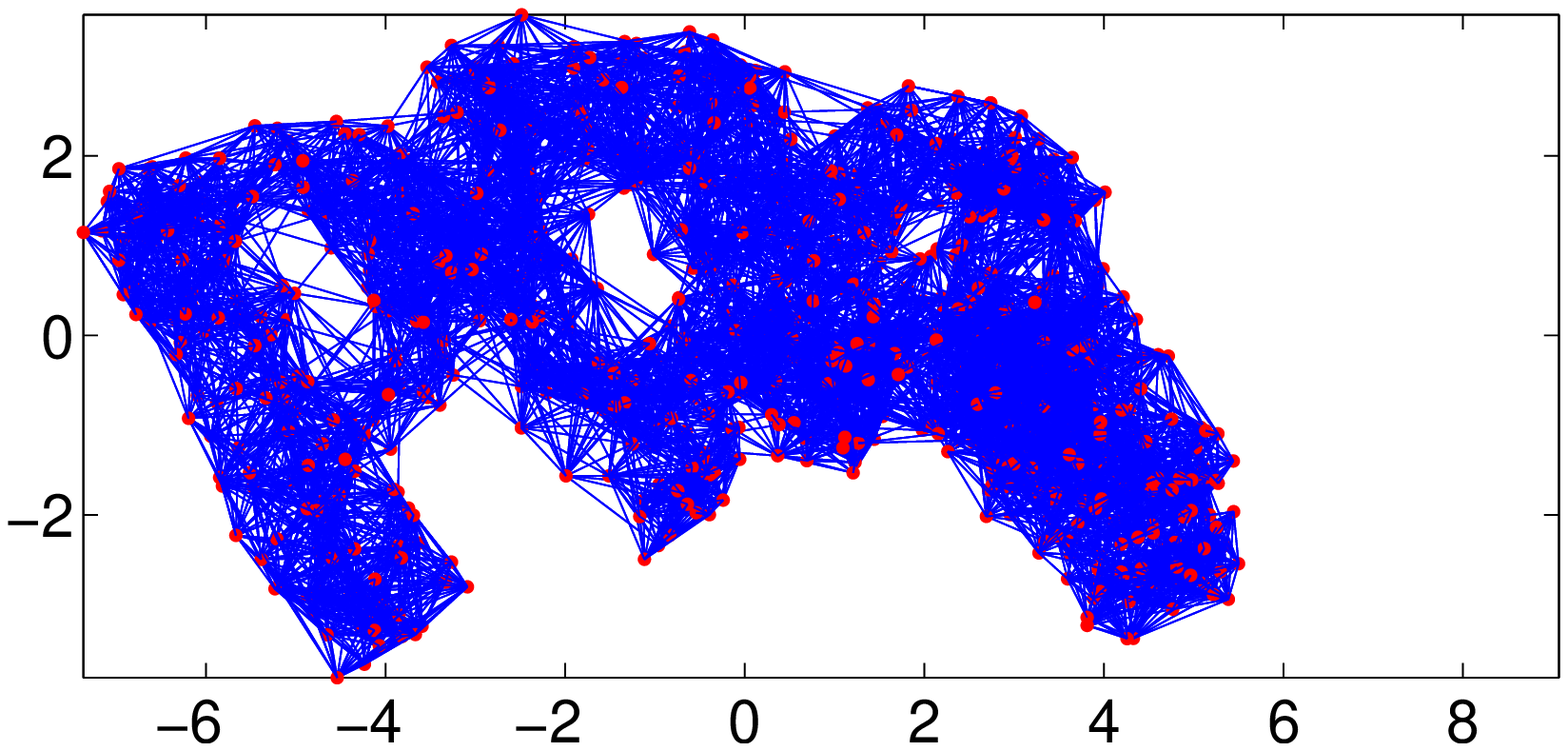} &
\includegraphics[width=0.43\columnwidth]{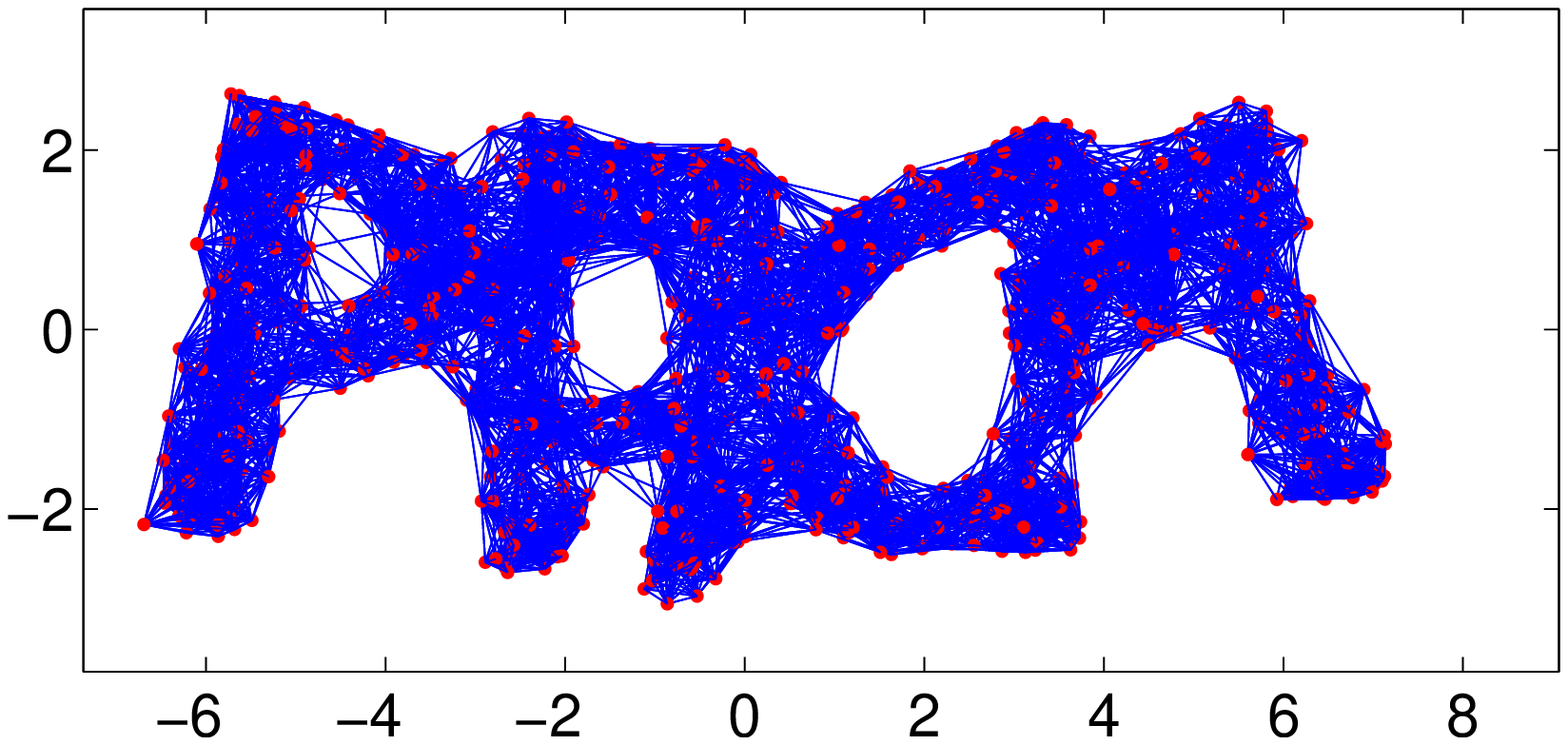}   \\
\end{tabular}
\caption{Reconstructions of the PACM graph with $n=800$ nodes, sensing radius $\rho=1.2$ and $\eta = 0\%, 10\%, 25\%, 30\%, 35\%, 40\%$.}
\label{fig:RECS_PACM}
\end{figure}

\begin{figure}[ht]
\renewcommand{\arraystretch}{1.3}
\centering
\begin{tabular}{@{\hspace{0.1cm}} m{0.04\columnwidth}@{\hspace{0.2cm}}|@{\hspace{0.2cm}}
m{0.45\columnwidth}@{\hspace{0.1cm}} m{0.45\columnwidth}@{\hspace{0.1cm}} }
 \ \ $\eta$ &   \quad\quad\quad\quad\quad 3D-ASAP  & \quad\quad\quad\quad\quad DISCO \\ \hline
0\% &
\includegraphics[width=0.40\columnwidth] {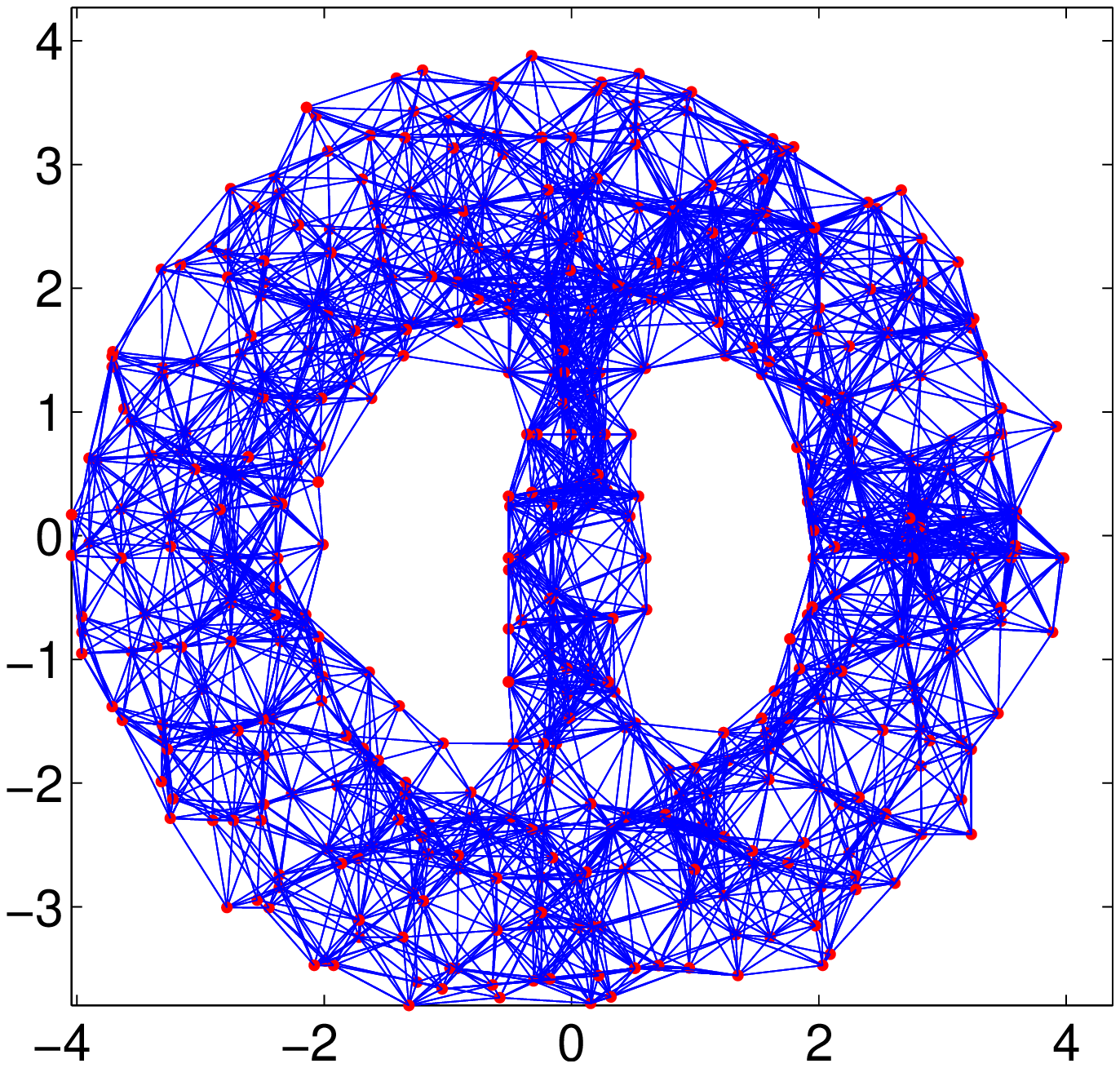} &
\includegraphics[width=0.40\columnwidth] {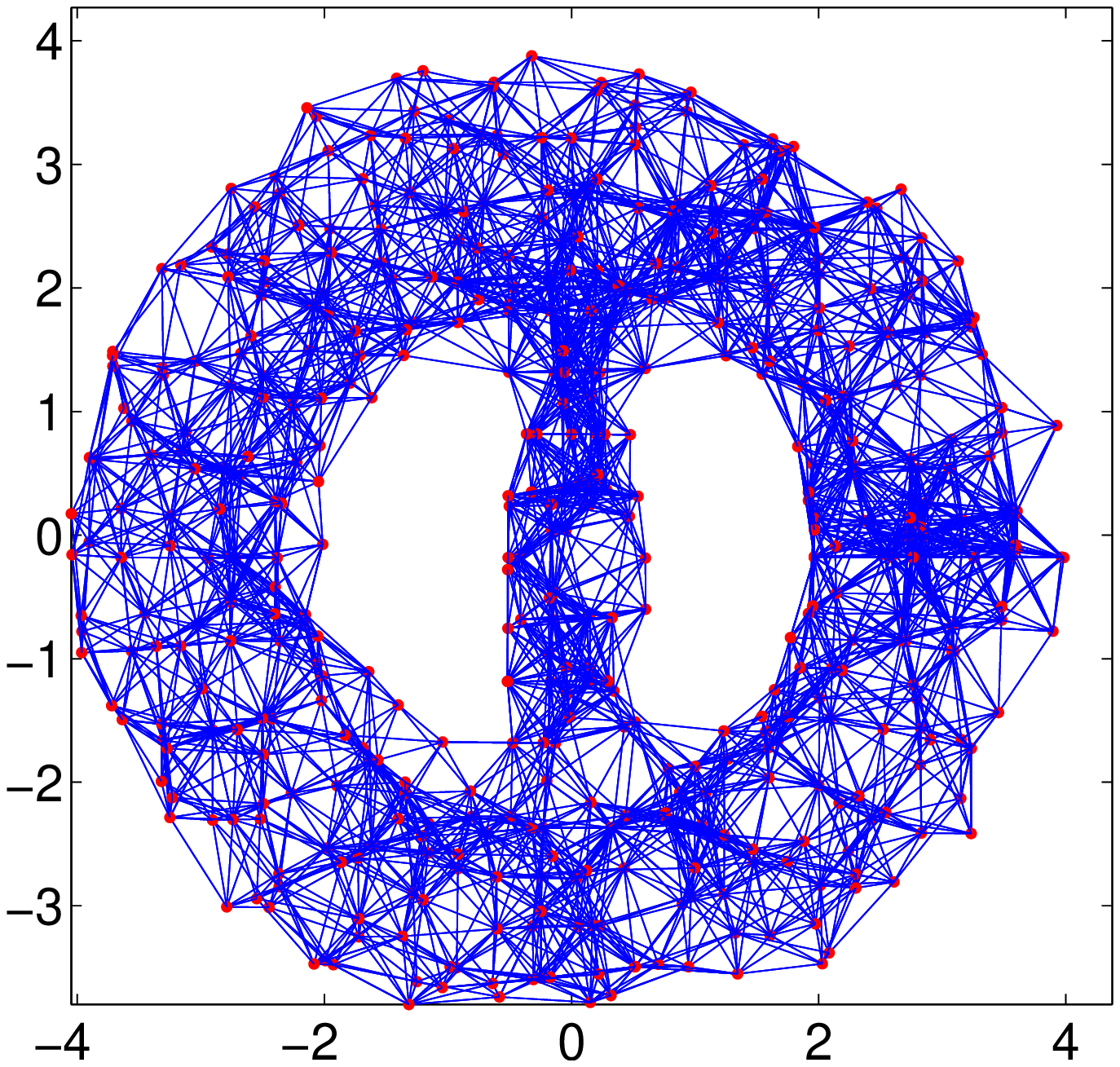}  \\
20\% &
\includegraphics[width=0.40\columnwidth ] {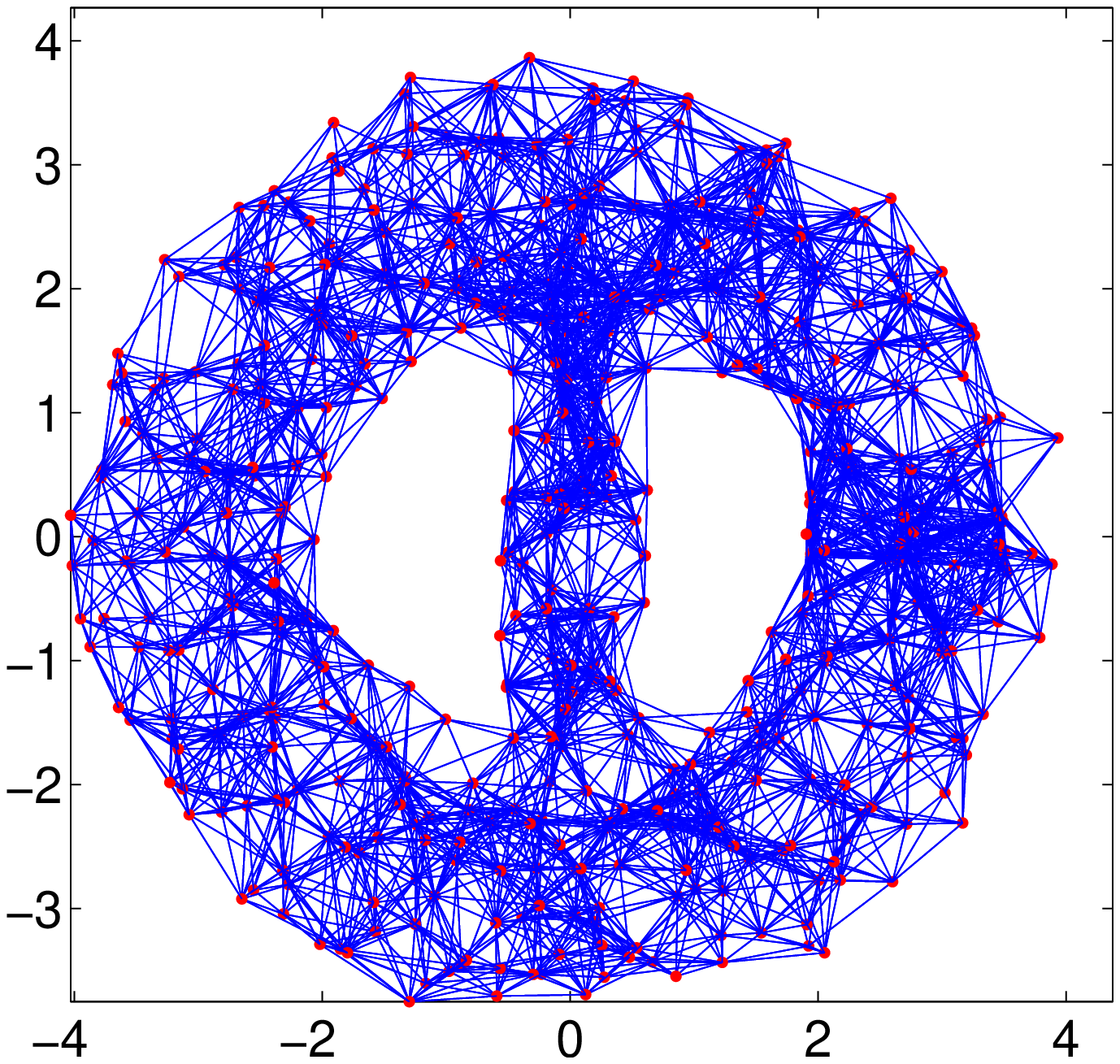} &
\includegraphics[width=0.40\columnwidth ] {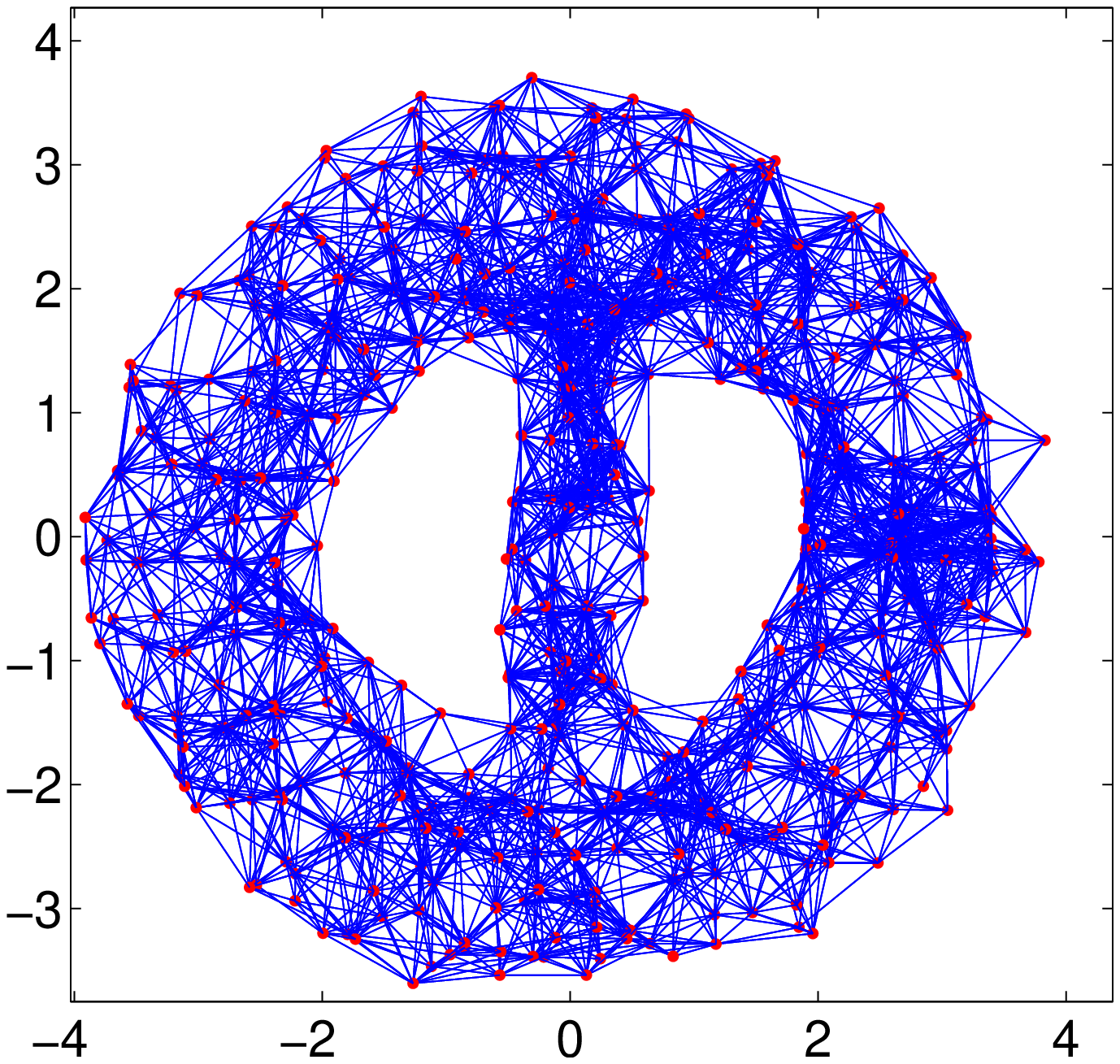} \\
40\% &
\includegraphics[width=0.40\columnwidth ] {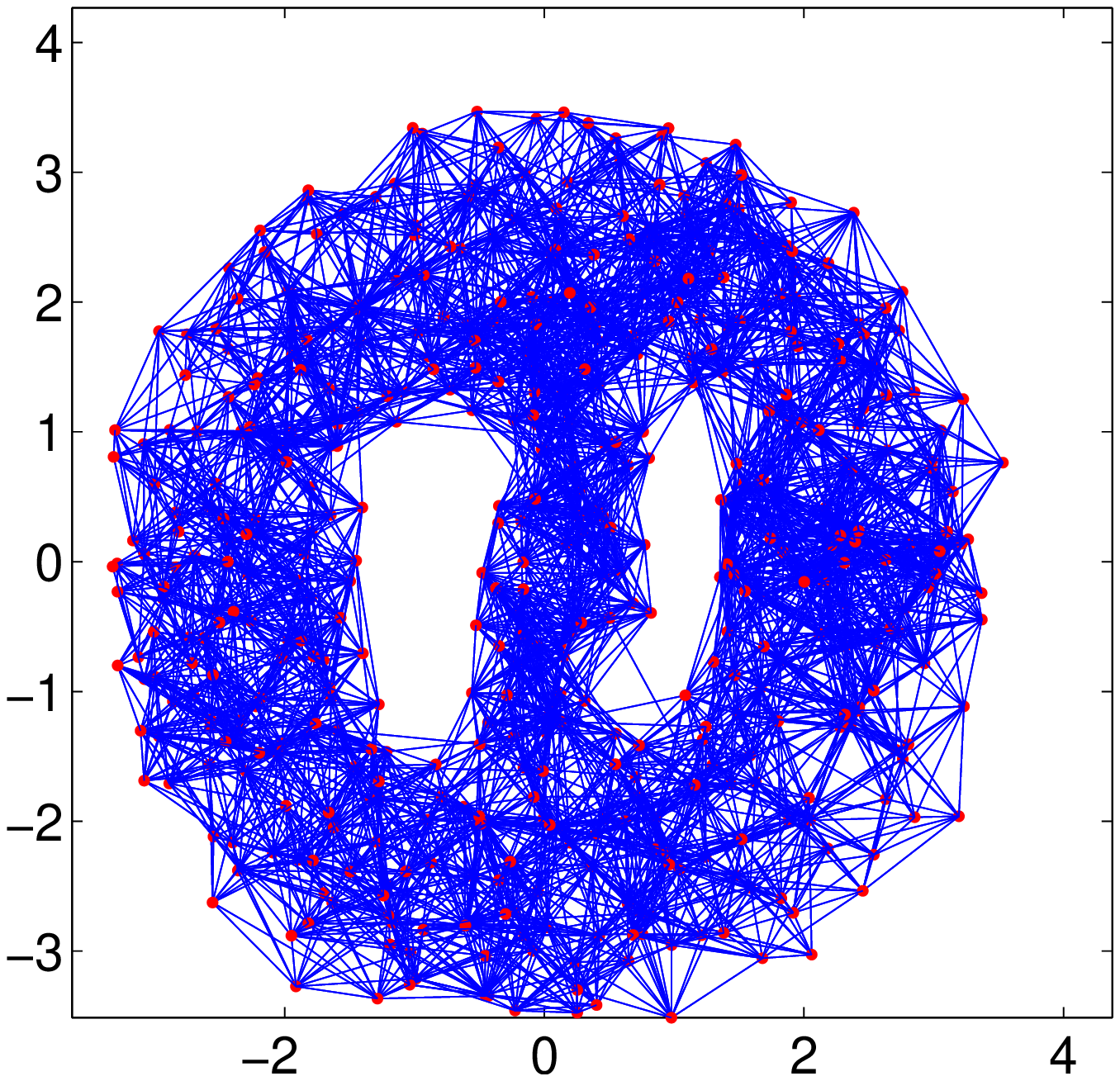}&
\includegraphics[width=0.40\columnwidth ] {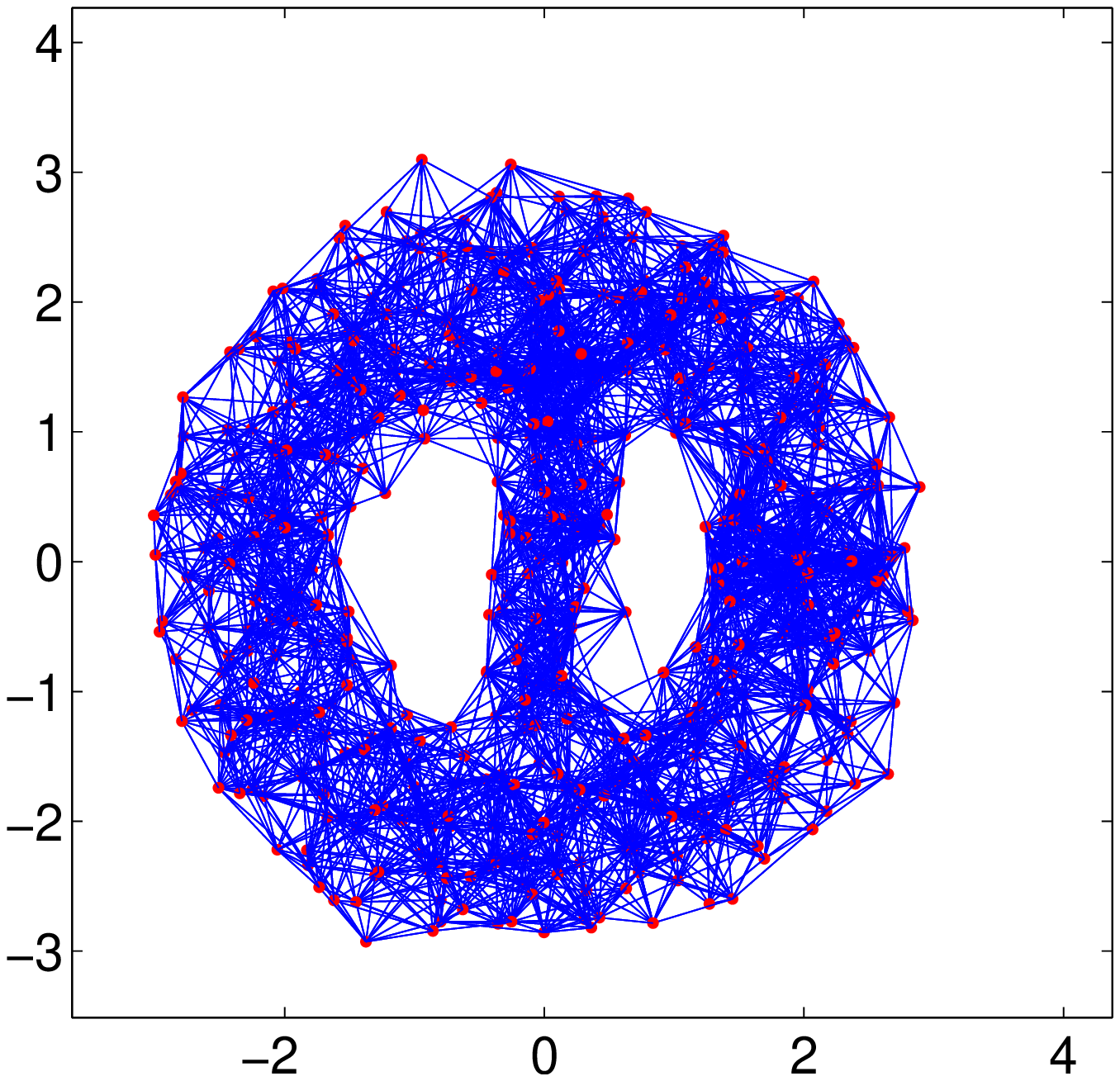}  \\
50\% &
\includegraphics[width=0.40\columnwidth ] {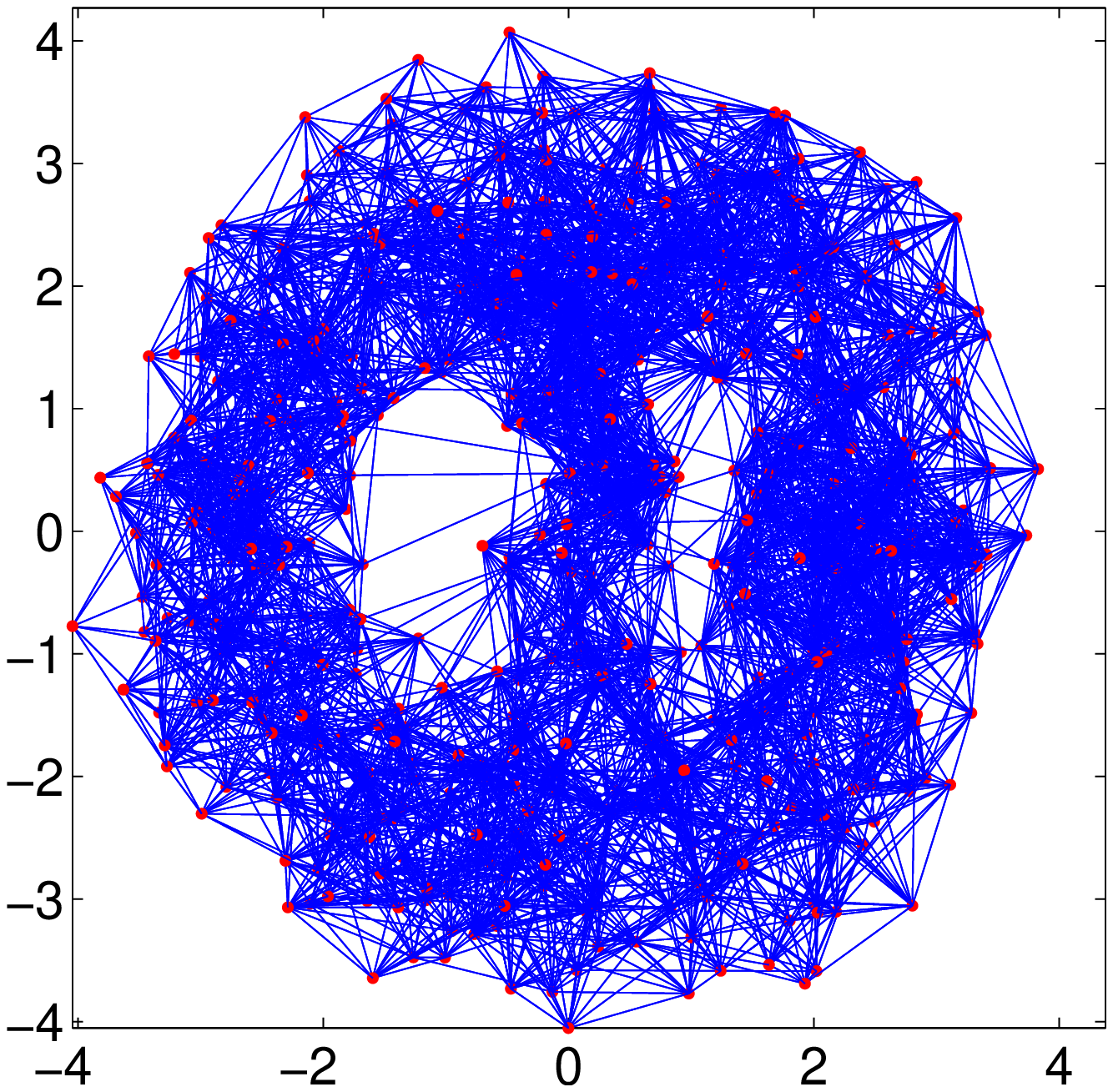} &
\includegraphics[width=0.40\columnwidth ] {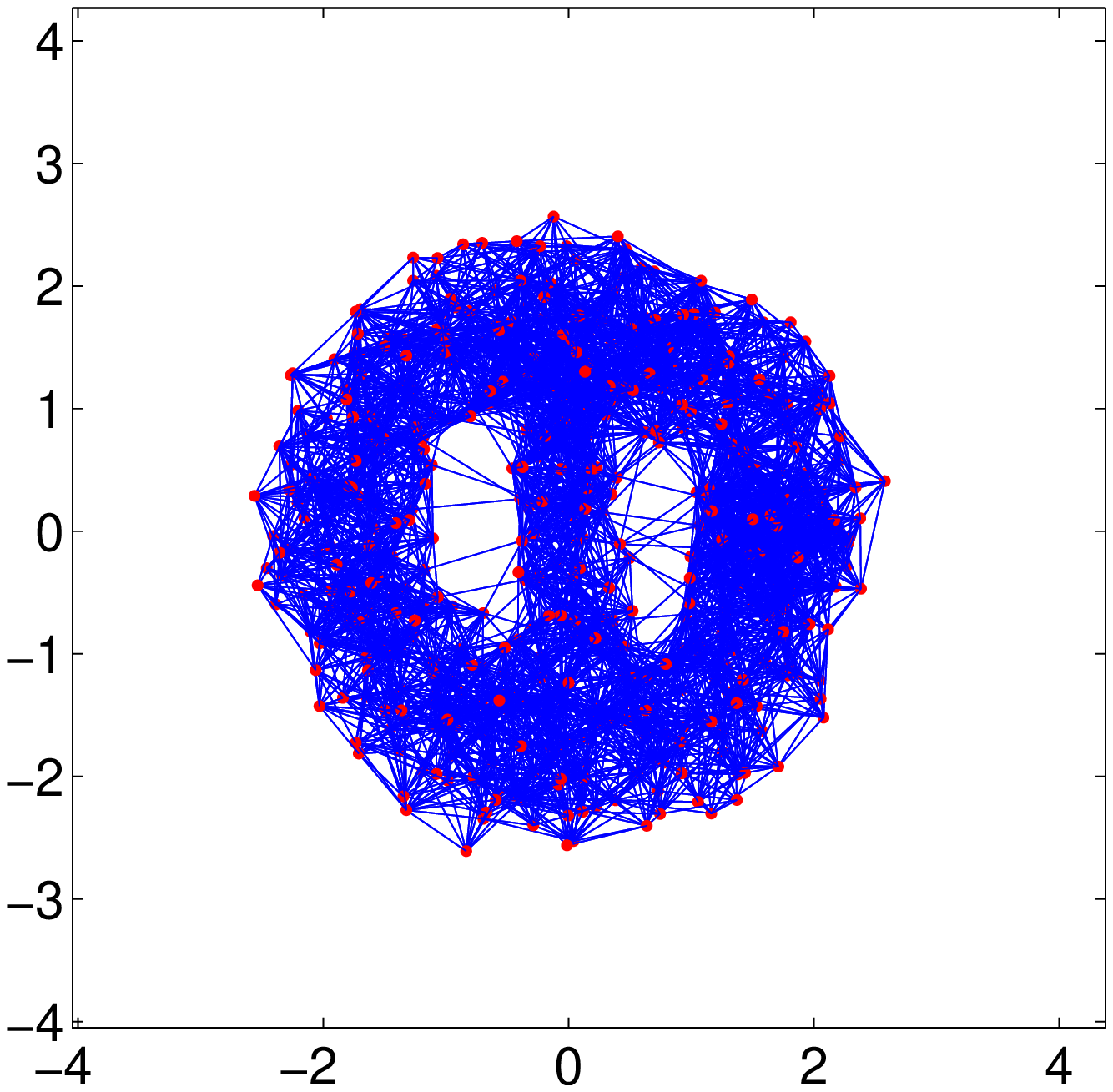} \\
\end{tabular}
\caption{Reconstructions of the BRIDGE-DONUT graph with $n=500$ nodes, sensing radius $\rho=0.92$ and $\eta = 0\%, 10\%, 20\%, 30\%, 40\%, 50\%$.}
\label{fig:RECS_Donut}
\end{figure}

Table \ref{tab:times} shows the running time of the various steps of the 3D-ASAP algorithm corresponding to our not particularly optimized MATLAB implementation. Our experimental platform was a PC machine equipped with an Intel(R) Core(TM)2 Duo CPU E8500 @ 3.16GHz $4$ GB RAM. Notice that all steps are amenable to a distributed implementation, thus a parallelized implementation would significantly reduce the running times.
Note the running time of 3D-ASAP is significantly larger than that of 3D-SP-ASAP and DISCO, due to the large number of patches (linear in the size of the network) that need to be localized. 3D-SP-ASAP addresses this issue, and reduces the running time from 474 to 108 seconds (for $\eta=0\%$), and from 1770 to 186 seconds (for $\eta=35\%$). Note that all steps of the algorithm scale linearly in the size of the network, except for the eigenvector computation, which is nearly-linear. We refer the reader to Section 7 of \cite{ASAP} for a complexity analysis of each step of 2D-ASAP, and remark that it is very similar to the complexity of 3D-ASAP.

\begin{table}[h]
\begin{minipage}[b]{0.98\linewidth}
\centering
\begin{tabular}{|l|l|l|l|l|}
\hline
Number of partitions   $k$      & \multicolumn{2}{c|}{$k=8$} &  $k=25$  \\
\hline
Noise level $\eta $ & $0\%$ & $35\%$ &  $35\%$  \\
\hline
  Spectral partitioning & 0.3 & 0.3 &  0.7 \\
  Finding WUL subgraphs &  48  & 81 &   89 \\
  Embedding FULL-SDP	 &  53 & (82) &  (89)\\
  Embedding SNL-SDP  	& (26) & 50  &   86 \\
  Step 1 ($\mathbb{Z}_2 \times \mbox{S0(3)}$ )  & 1 & 1  & 1 \\
  Step 2 (Least squares) & 6 & 8  & 9  \\
\hline
  Total & 108 & 140  & 186 \\
\hline
\end{tabular}
\caption{Running times (in seconds) of the 3D-SP-ASAP algorithm for the BRIDGE-DONUT graph with $n=500$ nodes, $\eta=0\%,35\%$, $deg=18,21$, and $k=8,25$ partitions. For $\eta=0 \%$ we embed the patches using FULL-SDP, while for $\eta= 35 \%$ we use SNL-SDP since the regularization term improves the localization.}
\label{tab:SP_times}
\end{minipage}

\begin{minipage}[b]{0.98\linewidth}
\centering
\begin{tabular}{|l|l|l|}
\hline
   & $\eta=0\%$ & $\eta = 35\% $ \\
\hline
  Break $G$ into patches & 59 & 90  \\
  Finding WUL subgraphs & 210 & 233 \\
  Embedding FULL-SDP &  154 & (252) \\
  Embedding SNL-SDP  & (966) & 1368 \\
  Denoising patches &  30 & 39 \\
  Step 1 ($\mathbb{Z}_2 \times \mbox{S0(3)}$ )  & 15 & 19    \\
  Step 2 (Least squares) & 6 & 21 \\
\hline
  3D-ASAP & 474 & 1770  \\
\hline
\hline
DISCO & 196 & 197 \\
\hline
3D-SP-ASAP & 108 & 186\\
\hline
\end{tabular}
\caption{Running times (in seconds) of the 3D-ASAP algorithm for the BRIDGE-DONUT graph with $n=500$ nodes, $\eta=0\%,35\%$, $deg=18,21$, $N=533,541$ patches, and average patch sizes $17.8, 20.3$.
For $\eta=0 \%$ we embed the patches using FULL-SDP, while for $\eta= 35 \%$ we use SNL-SDP since the regularization term improves the localization. For $\eta= 35 \%$, during the least squares step we also use the scaling heuristic, and run the gradient descent algorithm several times, hence the increase in the running time from 6 to 21 seconds.}
\label{tab:times}
\end{minipage}
\end{table}

\section{Summary and discussion} \label{conclusion}

In this paper, we introduced 3D-ASAP (As-Synchronized-As-Possible), a novel divide and conquer, non-incremental non-iterative anchor-free algorithm for solving (ab-initio) the molecule problem. In addition, we also proposed  3D-SP-ASAP, a faster version of 3D-ASAP, which uses a spectral partitioning algorithm as a preprocessing step for dividing the initial graph into smaller subgraphs. Our extensive numerical simulations show that 3D-ASAP and 3D-SP-ASAP are very robust to high levels of noise in the measured distances and to sparse connectivity in the measurement graph.

We build on the approach used in 2D-ASAP to accommodate for the additional challenges posed by rigidity theory  in $\mathbb{R}^3$ as opposed to $\mathbb{R}^2$. In particular, we extract patches that are not only globally rigid, but also weakly uniquely localizable, a notion that is based on the recent unique localizability of So and Ye \cite{SY05}. In addition, we also increase the robustness to noise of the  algorithm by using a median-based denoising algorithm in the preprocessing step, by combining into one step the methods for computing the reflections and rotations, thus doing synchronization over O(3)= $\mathbb{Z}_2 \times$ SO(3) rather than individually over $\mathbb{Z}_2$ followed by SO(3), and finally by incorporating a scaling correction in the final step where we  compute the translations of each patch by solving an overdetermined linear system by least squares. Another feature of 3D-ASAP is being able to incorporate readily available structural information on various parts of the network.

Furthermore, in terms of robustness to noise, 3D-ASAP compares favorably to some of the current state-of-the-art graph localization algorithms. The 3D-ASAP algorithm follows the same ``divide and conquer" philosophy that is behind our previous 2D-ASAP algorithm, and starts with local coordinate computations based only on the 1-hop neighborhood information (of a single node, or set of nodes whose coordinates are known a priori ), but unlike previous incremental methods, it synchronizes all such local information in a noise robust global optimization process using an efficient eigenvector computation. In the preprocessing step of doing the local computations, a median-based denoising algorithm improves the accuracy of the patch reconstructions, previously computed by the SDP localization algorithm.

Across almost all graphs that we have tested, 3D-ASAP constantly gives better results in terms of the averaged normalized error of the reconstruction. Except for the PACM graph, whose topology greatly favors the divide and conquer approach used by DISCO, 3D-ASAP returns reconstructions that are often significantly more accurate in the presence of large noise. Furthermore, for the case of noiseless distance measurements, the notion of weakly uniquely localizable graphs that we introduced, leads to reconstructions that are an order of magnitude more accurate than DISCO and SNL-SDP.


The geometric graph assumption, which comes up naturally in many problems of practical interest, is essential to the performance of the 3D-ASAP algorithm as it favors the existence of globally rigid or WUL patches of relatively large size. When the geometric graph assumption does not hold, the 1-hop neighborhood of a node may be extremely sparse, and thus breaking up such a sparse star graph leads to many small patches (i.e., most of them may contain only a few nodes), with only a few of them having a large enough pairwise intersection. Since small patches lead to small patch intersections, it would therefore be difficult for 3D-ASAP to align patches correctly and compute a robust final solution.
Note that in the case of random Erd\H{o}s-R\'{e}nyi graphs we expect the SDP methods, or even the low-rank matrix completion approaches, to work well. In other words, while these methods rely on randomness, our 3D-ASAP algorithm   benefits from structure the most.


For the molecule problem, while 3D-ASAP can benefit from any existing molecular fragments, there is still information that it does not take advantage of, and which can further improve the performance of the algorithm. The most important information missing from our 3D-ASAP formulation are the residual dipolar couplings (RDC) measurements that give angle information ($cos^2(\theta)$) with respect to a global orientation. Another possible approach is to consider an energy based formulation that captures the interaction between pairs of atoms (e.g. Lennard-Jones potential), and use this information to better localize the patches.


Another information one may use to further increase the robustness to noise is the distinction between the ``good" and ``noisy" edges. There are two parts of the algorithm that can benefit from such information. First, in the preprocessing step for localizing the patches one may enforce the ``good" distances as hard constraints in the SDP formulation. Second, in the step that synchronizes the translations using least squares, one may choose to give more weight to equations involving the ``good" edges, keeping in mind however that such equations are not noise free since the direction of an edge may be noisy as a result of steps 1 and 2, even if the distances are accurate. However, note that 3D-ASAP does use the ``good" edges as hard constraints in the gradient descent refinement at the end of step 3.

One other possible future direction is combining the reflection, rotation and translation steps into a single step, thus doing synchronization over the Euclidean group. Our current approach in 3D-ASAP takes one step in this direction, and combines the reflection and rotations steps by doing synchronization over the Orthogonal group O(3). However, incorporating the translations step imposes additional challenges due to the non-compactness of the group Euc(3), rendering the eigenvector method no longer applicable directly. 

In general, there exist very few theoretical guaranties for graph localization algorithms, especially in the presence of noise. A natural extension of this paper is a theoretical analysis of ASAP, including performance guarantees in terms of robustness to noise for a variety of graph models and noise models.
However, a complete analysis of the noise propagation through the pipeline of Figure \ref{fig:pipeline} is out of our reach at the moment, and first calls for theoretical guarantees for the SDP formulations used for localizing the patches in the preprocessing step. Recent work in this direction is due to \cite{montanari_localization}, whose main result provides a theoretical characterization of the robustness properties of an SDP-based algorithm, for random geometric graphs and uniformly bounded adversarial noise.

While the experimental results returned by 3D-ASAP are encouraging when compared to other graph realization algorithms, we still believe that there is room for improvement, and expect that further combining the approaches of 3D-ASAP and DISCO would increase the robustness to noise even more. The disadvantage of 3D-ASAP is that it sometimes uses patches of very small size, which in the noisy scenario, can be poorly aligned with neighboring patches because of small, possibly inaccurate set of overlapping nodes. One of the advantages of DISCO is that it uses larger patches, which leads to larger overlappings and more robust alignments. At the same time, the number of patches that need to be localized by an SDP algorithm is small in the case of DISCO ($2^{h}$ where $h$ is the height of the tree in the graph decomposition), thus reducing the computational cost of the algorithm. However, DISCO does not take advantage, at a global level, of pairwise alignment information that may involve more than two patches, while the eigenvector synchronization algorithm of 3D-ASAP incorporates such local information in a globally consistent framework. Straddling the boundary between SDP computational feasibility and robustness to noise, as well as finding the ``right" method of dividing the initial problem are future research directions. For a given patch of size $k$, is it better to run an SDP localization algorithm on the whole graph, or to first split the graph into two or more subgraphs, localize each one and then merge the solutions to recover the initial whole patch? An analysis of this question, both from a computational point of view and with respect to robustness to noise, would reveal more insight into creating a hybrid algorithm that combines the best aspects of 3D-ASAP and DISCO. The 3D-SP-ASAP algorithm is one step in the direction, and was able to address some of these challenges.



\section*{Acknowledgments}
The authors would like to thank
Yinyu Ye for useful discussions on rigidity theory and for sharing the FULL-SDP and SNL-SDP code.
A. Singer and M. Cucuringu acknowledge support by Award Number R01GM090200 from the NIGMS and Award Number FA9550-09-1-0551 from AFOSR. A. Singer was partially supported by the Alfred P. Sloan Foundation.

\bibliographystyle{siam}
\bibliography{3D_ASAP_main}

\end{document}